\documentclass{scrartcl}

\usepackage{hyperref}
\usepackage{amsmath,amssymb}
\usepackage{amsthm}
\usepackage[utf8]{inputenc}
\usepackage{srcltx}
\usepackage{graphics,graphicx}
\usepackage{algorithm}
\usepackage{algpseudocode}
\usepackage{subfig}

\usepackage[toc,page]{appendix}

\newcommand{\diag}{\text{diag}}
\newcommand{\trace}{\text{Tr}\,}
\newcommand{\R}{\mathbb{R}}
\newcommand{\C}{\mathbb{C}}
\newcommand{\Id}{\text{Id}}

\newcommand{\vertiii}[1]{{\left\vert\kern-0.25ex\left\vert\kern-0.25ex\left\vert #1 
		\right\vert\kern-0.25ex\right\vert\kern-0.25ex\right\vert}}
\newcommand{\Bnorm}[1]{{\Vert #1 \Vert_{B}}}
\newcommand{\Binorm}[1]{{\Vert #1 \Vert_{B_i}}}
\newcommand{\Binormbig}[1]{{\Big\Vert #1 \Big\Vert_{B_i}}}

\let\saveqed\qed
\renewcommand\qed{%
	\ifmmode\displaymath@qed
	\else\saveqed
	\fi}

\numberwithin{equation}{section}
\newtheorem{theorem}{Theorem}[section]
\newtheorem{lemma}[theorem]{Lemma} 
\newtheorem{remark}[theorem]{Remark}
\newtheorem{definition}[theorem]{Definition}
\newtheorem{proposition}[theorem]{Proposition}
\newtheorem{corollary}[theorem]{Corollary}
\begin{document}
	\title{Blind Demixing and Deconvolution at Near-Optimal Rate
          \thanks{
            The results of this paper have been presented in part at the 
            International Workshop on Compressed Sensing Theory and
            its Applications to Radar, Sonar, and Remote Sensing
            (Cosera), Aachen, Germany 2016
            \cite{Stoeger:cosera16} and 21st International ITG
            Workshop on Smart Antenna 2017, Berlin, Germany \cite{Stoeger:wsa17}
          }}
      \author{Peter Jung\thanks{Communications and Information Theory, Technische Universit\"at Berlin, 10587 Berlin}, Felix Krahmer\thanks{Department of Mathematics, Technische Universit\"at M\"unchen, 85748 Garching/Munich, Germany}, Dominik St\"oger\footnotemark[3]}
	\maketitle

\begin{abstract}
  We consider simultaneous blind deconvolution of $r$ source signals
  from their noisy superposition, a problem also referred to \emph{blind
    demixing and deconvolution}. This signal processing problem occurs
  in the context of the Internet of Things where a massive number of
  sensors sporadically communicate only short messages over unknown
  channels. We show that robust recovery of message and channel
  vectors can be achieved via convex optimization when random linear
  encoding using i.i.d. complex Gaussian matrices is used
  at the devices and the number of required measurements at the
  receiver scales with the degrees of freedom of the overall
  estimation problem. Since the scaling is linear in $r$ our result
  significantly improves over recent works.
\end{abstract}

\section{Introduction}
Recent progress regarding recovery problems for low-complexity
structures in high-dimensional data have shown that a substantial
reduction in sampling and storage complexity can be achieved in many
relevant non--adaptive linear signal separation and estimation
problems, in particular in the case of randomized strategies.  This
includes the recovery of sparse and compressible vectors (often
referred to as {\em compressed sensing}) \cite{candes2006robust,
  donoho2006compressed}, low--rank matrices \cite{rechtfazel}, and
higher--order tensors from subsampled linear measurements
\cite{rauhut2017low}, as well as the compressive demixing of multiple
source signals \cite{mccoy2014sharp}. An important step in many of
such vector and matrix recovery problems is to establish computational
tractability in the sense of complexity theory; a common strategy to
achieve this is to show that, under appropriate assumptions on the
measurement map, the reconstruction problem can be recast as a
tractable convex program.

In practice, however, one faces additional difficulties. Namely, the
data acquisition process has to cope with uncalibrated measurement
devices depending on further unknown parameters. In many such
scenarios one can only sample the output of an unknown or partially
known linear system. In such cases the object/signal $s$ to
recover is coupled with the unknown or partially known environment $w$
in a multiplicative way giving rise to a \emph{bilinear inverse problem},
i.e., solve for $s$ and $w$ given a bilinear combination
$\mathcal{B}(w,s)$.  Relevant examples are when the effective sensing
matrix might be subject to uncertainties
\cite{balzano2007blind, Herman:generaldeviants,Chi:sensitivity,Gleichman2011}, or
signals might have been transmitted through individual channels whose
properties are not completely known \cite{wang1998blind}. Our current
understanding of these \emph{blind information retrieval} tasks is at
the very beginning and usually it forces one therefore to operate at
sub-optimal sensing rates, or else incur significant reconstruction
errors due to model mismatch. The situation is all the more
unsatisfactory, as such blind sampling problems are often much closer
to practical applications than the original linear models.
\subsection{Blind Deconvolution}
The prototypical bilinear mapping, practically relevant in many
applications, is the convolution 
\[
w\ast s:= (\sum_{j=1}^L {w}_j s_{k-j})_{k=1}^L.
\]
For technical reasons we will consider the circular convolution, where
the index difference $k-j$ is considered modulo $L$. The classical
convolution can be reduced to this setup by appropriate zero padding.
Then the corresponding inverse problem, that is, the problem of
recovering $s$ and $w$ from their convolution up to inherent
ambiguities, % (up to the inherent ambiguity of multiplying $s$ by a scalar $\lambda$ and $w$ by $\tfrac{1}{\lambda}$)
is known as \emph{blind deconvolution} \cite{Haykin1994}.
%Here, the
%goal is to infer from an observed convolution $w\ast s$ the unknown
%desired signal $s$ without knowing the parameter vector
%$w$. 
The precise role of $s$ and $w$ depends on the underlying application. In imaging, for example, the signal vector $s$ typically represents the image and $w$ is
an unknown blurring kernel \cite{Stockham1975a}. In communication
engineering, $w$ represents the channel parameters and the task is to demodulate and decode the signal information $s$ only having access to the {\em channel output}
$w\ast s$, and the important question is how much overhead is
required for coping with the unknown impulse response $w$ of the
communication channel \cite{Godard1980}.

 Obviously, without further
 constraining $s$ and $w$ the convolution $(s,w)\rightarrow w\ast s$ has many more degrees of freedom than measurements and is hence
 far from injective, exhibiting various kinds of ambiguities. The goal must then be to eliminate these ambiguities as much as possible by imposing structural constraints on the signal and the channel paramters. It should be noted that a scaling ambiguity will always remain, as any bilinear mapping $\mathcal{B}$ satisfies $\mathcal{B}(s,w)=\mathcal{B}(\lambda s,w/\lambda)$ for any $0\neq\lambda\in\mathbb{C}$ and can hence be injective only up to a multiplicative factor. Specific scenarios can give rise to additional ambiguities, as it has been investigated in \cite{Choudhary2014}.
 For more detailed discussions of ambiguities in the one-dimensional
 case such as
 shifts or reflections, see
 \cite{Choudhary2014a} and \cite{walk:asilomar16}.
 In any case,  additional constraints like sparsity and subspace priors, depending on the specific application,
 are necessary to make blind deconvolution feasible. It has been shown that
 sparsity in the canonical basis alone is not sufficient for these
 purposes \cite{Choudhary2015}, and for generic bases, the subspace dimensions and sparsity levels that yield injectivity have been exactly classified \cite{li2015unified, li2017identifiability, Kech2016}.
 
 Even when injectivity can be established, this does not directly yield a tractable reconstruction scheme. While a number of works have studied algorithms for recovery (see, e.g.,\cite{chan2000convergence, levin2009understandingdeconvolution, almeida2013blind}), the focus has mostly been on algorithmic performance rather than on recoverability guarantees.
 The search for algorithms allowing for guaranteed recovery has recently shown significant progress by taking a compressed sensing viewpoint, namely aiming to choose remaining degrees of freedom to reduce the degree of ill-posedness. 
 %Several algorithms are known to solve such problems under various
 %different settings, constraints and objectives. However, an important
 %breakthrough in establishing 
 The first near-optimal rigourous recovery guarantees in a randomized setting have been established in
 %has been achieved with the work of 
 \cite{ARR2012} %for the
 %circular convolution which can be defined using $L\times L$ discrete
 %Fourier matrix $F$ by:
 %\begin{equation}
 %  w\ast s:=F^{-1}\diag(Fw)Fs
 %\end{equation}
 with high probability under the assumption that both the signal and the channel parameters lie in subspaces of small dimension, and one of them is chosen at random.
 %
 % $w=F^{-1}B h$ and $s=F^{-1}C\overline{x}$ with
 %$B\in\mathbb{C}^{L\times K}$ and
 %$C\in\mathbb{C}^{L\times N}$. 
 The main idea was to exploit that any bilinear map $\mathcal{B}(w,s)$ can be represented as a linear map in the outer product $ws^T$ of the two input vectors (this approach is often referred to as lifting) and hence analyzed using methods from the theory of low rank matrix recovery.
 % The lifting principle has been
 %adopted to settle blind deconvolution as a rank-one matrix recovery
 %problem with rank--one measurements. Every bilinear map
 %$\mathcal{B}(w,s)$ of two vectors can be written as a linear map
 %acting on the outer product $ws^T$ of both vectors $w$ and
 %$s$. 
 More precisely, exploiting the fact that the (normalized, unitary)
 $L\times L$ discrete Fourier matrix $F$ diagonalizes the circular
 convolution to establish the representation
 \begin{equation}
   w\ast s:=\sqrt{L}\cdot F^{*}\diag(Fw)Fs,
 \end{equation}
 with $\diag(v)$ denoting the diagonal matrix with the entries of $v$ on its diagonal.

 Under the subspace model, where both the signal $s$ and the vector of
 channel parameters are assumed to lie in a known low-dimensional
 subspace and hence can be represented as $w=F^{*}B h$ and
 $s=F^{*}C\overline{x}/\sqrt{L}$, for given $B\in\mathbb{C}^{L\times K}$ and
 $C\in\mathbb{C}^{L\times N}$, this translates to
\begin{equation}
  y:=F(w\ast s)=\diag(Bh)C\overline{x}=:\mathcal{A}(hx^*), \label{eq:convF}
\end{equation}
where $\mathcal A$ is a linear map and $M^*$ denotes the adjoint of a matrix $M$, that is, its conjugate transpose.
%
%$B\in\mathbb{C}^{L\times K}$ and
%$C\in\mathbb{C}^{L\times N}$.  The lifting principle has been
%adopted to settle blind deconvolution as a rank-one matrix recovery
%problem with rank--one measurements. Every bilinear map
%$\mathcal{B}(w,s)$ of two vectors can be written as a linear map
%acting on the outer product $ws^T$ of both vectors $w$ and
%$s$. Therefore the Fourier transform of the convolution of two 
%vectors precoded by $B$ and $C$ can be written in terms of a linear
%map $\mathcal{A}:\mathbb{C}^{K\times N}\rightarrow\mathbb{C}^L$ as:
%\begin{equation}
%	y=F(w\ast s)=\diag(Bh)C\overline{x}=:\mathcal{A}(hx^*)
%\end{equation}
This formulation yields a low rank recovery problem, as of all
potential matrices giving rise to measurements $y$, the rank one
matrix $hx^*$ is the one of the lowest rank.  Even though recovering a
low rank matrix from linear measurements is known to be, in general,
NP-hard \cite{chistov1984complexity}, it has been shown that under
appropriate random measurement models, one can establish recovery
guarantees for tractable algorithms with high probability
\cite{candes2011tight, gross2011recovering}.  While the results in
these works require more randomness than what is available in the
convolution setup due to the structure imposed by \eqref{eq:convF} and
hence do not apply directly, Ahmed, Recht, and Romberg \cite{ARR2012}
derived recovery guarantees for blind deconvolution. Their result
assumes that (i) $C$ has independent standard~Gaussian entries and
that (ii) $B^*B=1$ and $B$ is incoherent in two ways, namely that
$\mu^2_{\max}:=\frac{L}{K}\max_{\ell}\|b_\ell\|^2_{\ell_2}$ and
$\mu^2_h=L\cdot\max_{1\leq \ell\leq L}|b_\ell^*h|^2$ are
sufficiently small ($b_\ell$ are the columns of $B^*$).
%Since deterministic conditions on the matrices $B$ and $C$
%which ensure injectivity on the set of rank--one matrices $hx^*$ are
%difficult to setup random strategies are helpful to cope with such
%difficulties. One has to distinguish whether both, $B$ and $C$, are
%random matrices or only one of them is random. Assuming both matrices,
%$B$ and $C$, to be simultaneous random quantities renders bilinear
%inverse problem quite close to rank--one measurements already
%investigated in the field low--rank matrix recovery. A new and from an
%application perspective relevant case arises when $B$ is fixed and $C$
%is random. For this case Ahmed, Recht and Romberg \cite{ARR2012}
%investigated blind deconvolution (of real vectors) when (i) $C$ is iid. Gaussian
%distributed, (ii) $B^*B=1$ and $B$ is incoherent such that
%$\mu^2_{\max}:=\frac{L}{K}\max_{l}\|\hat{b}_l\|^2_{\ell_2}$ is sufficiently
%small.
  Under these assumptions, they showed that the unknown real $K\times N$--matrix $hx^*$ can be recovered
with overwhelming probability by nuclear norm minimization, that is, via the semidefinite program
\begin{equation}
  \min \lVert X\rVert_*\quad\text{s.t.}\quad \mathcal{A}(X)=y.
\end{equation}
%given that (iii) $\max(\mu^2_{\max}K,\mu^2_h N)\lesssim L/\log^3(L)$
%where
%$\mu^2_h=L\cdot\max_{1\leq l\leq L}|\langle h,\hat{b}_l\rangle|^2$.
Here, $\lVert X\rVert_*$ denotes the nuclear norm of the matrix $X$, which is defined to be the sum of its singular values.

Although nuclear norm
minimization is computational tractable, the lifted representation drastically increases the
size of the signal to be recovered.  Consequently, the resulting algorithm will be too slow
for most practical applications. The theoretical analysis of nuclear norm minimization has, however, paved the way for more efficient algorithms with similar guarantees. Namely, the recent work
\cite{Li2016} demonstrates that a gradient-based algorithm with a
suitable initialization can be used without lifting and in the regime
$\mu^2_h \max(K,N)\lesssim L/\log^2(L)$ which comes with considerably
reduced complexity. 

Finally, typical channel impulse responses $h$ exhibit further structural 
properties such as sparsity, which should be used as well.
However, the challenging extension of these works to sparsity models seems to
be much more involved. The difficulty with such models is that the lifted representation is both sparse and of low rank, and no straightforward tractable convex relaxation is known. In particular, minimizing convex combinations of nuclear and
$\ell_1$-norm regularizers has been shown to yield provably suboptimal recovery performance
\cite{Oymak2015}. Research regarding alternative convex
surrogates as for example in \cite{Richard2014} is only in its beginnings.
For this reason, some recent approaches ignore the rank constraint, just aiming for sparsity, as investigated  for the $\ell_1$--approach (sparse lift) in \cite{Ling2015}
and for the mixed $\ell_1/\ell_2$-case in \cite{Flinth2016}.

On the other hand, the search for non-convex alternatives to overcome
this obstacle is an active area of research. In particular, local
convergence guarantees as well as global convergence guarantees for
peaky signals have been derived in \cite{lee2017blind} for the sparse
power factorization method, an alternating minimization approach
originally introduced in \cite{Lee2013}, for the context of
deconvolution. The near-optimal recovery guarantees build on some
property similar to the restricted isometry property, which has been
derived in \cite{riplikeproperties} (for both inputs lying in random
subspaces). The search for global recovery guarantees in the sparsity
model without peakiness assumptions, however, remains open.

%Meanwhile injectivity results \cite{Kech2016}
%and conditions on restricted (almost-) isometric properties
%\cite{riplikeproperties} exists, however the straightforward convex
%recovery approach using convex combinations of nuclear and
%$\ell_1$-norm regularizers will most-likely fail in this regime
%\cite{Oymak2015}. Nevertheless, for particular configurations of
%sparsity and dimension parameters the pure vector--norm approach,
%although being multiplicative in the complexity parameters, 
%can be superior to the low-rank penalty. This was investigated for
%example for the $\ell_1$--approach (sparse lift) in \cite{Ling2015}
%and for the mixed $\ell_1/\ell_2$-case in \cite{Flinth2016}.
%Summarizing, there is the need for either other new convex
%surrogates as for example in \cite{Richard2014}
%or the investigation of iterative non-convex algorithms, 
%see exemplary the work in \cite{Lee2013} and \cite{Lee:sampta}.

\subsection{Simultaneous Demixing and  Blind Deconvolution}
The extension of the model we shall consider here is blind
deconvolution and simultaneously demixing multiple source signals.
% at the almost--optimal rate. 
This setting is motivated by recent challenges in future wireless
multi--terminal communication scenarios for uncoordinated sporadic
communication \cite{Wunder2015:sparse5G,Jung2014}.  We consider the
prototypical case of $R$ transmitters each having an individual
information message encoded into the vector $x_i\in\mathbb{C}^{N_i}$
for $i=1,\ldots ,R$ using, for example, classical modulation alphabets
and error--correcting codes. In fact, such data could be independent
user data payloads or even correlated sensor readings on a common
source. For reasons of simplicity, we focus on the case of independent
data sources.  Each transmitter generates its transmit signal
$s_i=F^{*}C_i\overline{x}_i/\sqrt{L}\in\mathbb{C}^L$ by multiplying
(linearly encoding) its complex--valued (conjugated) message vector
$\overline{x}_i$ by an $L\times N_i$ matrix $F^*C_i/\sqrt{L}$ which
is then transmitted into the shared channel.  Note that, from the
perspective of communication engineering, this procedure has been
simplified to facilitate the analysis.  In a more advanced setting one
could consider a directly randomized mapping from bits to sequences in
$\mathbb{C}^L$.  Now consider a single receiver, for example a base
station. Each transmitter $i$ has its individual impulse response
$w_i$ describing the channel propagation conditions to this base
station.  For simplicity we consider a low--mobility scenario where,
for appropriate block length $L$, the channel is time--invariant and
can be modeled by a convolution of the transmit signal with a channel
impulse response $w_i$. Furthermore, with cyclic extensions and
zero-padding at the transmitter such a signal propagation can then be
modeled as a circular convolution.  To incorporate further structure
for the channel impulse response we write it as $w_i=F^*B_ih_i$
where $B_i\in\mathbb{C}^{L\times K_i}$.  A reasonable assumption for
our application is that the unknown coefficients $h_i$ are located on
the first samples since the path delays in the channel are usually
much shorter than the frame length $L$. In this case $F^*B_i$ is a
truncated identity, i.e., $B_i^*B_i=\Id$.

In practice, since the desired deployment scenario is uncoordinated
and sporadic, only a small fraction of size $r$ of $R$ devices are
online and transmitting data. We assume for this work that the
receiver is able to detect the activity pattern correctly (which can
be achieved through a separate control channel, see for example
\cite{kueng:itw16} for a certain approach). One can even detect
activity  simultaneously with data.  However, algorithms for
blind deconvolution and demixing are usually quite complex from
practical and computational aspects and it is desired to reduce the
problem size as much as possible already from the beginning.  This
means, restricted and resorting to the active set, the receiver
observes the noisy superposition
\begin{equation}
  y
  = \sum_{i=1}^{r} F(w_i  \ast s_i) + e
  = \sum_{i=1}^{r} \diag(B_i h_i)C_i \overline{x}_i+e
  = \sum_{i=1}^{r} \mathcal{A}_i(h_ix_i^*)+e
  \label{eq:intro:sysmodel}
\end{equation}
of $r$ signal contributions where the vector $e\in\mathbb{C}^L$
denotes additive noise.

The conventional approach is (i) to
design the matrices $C_i$ is such a way that resources are used
exclusively by $\mathcal{O}(R)$ devices which requires considerable processing,
resource planning and allocation algorithms and (ii) estimate the
channel from pilot signals during a calibration phase prior to data
transmission.  However, in an increasing number of new applications
the typical data traffic consists only of short messages (status
updates or sensor data) yielding a sporadic traffic type and then the
overall communication in a network is then considerable dominated by
control data.

In \cite{lingstrohmer} it has therefore been proposed to consider the scenario of
\emph{simultaneous blind deconvolution and demixing} of multiple
signals from its superposition $y$, which we will also study in this paper. Demixing by convex programming
methods has been intensively investigated in the fields of ``sine and
spikes'' (and pairs of bases) decompositions, see \cite{Donoho2001}
and \cite{Amelunxen:2013}, and in the field of sparse and low--rank
decomposition, see, e.g., the work \cite{Chandrasekaran09}.  More
generally, as for example outlined in \cite{McCoy:cdemix13} and
\cite{Wright:cPCA13}, a convex approach consists of minimizing the
sum of the individual regularizers over all signal formations which
are conform with the model and consistent with the observations.  To
this end, assuming a priori that $\lVert e\rVert_{\ell_2}\leq\tau$, we
consider the convex optimization problem
\begin{equation}
  \label{eq:nucbpdn}
  \min \sum_{i=1}^r\lVert X_i\rVert_*\quad\text{s.t.}
  \quad \lVert\sum_{i=1}^r\mathcal{A}_i(X_i)-y\rVert_{\ell_2}\leq\tau.
\end{equation}
According to \cite{McCoy:cdemix13}, reliable convex demixing is
possible whenever (i) the signal contributions are incoherent to each
other and (ii) the number of observations is sufficiently above the sum
of effective dimensions of the descent cones of the individual
regularizers at the unknown ground truth. Since the rank-one matrix
$X_i=h_ix^*_i$ has effective dimension $K_i+N_i$ this amounts to
$\mathcal{O}(r(K+N))$ observations, where $K=\max_i(K_i)$ and
$N=\max_i(N_i)$. First results and guarantees, based on the
incoherence between the mappings $\mathcal{A}_i$ which explicitly
occur in blind deconvolution \eqref{eq:intro:sysmodel} with
random $C_i$'s  are worked out in \cite{lingstrohmer}.  
The result in this paper states that if (up to logarithmic orders)
$L=\mathcal{O}(r^2\max(K,N))$ the minimizer
$(\hat{X}_1,\dots,\hat{X}_r)$ of the program \eqref{eq:nucbpdn}
satisfies with high probability that
\begin{equation}
  \label{eq:LS15:estimationerror}
  \sum_{i=1}^r\Vert \hat{X}_i  - X_i^0 \Vert^2_{F}
  \lesssim  r^2\cdot\max  \left\{ K; N   \right\}  \tau^2
\end{equation}
Hence, for $\tau=0$ the ground truth $(\hat{X}^0_1,\dots,\hat{X}^0_r)$
is recovered exactly.  However, the embedding dimension does not quite
match the effective dimension, which would suggest a linear dependence
on $r$.  Ling and Strohmer suggested that this mismatch is a proof
artifact, observing numerically that linear dependence on $r$. In this
paper, we will analytically justify these observations.
In the special case of  partial (low-frequency) Fourier matrices $B_i$ mentioned above,
our main result, Theorem \ref{theorem:mainwithnoise}, reads as follows.
\begin{theorem}\label{theorem:mainwithnoise:dft}
  Let $\omega \ge 1$ and set $\mu^2_h=L\max_{i,l}|b^*_{i,\ell}h_i|^2$.
  Assume $\lVert e\rVert_{\ell_2}\leq\tau$ and that 
  \begin{equation}\label{Lbound}
    L \ge C_{\omega} r \left( K \log K + N \mu^2_h \right) \log^3 L,
  \end{equation}
  where $C_{\omega}$ is a universal constant only depending on
  $\omega$. Then with probability at least
  $1- \mathcal{O} \left( L^{- \omega} \right) $ the minimizer
  $ \hat{X}$ of the recovery program \eqref{eq:nucbpdn} satisfies
  \begin{equation}\label{ineq:estimationerror}
    \sum_{i=1}^r\Vert \hat{X}_i  - X_i^0 \Vert^2_{F} \lesssim  r\cdot\max  \left\{ 1; \frac{r KN}{L}   \right\}  \log(L)    \ \tau^2.
  \end{equation}
\end{theorem}
%$\mu^2_{1}=\dots=\mu^2_{r}=1$ , 

Shortly
before completion of this manuscript Ling and Strohmer presented
recovery guarantees for (considerably more efficient) nonconvex
gradient (Wirtinger) based methods \cite{Ling:2017}, again with
quadratic scaling in $r$.  Again they conjecture linear dependence, as
observed in their numerical experiments. We also include some
numerical experiments in Section \ref{sec:numerics} at the end that
illustrate the linear dependence. We expect that our paper at hand
will pave the way to an optimized parameter dependence also for more efficient algorithms.

\section{General Framework and Main Result}\label{sec:generalframework}
\subsection{Notation}
Before we describe the mathematical model we introduce some basic notation. For complex numbers $z \in \mathbb{C}$ we 
denote its conjugate by $\bar{z}$ and write $\text{Re } z$ and
$\text{Im } z$ for the real and imaginary part. Similarly, for a
vector $w= \left( w \lbrack 1 \rbrack , \ldots, w  \lbrack n \rbrack \right)\in \C^n $ we use the
notation $\text{Re } w = \left( \text{Re } w \lbrack 1 \rbrack, \ldots, \text{Re } w \lbrack n \rbrack
\right) $ and $\text{Im } w = \left( \text{Im } w  \lbrack 1 \rbrack, \ldots, \text{Im
  } w  \lbrack n \rbrack \right)  $.  For a matrix $A \in \mathbb{C}^{d_1 \times d_2} $
we will denote its adjoint by $ A^* $ and 
(for $d_1=d_2$) its trace by $ \trace \left(A\right) $. For matrices 
$A,B \in \mathbb{C}^{d_1 \times d_2}$ we will define the inner product
by $\langle A,B \rangle_F = \trace \left( AB^*\right) $. The
Frobenius norm of $A$ is $\Vert A\Vert^2_F=\langle A,A\rangle_F$
and $\Vert A \Vert_{2 \rightarrow 2}$ denotes its operator norm. If $\mathcal{B}$ is a linear operator mapping matrices to vectors or matrices, we will denote its operator norm by $ \Vert \cdot \Vert_{F \rightarrow 2} $ or $ \Vert \cdot \Vert_{F \rightarrow F} $, respectively. 
The nuclear norm of the matrix $A$, which is defined as the sum of its
singular values, will be denoted by $\Vert A \Vert_{\ast}$. Note
that the notation for $\Vert \cdot \Vert_{\ast}$,
$\lVert\cdot\rVert_F$ and $\langle \cdot, \cdot \rangle_F$ will be
used later in a more generalized setting, as will be pointed out in
the next section.
The matrix $ \Id_d$ will denote the identity matrix in 
$\mathbb{C}^{d\times d} $. If no confusion can arise, we will suppress $d$ and write $ \Id $ instead of $ \Id_d $. For a vector $v \in \mathbb{C}^d$  $ \text{diag} \left(v\right) $ denotes the matrix whose diagonal entries are given by $v$. Furthermore, $\Vert v \Vert_{\ell_2}$ denotes the
$\ell_2$-norm of this vector, i.e. $\Vert v \Vert^2_{\ell_2}= \langle v,v\rangle=\trace \left( v v^*\right)$.\\

By $ \mathbb{P} \left(E\right) $ we will denote the probability of an event $E$. For any $ N \in \mathbb{N} $ we will denote the set $ \left\{ 1, \ldots, N \right\} $ by $ \lbrack N \rbrack $. For a set $S$ we will denote its cardinality by $ \vert S \vert $. The notation $ \log \left( \cdot \right) $ will refer to the logarithm of base $2$. Furthermore, during the whole manuscript $C$ will denote positive numerical constants, which are independent of all other variables which appear in the text and whose value may change from line to line. Similarly, $C_{\omega} $ will denote universal numerical constants, which only depend on $\omega $. We will write $ a \lesssim b $, if $ a \le Cb $ and $ a \lesssim_{\omega} b $, if $ a \le C_{\omega} b $. We will write $ a \sim b $, if we have $ a \lesssim b $ as well as $ b \lesssim a $.

\subsection{The General Model}
In this paper we will work with a more general model, as also studied in \cite{lingstrohmer}, which includes the demixing-deconvolution scenario given above as special case. Assume that the vector $ y \in \mathbb{C}^L $
of $L$ noisy measurements corresponding to inputs $ \left\{ h_i \right\}^r_{i=1} $, $h_i \in \mathbb{C}^{K_i}$ and $ \left\{ x_i \right\}^r_{i=1} $, $x_i \in \mathbb{C}^{N_i} $, is given by
\begin{equation}\label{equ:ybasicmodel}
  y =  \sum_{i=1}^{r} \diag \left( B_i h_i \right) C_i \overline{x}_i+e. 
\end{equation}
where $e$ is additive noise, the matrices $B_i \in \mathbb{C}^{L \times K_i} $ satisfy $ B^*_i B_i = \Id_{K_i} $ for all $ i \in \lbrack r \rbrack $, and
all the entries of the random matrices $C_i \in \mathbb{C}^{L \times N_i}$ are independent and follow a standard circular-symmetric complex normal distribution $ \mathcal{CN} \left(0,1\right) $ (see Appendix \ref{sec:complexgaussian} for more details).
The vectors $h_i$ are assumed to be normalized, $ \Vert h_i \Vert_{\ell_2}=1 $, whereas the norms of
$x_i $ are arbitrary. (This is not restrictive as there is an inherent scaling ambiguity.)   %The matrices
%$C_i \in \mathbb{C}^{L \times N_i}$ will be independent random
%matrices. More precisely, the entries of $C_i$
%will be independent random variables following
%the circular-symmetric complex normal distribution $\mathcal{CN}(0,1)$.
%For more details concerning this probability distribution see Appendix
%\ref{sec:complexgaussian}. 
%distributed random variables with
%probability distribution
%$ \mathcal{N} \left(0, \frac{1}{2}\right) + i \mathcal{N} \left(0,
%  \frac{1}{2}\right) $.
%By $ \mathcal{N} \left(a, b \right) $ we denote the distribution of a
%standard Gaussian variable with expectation $a \in \mathbb{R}$ and
%variance $b \in \lbrack 0, +\infty ) $ 
Furthermore, we set
\begin{align*}
  K= \underset{i \in \lbrack r \rbrack}{\max}\ K_i \quad\text{and}\quad
  N= \underset{i \in \lbrack r \rbrack}{\max}\ N_i.
\end{align*}
Let us denote by $ b_{i, \ell} $ the $\ell$th column of $ B^*_i $ and
by $c_{i,\ell} $ the $ \ell $th column of $C_i $. Then, the $\ell$th entry of $y$ is given by
\begin{equation*}
y \lbrack \ell \rbrack= \sum_{i=1}^{r} b^*_{i,\ell } h_i x^*_i c_{i, \ell} + e \lbrack \ell \rbrack .
\end{equation*} 
We observe that the overall vector $y$ only depends on the outer
products $h_i x^*_i $. Thus, we may proceed by considering a lifted representation (see, e.g., \cite{bodmannliftingtrick}). Defining for each $i \in \lbrack r \rbrack$ the operator
$\mathcal{A}_i:\mathbb{C}^{K_i \times N_i} \longrightarrow \mathbb{C}^L$ via
\begin{equation*}
  \mathcal{A}_i \left( Z \right):= \left(  b^*_{i,\ell} Z c_{i, \ell}  \right)^{L}_{\ell=1}
  \label{eq:def:Ai}
\end{equation*}
we obtain that %for a given argument $Z \in \mathbb{C}^{K_i \times N_i}$. This yields
\begin{equation*}
  y = \sum_{i=1}^{r} \mathcal{A}_i \left( h_i x^*_i  \right) + e .
\end{equation*}
In the following we will use the decomposition $x_i= \sigma_i m_i $ where  $ \sigma_i \ge 0 $ and some $m_i \in \mathbb{C}^{N_i} $ such that $ \Vert m_i \Vert_{\ell_2} =1 $. (If $x_i  =0$ we set $ \sigma_i =0 $ and choose $ m_i$ arbitrarily.) Thus, the signal to be recovered may be written as
\begin{equation*}
  X^0 =  \left( h_1 x^*_1, \ldots, h_r x^*_r  \right) = 
  \left( \sigma_1 h_1 m^*_1, \ldots, \sigma_r h_r m^*_r \right)=:\left(X_1, \ldots, X_r \right).
\end{equation*}
% where we have set $ X_i = h_i m^*_i $.  (hat nicht gestimmt, denn X_i
% enthält auch \sigma_i
% As mentioned in the introduction the case that only a few sensors are active is of particular interest for us. Note that the active sensors correspond in our general framework to those $i \in \lbrack R \rbrack$ such that $ \sigma_i \ne 0$. This set will be denoted by
% \begin{equation*}
%   \text{supp } X_0 = \left\{  i \in \lbrack R \rbrack: \sigma_i \ne 0  \right\}.
%\end{equation*}
%Throughout the text we will assume that $ \text{supp } X_0 $ has cardinality $\vert \text{supp } X_0 \vert =r $.
Define 
\begin{equation*}
  \mathcal{M}:=  \left\{ \left(Z_1, \ldots, Z_r\right): \ Z_i \in \mathbb{C}^{K_i \times N_i} \text{ for all } i \in \left[ r \right]  \right\}
\end{equation*}
and note that $\mathcal{M}$ is naturally equipped with the algebraic structure of a vector space, as it may be regarded as the product space of the vector spaces $ \mathbb{C}^{K_i \times N_i} $.  The linear operator $ \mathcal{A}: \mathcal{M} \rightarrow \mathbb{C}^L $ is defined by
\begin{equation*}
  \mathcal{A} \left( Z \right) :=  \sum_{i=1}^{r} \mathcal{A}_i \left( Z_i \right)
  \label{eq:def:Aall}
\end{equation*}
for $ Z= \left( Z_1, \ldots, Z_r \right) \in \mathcal{M}$. The linear
space $ \mathcal{M}$ will be endowed with a norm and an inner product defined by
\begin{equation*}
  \langle W, Z \rangle_{F} = \sum_{i=1}^{r} \langle W_i, Z_i \rangle_F
  \quad\text{and}\quad
  \big\Vert W \big\Vert^2_{F} = \langle W,W\rangle_F =\sum_{i=1}^{r} \Vert W_i \Vert^2_F.
\end{equation*}
for all $ W, Z \in \mathcal{M}$. The operator norms $ \Vert \cdot \Vert_{F \rightarrow 2} $ and $ \Vert \cdot \Vert_{F \rightarrow F}  $ of linear maps on $ \mathcal{M}$ are defined analogously to the matrix case. For the adjoint $\mathcal{A}^*$ of $\mathcal{A}$ with respect to the inner product on $ \mathcal{M}$  
it follows $ \mathcal{A}^* \left(y \right) = \left( \mathcal{A}_1^*
  \left(y\right), \ldots, \mathcal{A}^*_r \left( y \right)  \right) $
for all $y \in \mathbb{C}^L $. Note that the adjoint operations $\mathcal{A}^*_i(y)$ itself are given by
\begin{equation}\label{equ:Aitransposed}
\mathcal{A}^*_i \left(y\right) = \sum_{\ell=1}^{L} y \lbrack \ell \rbrack b_{i,\ell} c^*_{i,\ell} \quad \text{for all } y \in \mathbb{C}^L.
\end{equation}
We will also use the norm defined by $\Vert W \Vert_{\ast} = \sum_{i=1}^{r} \Vert W_i \Vert_{\ast}$. For reasons which will become clear in Section \ref{sectionsufficientconditions} we set
\newcommand{\sgn}{\text{sgn}}
\begin{equation*}
\sgn(X^0_i) := \begin{cases}
h_i m^*_i &  \sigma_i > 0\\
0 & \text{else} \end{cases}
\end{equation*} 
%\begin{equation*}
%\sgn(X_i) = \sgn(\sigma_i) h_i m^*_i
%\end{equation*} 
for $ i \in \lbrack r \rbrack $ (recall that $\sigma_i\geq 0$). This allows us to define
\begin{equation*}
\sgn(X^0)  := \left( \sgn(X^0_1) , \ldots, \sgn(X^0_r)  \right).
\end{equation*}

\subsection{Partition of Measurements and Incoherence Assumptions}\label{section:partitioncoherence}
%In the following we will need a notion of incoherence of the matrices
%$ \left\{ B_i \right\}_{i \in \lbrack r \rbrack} $. It will be given
%by the parameter defined by
As those of of \cite{ARR2012,lingstrohmer}, our results are based on
two notions of coherence. The first is captured by the coherence
parameter
\begin{equation}\label{def:mumax}
  \mu^2_{i} = \underset{  \ell \in \lbrack L \rbrack}{\max} \frac{L}{K_i} \Vert b_{i, \ell} \Vert^2_{\ell_2}  \quad \text{for } i \in \lbrack r \rbrack .
\end{equation}
%Note that due to $B^*_i B_i = \Id_{K_i} $ we find that $1 \le \mu^2_{\max} \le \frac{L}{K} $.
Note that $ B^*_i B_i = \Id \in \mathbb{C}^{K_i \times K_i} $ for all
$ i \in \lbrack r \rbrack $ implies that
$ 1 \le \mu^2_{i} \le \frac{L}{K_i} $. In the (important) case that
all matrices $B_i$ are partial (low-frequency) DFT matrices, which refers to the
special situation described in the introduction, we have minimal
coherence $ \mu^2_{i} =1 $. In order to simplify notation we introduce
the quantities
\begin{equation}\label{def:Kmu}
K_{i, \mu} := K_i \mu^2_{i}, \quad \quad K_{\mu}:= \underset{i \in \left[r\right]}{\max}~K_{i,\mu}.
\end{equation}
We observe that $K_i \le K_{i, \mu} \le L $.
Again, in the special case that the matrices $B_i$ are partial
(low-frequency) DFT matrices we obtain that $K_{i, \mu} = K_i$.\\

For the proof of our results we will use the Golfing Scheme
\cite{gross2011recovering}, see Section \ref{subsec:golfing}. This
requires a partition $ \left\{ \Gamma_p \right\}^P_{p=1} $ of the set
of the measurements $ \lbrack L \rbrack $ with associated measurement
operators $ \mathcal{A}^p $. The second coherence parameter will also
depend on this partition. In order to guarantee that the Golfing
Scheme is successful with high probability we will need that
$ T_{i,p} := \frac{L}{Q} \sum_{\ell \in \Gamma_p} b_{i,\ell}
b^*_{i,\ell} \approx \Id_{K_i} $, as it will become clear in Remark \ref{Remark:Golfingmodified}. Thus, we have to assure that the partition
$ \left\{ \Gamma_p \right\}^P_{p=1} $ is chosen such that for
$Q:= \frac{L}{P} $ and $ \nu > 0 $ small enough one has
\begin{equation}\label{partitionequation}
\underset{i \in \lbrack r \rbrack, \ p \in \lbrack P \rbrack }{\max}\ \Big\Vert \Id_{K_i} - T_{i,p} \Big\Vert_{2\rightarrow 2} \le \nu.
\end{equation}
Furthermore, we require that $ \vert \Gamma_p \vert $ is large enough for
all $ p \in \lbrack P \rbrack $, i.e., each operator $\mathcal{A}^p $ contains enough measurements, and also the partition consists of the right number of sets, that is, $P$ is bounded above and below.
More precisely, we require that the partition is $\omega$-admissible in the sense of the following definition.
\begin{definition}\label{definition:admissiblepartition}
	Let $ \omega \ge 1 $ and let $\left\{ \Gamma_p \right\}^P_{p=1}$ be a partition of $ \lbrack L \rbrack $. The set $\left\{ \Gamma_p \right\}^P_{p=1}$ is called $\omega$-admissible if the following three conditions are satisfied:
	\begin{enumerate}
		\item $\frac{1}{2} Q \le \vert \Gamma_p \vert \le \frac{3}{2} Q $ for all $ p \in \lbrack P \rbrack $, where $Q = \frac{L}{P} $.
		\item (\ref{partitionequation}) is fulfilled with $ \nu = \frac{1}{32} $.
		\item It holds that $  \log \left( 8 \tilde{\gamma} \sqrt{r} \right)  \ge  P \ge \frac{1}{2} \log \left( 8 \tilde{\gamma} \sqrt{r} \right) $, where
		\begin{equation*}
		\tilde{\gamma} = 2\sqrt{  \omega \max  \left\{ 1; \frac{r K_{\mu} N }{L}   \right\}  \log \left(L+rKN\right) }.
		\end{equation*}
	\end{enumerate}
\end{definition}
Here the parameter $ \omega $ is the same that appears in Theorem \ref{theorem:mainwithnoise:dft} and in Theorem \ref{theorem:mainwithnoise}.\\

 This definition gives rise to the question whether such a partition exists in general and how one can construct them. This has already been discussed in \cite[Section 2.3]{lingstrohmer} for several
important special cases of matrices
$B_i \in \mathbb{C}^{K_i \times N_i} $. In particular, it is proven
that in the special case that the $B_i$'s are partial (low-frequency) Fourier matrices
of the same size and if $L=PQ$ one may find a partition such that $\nu=0$. In
\cite{ARR2012}, the authors discussed the construction of such a
partition for $r=1$ and for a general matrix $B \in \mathbb{C}^{K \times N}
$ which satisfies $ B^* B = \Id_K $. However, such a partition can be constructed for all matrices $B_i \in \mathbb{C}^{K_i \times N_i} $ simultanously via the following lemma.
\begin{lemma}\label{partitionlemma}
Let $P \in \lbrack L \rbrack$ and $\nu \in \left(0,1\right) $ be fixed. Set $Q= \frac{L}{P} $. There is a universal constant $C>0$ such that if
\begin{equation}\label{Qminimalsize}
Q \ge C \frac{K_{\mu}}{\nu^2} \log \left( \max \left\{ r; P; K   \right\}  \right)
\end{equation}
then there is a partition $ \left\{ \Gamma_p \right\}^P_{p=1} $ of $ [L] $ such that (\ref{partitionequation}) is satisfied and $ \frac{1}{2} Q \le  \vert \Gamma_p \vert \le \frac{3}{2} Q $ holds for all $p \in \lbrack P \rbrack $.
\end{lemma}
A proof of this result is included in Appendix~\ref{constructionpartition}. As $P=\tfrac{L}{Q}$, this lemma implies the existence of an $\omega$-admissible partitions provided that 

\[L \gtrsim \sqrt{r} \log \left( 8 \tilde{\gamma} \sqrt{r} \right) \frac{K_{\mu}}{\nu^2} \log \left( \max \left\{ r; P; K   \right\}  \right) ,\]  with $\tilde \gamma$ as in Definition~\ref{definition:admissiblepartition}, which is a somewhat milder assumption than what is required in  our main theorem.

%
% As already pointed out, our proof for the main result Theorem \ref{theorem:mainwithnoise} will not work for every possible partition $ \left\{ \Gamma_p \right\}^P_{p=1}$. Instead, it will be designed for the following partitions, which we call $\omega$\textbf{-admissible}:
%\begin{definition}\label{definition:admissiblepartition}
%Let $ \omega \ge 1 $ and let $\left\{ \Gamma_p \right\}^P_{p=1}$ be a partition of $ \lbrack L \rbrack $. The set $\left\{ \Gamma_p \right\}^P_{p=1}$ is called $\omega$\textbf{-admissible} if the following three conditions are satisfied:
%\begin{enumerate}
%\item $\frac{1}{2} Q \le \vert \Gamma_p \vert \le \frac{3}{2} Q $ for all $ p \in \lbrack P \rbrack $, where $Q = \frac{L}{P} $.
%\item (\ref{partitionequation}) is fulfilled with $ \nu = \frac{1}{32} $.
%\item It holds that $  \log \left( 8 \tilde{\gamma} \sqrt{r} \right)  \ge  P \ge \frac{1}{2} \log \left( 8 \tilde{\gamma} \sqrt{r} \right) $, where
%\begin{equation*}
% \tilde{\gamma} = 2\sqrt{  \omega \max  \left\{ 1; \frac{r K_{\mu} N }{L}   \right\}  \log \left(L+rKN\right) }.
%\end{equation*}
%\end{enumerate}
%\end{definition}
%The parameter $ \omega $ corresponds to the same constant in Theorem \ref{theorem:mainwithnoise:dft} and in Theorem \ref{theorem:mainwithnoise}.\\

The second incoherence parameter will depend on the choice of such an $\omega$-admissible partition, measuring how aligned the input $h_i$ is with the basis vectors $b_{i,\ell}$ distorted by a family of linear maps corresponding to the different sets in the partition.

%We need to introduce a further coherence parameter as the number of needed measurements $ \lbrack L \rbrack $ will also depend on how much the
%vectors $h_i$ are aligned with $b_{i,\ell}$ for
%$ \ell \in \left[L\right] $ and $ i \in \left[ r \right] $.

 More precisely, for a fixed  $\omega$-admissible partition $  \left\{ \Gamma_p \right\}^P_{p=1} $ we define
\begin{equation}\label{definition:muh}
\mu^2_{h} := L  \max \left\{ \underset{ \ell  \in \lbrack L \rbrack, i \in \lbrack r \rbrack}{\max} \vert b^*_{i, \ell } h_i \vert^2   ,  \underset{ p \in \lbrack P \rbrack, \ell  \in \lbrack L \rbrack, i\in \lbrack r \rbrack }{\max}  \vert b^*_{i, \ell } S_{i,p}  h_i \vert^2  \right\},
\end{equation}
where we have set $ S_{i,p}= T^{-1}_{i,p} $. The proof in Section \ref{sec:proof} will yield the strongest result when $ \mu^2_h $ is small. Thus, we will choose for our proof a partition, which minimizes the quantity defined in (\ref{definition:muh}). This motivates the introduction of the following quantity.
\begin{equation}\label{definition:muhomega}
  \mu^2_{h,\omega} = L  \min_{\left\{ \Gamma_p \right\}^P_{p=1} \omega\text{-admissible}}  \max \left\{ \underset{ \ell  \in \lbrack L \rbrack, i \in \lbrack r \rbrack}{\max} \vert b^*_{i, \ell } h_i \vert^2   ,  \underset{ p \in \lbrack P \rbrack, \ell  \in \lbrack L \rbrack, i\in \lbrack r \rbrack }{\max}  \vert b^*_{i, \ell } S_{i,p}  h_i \vert^2  \right\}.
\end{equation}
\begin{lemma}
Let $ \left\{ \Gamma_p \right\}^P_{p=1} $ be a $\omega$-admissible partition of $ \left[L\right] $. Then $1\le \mu^2_{h} \le \left(\frac{32}{31}\right)^2 K_{\mu}. $
\end{lemma}
\begin{proof}
The lower bound follows immediately from the observation
\begin{equation*}
  \sum_{\ell=1}^{L} \Vert b^*_{i,\ell} h_i \Vert^2_{\ell_2} =   \sum_{\ell=1}^{L} h^*_i b_{i,\ell} b^*_{i,\ell} h_i    = \Vert h_i \Vert^2_{\ell_2} =1.
\end{equation*}
For the upper bound it is enough to observe that $ L\vert b^*_{i,\ell} h_i \vert^2 \le L \Vert b_{i,\ell} \Vert^2_{\ell_2} \Vert h_i \Vert^2_{\ell_2} \le K_{\mu} $ and similarly $ L\vert b^*_{i,\ell} S_{i,p} h_i \vert^2 \le L \Vert S_{i,p} \Vert^2_{2 \rightarrow 2} \Vert b_{i,\ell} \Vert^2_{\ell_2} \Vert h_i \Vert^2_{\ell_2}  $. The result follows from the observation $ \Vert S_{i,p} \Vert_{2 \rightarrow 2} \le \frac{32}{31} $, which is due to $ \Vert \Id - T_{i,p} \Vert_{2 \rightarrow 2} \le \frac{1}{32} $.
\end{proof}

\begin{remark}
As already pointed out in \cite[Remark 2.1]{lingstrohmer} the appearance of the second term in the definition of $\mu_h$ is due to the modified Golfing Scheme (cf. Remark \ref{Remark:Golfingmodified}). %(This additional later term does not appear in \cite{ARR2012} as their proof is slightly inaccurate.)
Note, however, that our definition of $ \mu^2_h $ is slightly different to the definition of $ \mu^2_h $ in \cite{lingstrohmer}. In our definition, the second term the maximum is over all $ \ell \in \lbrack L \rbrack$, whereas in \cite{lingstrohmer} the maximum is only over all $ \ell \in \Gamma_p$. The reason is that of a simpler presentation and a less technical argument; it is possible to obtain our result with $ \mu^2_h $ as defined in \cite{lingstrohmer} by a slightly more involved argument: One needs to replace the norm $ \Vert \cdot \Vert_B $, which will be introduced in Section \ref{localisometry}, by norms which depend on the individual partitions $ \Gamma_p$.
\end{remark}
One may ask whether the second term in the definition of $ \mu^2_h $ can be removed. By a closer look at the proof of Lemma \ref{partitionlemma} one infers that for fixed $P$, which satisfies the third condition in Definition \ref{definition:admissiblepartition},  a constant fraction of all partitions are $ \mu $-admissible. Thus, one might conjecture that there is at least one partitition such that the quantity $ \underset{ p \in \lbrack P \rbrack, \ell  \in \lbrack L \rbrack, i\in \lbrack r \rbrack }{\max}  \vert b^*_{i, \ell } S_{i,p}  h_i \vert^2 $ is small such that it can be neglected. We leave this problem for future work.

\subsection{Main Result}
\label{sec:mainresults}
Our main result establishes a recovery guarantee for the general measurement model \eqref{equ:ybasicmodel}. Reconstruction proceeds via nuclear norm minimization, the semidefinite program formulated in \eqref{eq:nucbpdn}.
%Let $y$ be given by equation (\ref{equ:ybasicmodel}), which is our
%basic model. Then we choose the following convex program in order to
%recover all vectors $h_i \in \mathbb{C}^{K_i}$ and all vectors
%$ x_i \in \mathbb{C}^{N_i}$:
%\begin{equation}
%  \label{eq:nucbpdn:compact}
%  \begin{split}
%    \text{minimize } & \Vert X \Vert_{\ast} \\
%    \text{subject to } & \Vert \mathcal{A} \left( X \right) - y  \Vert_{\ell_2} \le \tau
%  \end{split}
%\end{equation}
%Under the assumptions of Section \ref{sec:generalframework} we can give the following recovery guarantee, which is our main result:
\begin{theorem}\label{theorem:mainwithnoise}
	Let $\omega \ge 1$ and let $y \in \mathbb{C}^L$ be given by (\ref{equ:ybasicmodel}) with $ \Vert e \Vert_{\ell_2} \le \tau $. Assume that 
	\begin{equation}\label{Lbound}
	L \ge C_{\omega} r \left( \underset{i \in \lbrack r \rbrack}{\max}~\left(K_i \mu^2_{i} \log \left( K_{i} \mu^2_i \right) \right) + N \mu^2_{h,\omega} \right) \log^3 L,
	\end{equation}
	where $C_{\omega}$ is a universal constant only depending on
        $\omega$. Then, with probability at least
        $1- \mathcal{O} \left( L^{- \omega} \right) $ the minimizer
        $ \hat{X}$ of the recovery program \eqref{eq:nucbpdn} satisfies
	\begin{equation}\label{ineq:estimationerror}
          \Vert \hat{X}  - X^0 \Vert_{F} \lesssim \tau  \sqrt{ r  \max  \left\{ 1; \underset{i \in \lbrack r \rbrack}{\max} \frac{r K_{i} \mu^2_i N}{L}   \right\}  \log L }  .
	\end{equation}
\end{theorem}
In the important special case of noiseless measurements, i.e., $ \tau= 0$, Theorem \ref{theorem:mainwithnoise} yields exact recovery with high probability, if $L$ satisfies
condition (\ref{Lbound}), i.e., $ X^0$ is the unique minimizer of the
semidefinite program (\ref{eq:nucbpdn}). As already mentioned
in the introduction our result significantly improves upon the result of
\cite{lingstrohmer} and exhibits optimal scaling in the degrees of freedom up to logarithmic factors. In the noisy case, the estimation error
(\ref{ineq:estimationerror}) is improved at least by a factor of
$\sqrt{r}$ (cf. \cite[Theorem 3.3]{lingstrohmer}).
\section{Preliminaries}
%In this section we will introduce the tools, which are needed for the proof in the next section.  They will mostly stem from convex geometry, random matrix theory and probability

\subsection{Concentration Inequalities}
In our proof we will have to estimate the spectral norm of a random matrix several times. Amongst others one tool we will apply is a  generalized version of the matrix Bernstein inequality, which may be seen as a corollorary from Theorem 4 in \cite{koltchinskii_generalizedbernstein}. It is based on so-called Orlicz norms $\Vert \cdot \Vert_{\psi_{\alpha}}$ , which may be regarded as a measure for the tail decay of random variables.
\begin{definition}
	Let $X$ be a complex-valued random variable. For $\alpha \ge 1$ we define the Orlicz norm $ \Vert \cdot \Vert_{\psi_{\alpha}} $ by
	\begin{equation*}
	\Vert X \Vert_{\psi_{\alpha}} = \inf \left\{ t>0 :  \mathbb{E} \left[ \exp \left(  \frac{\vert X \vert^{\alpha}}{t^{\alpha}} \right) \right] \le 2  \right\}.
	\end{equation*}
	
\end{definition}
It is straightforward to check that $ \Vert \cdot \Vert_{\psi_{\alpha}} $ is a norm (on the vector space of all complex-valued random variables $X$ such that $\Vert X \Vert_{\psi_{\alpha}} < + \infty $). Furthermore, as shown in  \cite{rutickii1961convex}, any two random variables $X,Y$ satisfy the Hoelder inequality
\begin{equation}\label{ineq:Hoelderinequality}
\Vert XY \Vert_{\psi_1} \le \Vert X \Vert_{\psi_2} \Vert Y \Vert_{\psi_2}.
\end{equation}
% If $\Vert X \Vert_{\psi_1} < + \infty  $, we will say the random variable is subexponential and if $\Vert X \Vert_{\psi_2} < + \infty $, we will say that the random variable is subgaussian. The lecture notes \cite{vershyninlecturenotes} contain a treatment of the most important properties of the norms $ \Vert \cdot \Vert_{\psi_1}$ and $\Vert \cdot \Vert_{\psi_2} $. 
%\begin{lemma}\label{orliczHoelder}[H\"older inequality]
%	Let $p,q \ge 1$ such that $ \frac{1}{p}+ \frac{1}{q}=1 $. Then for all complex-valued random variables $X,Y $ we have
%	\begin{equation*}
%	\Vert X Y \Vert_{\psi_1} \le \Vert X \Vert_{\psi_p} \Vert Y \Vert_{\psi_q} .
%	\end{equation*}
%\end{lemma}
If $ \Vert X \Vert_{\psi_1} <  \infty$ we will call a random variable sub-exponential. For sub-exponential random variables we state the Bernstein inequality in the version of \cite[Proposition 5.16]{vershyninlecturenotes}.
\begin{theorem}\label{bernsteinonedimensional}
	Let $X_1, \ldots, X_n$ be independent, mean zero sub-exponential random variables, i.e., $  \Vert X_i \Vert_{\psi_1} <  \infty $ for all $ i \in \lbrack r \rbrack $. Then with probability at least $ 1 - 2 \exp \left( -t \right) $
	\begin{equation*}
	\Big\vert \sum_{i=1}^{n} X_i \Big\vert \lesssim  \max \left\{  \sqrt{ t \sum_{i=1}^{n} \Vert X_i \Vert^2_{\psi_1} }; t \left( \underset{i \in \lbrack n \rbrack}{ \max}~\Vert X_i \Vert_{\psi_1} \right)   \right\}.
	\end{equation*}
\end{theorem}
There are powerful generalizations of the Bernstein inequality for the matrix-valued case. Those generalizations were discovered first in \cite{ahlswede2002strong} and were refined in \cite{tropp2012user}. We will state a this theorem for unbounded random matrices, which is reformulation of a version of Koltchinskii \cite[Theorem 4]{koltchinskii_generalizedbernstein}.
\begin{theorem}[Matrix Bernstein Inequality]\label{matrixbernstein}
	Let $\alpha \in  \lbrack  1, + \infty ) $ and let $X_1, X_2, \ldots, X_n \in \mathbb{C}^{d_1 \times d_2} $ be independent random matrices that satisfy $\mathbb{E} \left[X_i\right]=0 $ for all $ i \in \lbrack n \rbrack $. Set  $R= \underset{i \in \lbrack n \rbrack}{\max} \Big\Vert \Vert X_i \Vert_{2 \rightarrow 2}  \Big\Vert_{\psi_{\alpha}} $ and
	\begin{equation}\label{definitionsigma}
	\sigma^2= \max \left\{ \Big\Vert \sum_{i=1}^{n}   \mathbb{E} \left[ X_i X^*_i \right]  \Big\Vert_{2\rightarrow 2}  ;  \Big\Vert \sum_{i=1}^{n}   \mathbb{E} \left[  X^*_i X_i \right]  \Big\Vert_{2\rightarrow 2}    \right\}.
	\end{equation}
	Set $Z = \sum_{i=1}^{n} X_i $. Then with probability at least $1 - \exp \left( -t \right) $ 
	\begin{equation*}
	\big\Vert Z \big\Vert_{2 \rightarrow 2} \lesssim  \max \left\{   \sigma \sqrt{ t + \log \left(d_1+d_2\right) }   ; \  R \left( 
	\log \left(  1+ \frac{nR^2}{\sigma^2} \right) \right)^{\frac{1}{\alpha}}   \left(t + \log \left( d_1+d_2 \right) \right)   \right\}.
	\end{equation*}
\end{theorem}
	Indeed, when $d_1 = d_2$ and the matrices $X_1,X_2, \ldots, X_n$ are self-adjoint, Theorem \ref{matrixbernstein} can be deduced from \cite[Theorem 4]{koltchinskii_generalizedbernstein} (by choosing $ \psi_{\alpha} \left( u\right) = \exp \left( u^{\alpha} \right) -1  $ and, for example, $ \delta =1 $). In order to pass from self-adjoint matrices to general matrices $ X_i \in \mathbb{C}^{d_1 \times d_2} $ one may use self-adjoint dilations and argue as in \cite[Section 4.6.5]{tropplecture}.\\

%This theorem has been proven in \cite[Proposition 2]{matrixbernstein2} if the matrices $X_1, \ldots, X_n$ are identically distributed and symmetric matrices. However, a closer look at the proof reveals that the matrices do not need to be identically distributed. In order to pass from symmetric matrices to general matrices $ X_i \in \mathbb{R}^{K \times N} $ one may argue as in \cite[Section 4.6.5]{tropplecture} and use self-adjoint dilations.\\
%\begin{proof}
%	See \cite[Proposition 2]{matrixbernstein1} for a less general version of this theorem. (The proof carries over to the general version verbatim, but maybe it should be written up for the sake of completeness.)
%\end{proof}

The  matrix Bernstein inequality is a powerful tool, which works in many different situations. However, for some more specific examples of random matrices there are other tools, which yield better estimates and which are easier to apply. The following theorem is useful, when the matrix $Z$ is the sum of matrices of the type $\gamma_i X_i $ where $X_i$ is a fixed matrix and $\gamma_i$ is a random variable which are circular-symmetric complex normally distributed. It is an immediate corollary of \cite[Theorem 4.1.1]{tropplecture}, where matrices of this type are called Matrix Gaussian Series. For completeness, we include a proof in the Appendix.
\begin{corollary}[Matrix Gaussian Series]\label{gaussianconcentration}
	Let $ X_1, \ldots, X_n \in \mathbb{C}^{ d_1 \times d_2} $ be (fixed)
	matrices, and let $ \gamma_1, \ldots, \gamma_n $ be independent, identically distributed
	random variables, where $\gamma_i$ has circular symmetric gaussian distribution
	$ \mathcal{CN} \left(0,1\right)  $.
	Set $Z= \sum^{n}_{i=1} \gamma_i X_i$ and
	%	\begin{equation*}
	%	\sigma^2 = \max \left\{  \Big\Vert \sum^n_{i=1} \sigma_i^2 X_i X^*_i  \Big\Vert_{2 \rightarrow 2}, \Big\Vert \sum^n_{i=1} \sigma_i^2 X^*_i X_i \Big\Vert_{2 \rightarrow 2}   \right\}.
	%	\end{equation*}
	\begin{align*}
	\sigma^2 &= \max \left\{  \Big\Vert  \mathbb{E} \left[ Z Z^* \right] \Big\Vert_{2 \rightarrow 2},  \Big\Vert \mathbb{E} \left[ Z^* Z \right] \Big\Vert_{2 \rightarrow 2}    \right\}\\
	&= \max \left\{ \Big\Vert \sum_{i=1}^{n} X_i X^*_i \Big\Vert_{2 \rightarrow 2} ; \Big\Vert \sum_{i=1}^{n} X^*_i X_i \Big\Vert_{2 \rightarrow 2}   \right\}.
	\end{align*}
	Then, for all $t>0$, with probability at least $1-2\exp \left(-t\right) $ 
	\begin{equation*}
	\big\Vert Z \big\Vert_{2 \rightarrow 2} \le  \sigma \sqrt{2 \left( t + \log \left(d_1 + d_2 \right) \right)}.
	\end{equation*}
	
	%	Then 
	%	\begin{equation*}
	%	\mathbb{E} \Vert Z \Vert_{2 \rightarrow 2} \le \sigma \sqrt{2 \log \left( d_1 + d_2\right)}.
	%	\end{equation*}
	%	Furthermore, for all $t>0$,
	%	\begin{equation*}
	%	\mathbb{P} \left( \Vert Z \Vert_{2 \rightarrow 2} \ge t \right) \le  \left(K + N \right) \exp \left( \frac{-t^2}{2 \sigma^2} \right).
	%	\end{equation*}
\end{corollary}

.
\subsection{Suprema of Chaos Processes}\label{sec:supremaandcoveringnumbers}
%A crucial part of this paper deals with finding an upper bound for 
%\begin{equation*}
%\underset{X \in T^p}{\sup}  \vert  \Vert \mathcal{A}^p \left(X\right) \Vert^2_{\ell_2} - \mathbb{E} \left[ \Vert \mathcal{A}^p \left(X \right) \right] \Vert^2_{\ell_2}    \vert.
%\end{equation*}
In addition to sums of random matrices, random variables of the form $ \underset{A \in \mathcal{X}}{\sup} \Vert A \xi \Vert  $, where $ \xi $ is a random vector and $ \mathcal{X}$ is a class of matrices,  will play an important role in this paper. To state a tail bound for such random variables, we need the $\gamma_2$-functional, a geometric quantity introduced by Talagrand (see \cite{talagrand2014upper}).
\begin{definition}
Let $ \left( X,  \vertiii{\cdot} \right) $ be a Banach space and suppose that $S \subset X $. We say that a sequence $ \left( S_n \right)_{n \ge 0}  $ of subsets of $ S $ is admissible, if $ \vert S_0 \vert = 1 $ and  $ \vert  S_n \vert \le 2^{2^n} $ for all $ n \ge 1 $. Then we set
\begin{equation*}
\gamma_2 \left( S,  \vertiii{\cdot} \right) = \inf_{\left( S_n \right)_{n \ge 0}} \underset{s \in S}{\sup} \sum_{n=0}^{\infty} 2^{n/2} \underset{s \in S_n}{\inf} \vertiii{ s-s_n} ,
\end{equation*}
where the infimum is taken over all admissible sequences $ \left( S_n \right)_{n \ge 0}  $.
\end{definition}
The $ \gamma_2$-functional fulfills the following inequality.
\begin{lemma}[\cite{riplikeproperties}, Lemma 2.1]\label{gamma_2sum}
	Let $\left(X, \vertiii{\cdot} \right)$ be an arbitrary Banach space. Suppose that $A,B \subset X$. Then
	\begin{equation*}
	\gamma_2 \left(A + B, \vertiii{\cdot} \right) \lesssim  \gamma_2 \left(A ,  \vertiii{\cdot} \right) + \gamma_2 \left( B, \vertiii{\cdot} \right) .
	\end{equation*}
\end{lemma}
Let $ \mathcal{X}$ be any set of matrices and define $  d_{F} \left( S \right) = \underset{ A \in \mathcal{X}}{\sup}~  \Vert A \Vert_F$ and $  d_{ 2 \rightarrow 2} \left( S \right) = \underset{ A \in \mathcal{X}}{\sup}~ \Vert A \Vert_{2 \rightarrow 2} $. We can now state the following theorem, which will be crucial in Section \ref{sec:localisometry}.
\begin{theorem}\cite[Theorem 1.4]{KMF2015}\label{KMR} 
	Let $ \mathcal{X} $ be a  symmetric set of matrices, i.e., \mbox{$ \mathcal{X} = - \mathcal{X} $} and let $ \xi $ be a random vector whose entries $\xi_i $ are independent circular-symmetric standard normal random variables with mean $0$ and variance $1$. Set
	\begin{align*}
	E&= \gamma_2 \left( \mathcal{X}, \Vert \cdot \Vert_{2\rightarrow 2}  \right) \left(  \gamma_2 \left( \mathcal{X}, \Vert \cdot \Vert_{2\rightarrow 2}  \right) + d_F (\mathcal{X}) \right) \\
	V&= d_{2 \rightarrow 2} \left( \mathcal{X} \right) \left( \gamma_2 \left( \mathcal{X}, \Vert \cdot \Vert_{2\rightarrow 2}  \right) +   d_F (\mathcal{X}) \right)\\
	U&= d^2_{2 \rightarrow 2} \left( \mathcal{X} \right) 
	\end{align*}
	Then, for $t>0$, 
	\begin{equation*}
	\mathbb{P} \left( \underset{A \in \mathcal{X}}{\sup} \big\vert \Vert A \xi \Vert_{\ell_2}^2 - \mathbb{E} \Vert A\xi \Vert_{\ell_2}^2  \big\vert \ge c_1 E + t \right) \le 2 \exp \left( -c_2 \min \left( \frac{t^2}{V^2}, \frac{t}{U} \right) \right).
	\end{equation*}
	The constants $c_1$ and $c_2 $ are universal.
\end{theorem}

%The $\gamma_2 $-functional can be bounded by an integral involving the covering number of the set $S$. 
Dudley's inequality yields a relation of the $ \gamma_2$-functional to covering numbers. Recall that the covering number $ N \left( S,  \vertiii{\cdot}, \varepsilon \right) $ is the minimum number of open $\vertiii{\cdot}$-balls with radius $ \varepsilon$, whose midpoint is contained in $S$, which are needed to cover $ S$. More precisely, Dudley's inequality (see \cite[Proposition 2.2.10]{talagrand2014upper}, \cite{DUDLEY1967290}) states that
\begin{equation}\label{dudley}
\gamma_2 \left( S,  \vertiii{\cdot} \right) \lesssim  \int_{0}^{d_{\vertiii{\cdot}} \left( S \right) } \sqrt{ \log N \left( S , \vertiii{\cdot} ,\varepsilon \right) } d\varepsilon,
\end{equation}
where $ d_{\vertiii{\cdot}} \left( S \right) = \underset{ x \in S}{\sup}~ \vertiii{x} $. For this reason, we will need some bounds for covering numbers, which are summarized in the following section.

\subsection{Covering Numbers}\label{sec:coveringnumbers}
The following lemma is a slight modification of the Maurey lemma by Carl \cite{Carl1985}. (See also \cite[Lemma 4.2]{KMF2015} for a formulation of this lemma using notation which is closer to the notation in this paper.)
\begin{lemma}\label{Maurey}
Let $\left(X,   \vertiii{\cdot}  \right)$ be a normed space, consider a finite set $\mathcal{U} \subset X$, and assume that for every $L \in \mathbb{N}$ and $\left( u_1, \ldots , u_L \right) \in \mathcal{U}^L$, $\mathbb{E}_{\varepsilon} \vertiii{ \sum_{j=1}^{L} \varepsilon_j u_j } \le A \sqrt{L} $, where $\left( \varepsilon_j \right)_{j=1}^{L} $ denotes a Rademacher vector. Then, for every $u > 0$, 
	\begin{equation*}
	\log N \left( \text{conv} \left( \mathcal{U} \right)  , \vertiii{ \cdot }  ,u\right) \lesssim  \frac{A^2}{u^2}  \log \vert \mathcal{U} \vert,
	\end{equation*}
	where $  \vert \mathcal{U} \vert $ denotes the cardinality of $ \mathcal{U}  $.
\end{lemma}
Let $ V \subset \mathbb{R}^n$ be a compact, convex, and symmetric set which is absorbing, i.e. $ \mathbb{R}^n = \underset{t>0}{\bigcup} t V $. We will denote by $ \Vert \cdot \Vert_{V}$ the norm associated with $V$, i.e., the unique norm whose unit ball is given by $V$. Furthermore, denote by $V^{\circ}$ the polar body of $V$, i.e., 
\begin{equation*}
V^{\circ} = \left\{ u \in \mathbb{R}^n:  \langle u,v \rangle \le 1 \text{ for all } v \in V  \right\}.
\end{equation*}
An elementary consequence of the definition is that the dual norm of $ \Vert \cdot \Vert_{V} $ is given by $\Vert \cdot \Vert_{V^{\circ}} $. %Furthermore, we have that $ \text{conv} \left( V \right) = V^{\circ \circ}  $. % In 1972, Pietsch conjectured (see \cite{pietschconjecture}) that if $V,W \subset \mathbb{R}^n$ are two symmetric, convex bodies, then $\log N \left(V, \Vert \cdot \Vert_{W}, \varepsilon  \right) $ is  equivalent to $ \log N \left(W^{\circ}, \Vert \cdot \Vert_{V^0}, \varepsilon \right) $ in an appropriate sense. %In 2004, Artstein, Milman, and Szarek were able to verify this claim in a special case, namely if one of the two bodies is the euclidean ball:
The following result about covering numbers of polar bodies solved a special instance of a conjecture by Pietsch \cite{pietschconjecture}.
\begin{theorem}[\cite{artsteinmilman}]\label{dualitymetricentropy}
	As above, assume $ V \subset \mathbb{R}^n$ to be a compact, convex, symmetric, and absorbing set. Then, for all $ \varepsilon>0$
	\begin{equation*}
	\log N \left( B_{1} \left( 0 \right)  ,\Vert \cdot \Vert_{V}   , \varepsilon \right)  \lesssim \log N \left(  V^{\circ}  ,\Vert \cdot \Vert_{\ell_2}  , c \varepsilon \right),
	\end{equation*}
	where $c>0$ is a universal constant and $B \left(0,1\right):= \left\{ x \in \mathbb{R}^n: \Vert x \Vert_{\ell_2} \le 1 \right\} $.
\end{theorem}

%One way to deal with the covering numbers arising in inequality \ref{dudley} is the volumetric estimate. Another possibility is given by the following lemma:

%In order to apply the lemma above one usually has to show $ \mathbb{E}_{\varepsilon} \Vert \sum_{j=1}^{L} \varepsilon_j u_j \Vert_X \le A \sqrt{L} $. The following lemma will be useful in these situations.
%\begin{lemma}[ \cite{FR2013} Proposition 7.14]\label{expbound}
%	Let $Z$ be a random variable which satisfies for all $u \ge \sqrt{2 ln \left( \beta \right)} $
%	\begin{equation*}
%	P \left( \vert Z \vert \ge \alpha u \right) \le \beta \exp \left( -\frac{u^2}{2} \right)
%	\end{equation*}
%	and for some constants  $\alpha > 0 $, $\beta \ge 2 $. Then 
%	\begin{equation*}
%	\mathbb{E} \Vert Z \Vert \le \frac{3}{2} \alpha \sqrt{\ln \left(4 \beta\right)} 
%	\end{equation*}
%	%	with $C_{\beta}= \sqrt{2} + \frac{1}{4 \sqrt{2} \ln \left(4 \beta \right) }  < \frac{3}{2}$.
%\end{lemma}

%\begin{proof}
%	A proof can be found in \cite[Proposition 7.14]{FR2013}.
%\end{proof}
%\begin{lemma}[\cite{FR2013}, Lemma C.9.]\label{intineq1}
%	For $\alpha > 0 $,
%	\begin{equation*}
%	\int_{0}^{\alpha} \sqrt{\log \left(1+ t^{-1} \right)} dt \le \alpha \sqrt{\log \left( e \left(1+ \alpha^{-1} \right) \right)}
%	\end{equation*}	
%\end{lemma}
%
%

\section{Outline of the Proof}
\label{sec:outline}
In this section we give a rough outline of our
proof and highlight the main differences to
previous work (\cite{ARR2012} and \cite{lingstrohmer}). In particular,
we want to point out those parts, which enabled us to overcome the suboptimal scaling in $r$. The overall strategy of our proof remains similar to
the one in \cite{lingstrohmer} and in \cite{ARR2012}: First, we will
prove sufficient conditions for recovery. These conditions will rely
on the existence of a so-called inexact dual certificate. In the
second step this certificate will be constructed via the Golfing
Scheme, a method which has been introduced by Gross and others (see
\cite{gross2011recovering}).\\

As already mentioned, the first part of the proof consists of showing that
the existence of the inexact dual certificate is a sufficient condition for recovery. 
This will be proven in Section \ref{sectionsufficientconditions}. The underlying observation is that in certain cases,
it suffices that standard conditions defining minimizers are only approximately satisfied.
In \cite{lingstrohmer}, these perturbed conditions are given by \cite[(28)]{lingstrohmer}. In order for them to imply that $X^0$ is a minimizer, one needs that $\mathcal{A}_i$ acts approximately as an isometry on each
\begin{equation*}
\mathcal{T}_i = \left\{ h_i u_i^* + v_i m^*_i : \ u_i \in \mathbb{C}^{K_i}, v_i \in \mathbb{C}^{N_i}   \right\}
\end{equation*}
 and that the images of these operators are almost orthogonal to each other. The latter condition is responsible for the appearance of the quadratic scaling in $r$ in \cite{lingstrohmer}. Our approach will be different: We will show that the operator $ \mathcal{A}$ acts as an approximate isometry on the full subspace
 \begin{equation*}
\mathcal{T} := \left\{  \left( X_1, \ldots, X_r \right): \ X_i \in \mathcal{T}_i \text{ for all } i \in \lbrack r \rbrack   \right\}.
 \end{equation*} 
%In the following we will denote by $ \mathcal{P}_{\mathcal{T}}$ the orthogonal projection onto $ \mathcal{T}$. Now let us specificy what we mean by saying that $\mathcal{A}$ acts as an approximate isometry on $\mathcal{T}$: 
in the sense of the following definition.
\begin{definition}[Local isometry property]
 $ \mathcal{A}$ fulfills the $ \delta$-local isometry property on $\mathcal{T}$ for some $ \delta > 0 $, if
\begin{equation}\label{ineq:deltalocalisometryoutline}
\left(1 - \delta \right)  \Vert X \Vert^2_F \le \Vert \mathcal{A} \left( X\right) \Vert^2_{\ell_2} \le \left(1 + \delta \right) \Vert X \Vert^2_F
\end{equation}
for all $X \in \mathcal{T} $.
\end{definition}

The main novelty in our proof is that our global viewpoint allows us to establish the local isometry
property on $\mathcal{T}$ with high probability if $L$ scales \emph{linearly}
with $r$. This will be achieved via Theorem \ref{KMR}, which involves a $\gamma_2$-functional. Thus a large part of Section \ref{sec:localisometry} is concerned with estimating those $\gamma_2$-functionals.\\

The local isometry property is not only needed in the first part but
also in the second part of the proof, where the dual certificate is
constructed via the Golfing Scheme. For that we will assume that $\left\{ \Gamma_p \right\}^P_{p=1} $ is fixed $\omega$-admissible partition (see Definition \ref{definition:admissiblepartition}) which minimizes (\ref{definition:muhomega}). For this partition we can introduce the operators $ \mathcal{A}^p$ defined by $\mathcal{A}^p \left( X \right) = P_{\Gamma_p} \left( \mathcal{A} \left( X \right)  \right) $, where $P_{\Gamma_p}: \mathbb{C}^L \rightarrow \mathbb{C}^L $ denotes the (coordinate) projection of $\mathbb{C}^L $ onto the coordinates contained in the set $ \Gamma_p $. Similarly, we will define $ \mathcal{A}^p_i $ by $\mathcal{A}_i^p \left( X \right) = P_{\Gamma_p} \left( \mathcal{A}_i \left( X \right)  \right)  $.\\

We will need that each operator $ \mathcal{A}^p$ satisfies the $\delta$-local isometry property on a subspace $\mathcal{T}^p$, which is slightly larger than $\mathcal{T}$. In order to define the space $ \mathcal{T}^p$ we need to introduce the operators $ \mathcal{S}^p : \mathcal{M} \rightarrow \mathcal{M} $. For that, recall $ S_{i,p} = T^{-1}_{i,p} $ as defined in Section \ref{section:partitioncoherence}. 
\begin{definition}\label{def:operatorSp}
	Let $p \in \lbrack P \rbrack $. Then the operator
	$\mathcal{S}^p: \mathcal{M} \rightarrow \mathcal{M}$ is defined by
	\begin{equation}\label{def:operatorS}
	\mathcal{S}^p \left( W \right) = \left( S_{1,p} W_1 , \ldots, S_{r,p} W_r  \right)
	\end{equation}
	for $W= \left( W_1, \ldots, W_r \right)\in \mathcal{M} $.
\end{definition}
Then $\mathcal{T}^p $ is defined by 
\begin{equation}\label{definition:Tpsubspace}
\mathcal{T}^p = \mathcal{T} + \mathcal{S}^p \left(  \mathcal{T} \right).
\end{equation}
Observe that we may write $\mathcal{T}= \mathcal{T}_h + \mathcal{T}_m $ and $ \mathcal{T}^p= \mathcal{T}_h + \mathcal{T}_{\mathcal{S}^p h}  + \mathcal{T}_m  $, when the subspaces $ \mathcal{T}_m $, $ \mathcal{T}_h  $, and $ \mathcal{T}_{\mathcal{S}^p h}  $ are given by
\begin{align}\label{equ:Tsubspaces}
	\begin{split}
		\mathcal{T}_m &= \left\{ \left( v_1 m^*_1, \ldots, v_r m^*_r \right): \ v_i \in \mathbb{C}^{K_i} \text{ for all } i \in \lbrack r \rbrack     \right\},\\
		\mathcal{T}_h &= \left\{ \left( h_1 u^*_1, \ldots, h_r u^*_r \right): \ u_i \in \mathbb{C}^{N_i} \text{ for all } i \in \lbrack r \rbrack     \right\},\\
		\mathcal{T}_{\mathcal{S}^p h} &= \left\{ \left( \left( S_{1,p} h_1 \right) u^*_1, \ldots, \left( S_{r,p} h_r \right) u^*_r \right): \ u_i \in \mathbb{C}^{N_i} \text{ for all } i \in \lbrack r \rbrack     \right\}.
\end{split}
\end{align}
%In the following, we will denote by $ \mathcal{P}_{\mathcal{T}^p} $ the orthogonal projection onto $ \mathcal{T}^p $. In Section \ref{subsec:golfing} we will explain why we will need that each operator $ \mathcal{A}^p $ fulfills the local isometry property on a larger subspace than $ \mathcal{T} $. \\
As already mentioned, the local isometry property on $\mathcal{T}$, respectively $ \mathcal{T}^p$, will be shown in Section \ref{sec:localisometry}. In Section \ref{subsec:dual:certificate} the approximate dual certificate will be constructed via the Golfing Scheme. Finally, in Section \ref{subsec:proofmaintheorem} we will prove the main result, Theorem \ref{theorem:mainwithnoise}.
% By saying that the operator $ \mathcal{A}^p   $ satisfies the $\delta$-local isometry property on $ \mathcal{T}^p $ we mean that
%\begin{equation*}
%\left(1- \delta \right) \Vert X \Vert^2_F \le \Vert \mathcal{A}^p \left( X \right) \Vert_{\ell_2} \le \left(1 + \delta  \right) \Vert X \Vert^2_F
% \end{equation*}
%is fulfilled for all $ X \in \mathcal{T}^p $.
%Observe that we can decompose $T_p$ into
%\begin{equation}
%T_p = T_h + T_{S_p h} + T_m,
%\end{equation}
%where $T_{S_p h} $ is defined by
%\begin{equation*}
%  T_{ S_{p} h} = \left\{  \left( S_{1,p} h_1 v^*_1, \cdots , S_{i,p}
%      h_i v_i^*, \cdots,  S_{r,p} h_r v^*_r \right) : 
%    \ v_i \in \mathbb{C}^{N_i} \text{ for all } 1\le i \le r    \right\}.
%\end{equation*}
%As already mentioned in Section \ref{sec:localisometry} we will prove that the operator $ \mathcal{A} $, respectively $ \mathcal{A}^p $, satisfies the $ \delta$-local isometry property on $\mathcal{T}$, respectively $\mathcal{T}^p$, with sufficiently high probability. After that, in Section \ref{subsec:dual:certificate}, the (inexact) dual certificate will be constructed using the Golfing Scheme, which finishes the proof.
\section{Proof of the Main Theorem}\label{sec:proof}
\subsection{Sufficient Conditions for Recovery}
\label{sectionsufficientconditions}
%At the core of both \cite{lingstrohmer} and our proof is the
%construction of a so-called approximate dual certificate. The
%underlying derivation (\cite{gross2011recovering}) is that in certain
%cases, it suffices that standard conditions defining minimizers of
%convex optimization problems are only approximately satisfied. In
%\cite{lingstrohmer}, these perturbed conditions are given by
%\cite[(28)]{lingstrohmer}. In order for them to imply that $X_0$ is a
%minimizer, one needs that $\mathcal{A}_i$ is approximately an isometry
%on $T_i$ and that the images of these operators are almost orthogonal
%to each other. The latter condition is responsible for the appearance
%of the $r^2$-bottleneck. We take a slightly different apprach not
%considering all entries of $Y$ separately in condition
%(\ref{dualcertificatecondition1}). This allows us to replace the
%conditions \cite[(28)]{lingstrohmer} by the $\delta$-local isometry
%property as introduced in (...).We obtain the following theorem.
As already mentioned in the outline of the proof, in this section we will show  that the existence of
an inexact dual certificate implies that the signal is approximately recovered. Therefore, we will denote in the following by $ \mathcal{P}_{\mathcal{T}}$ the orthogonal projection onto $ \mathcal{T}$. Similarly, we will denote by for all $i \in \lbrack r \rbrack $ the orthogonal projection onto $ \mathcal{T}_i $
\begin{lemma}\label{dualcertificate}
  Suppose that $\mathcal{A}$ satisfies the $\delta$-local
  isometry property on $\mathcal{T}$ (\ref{ineq:deltalocalisometryoutline}) and set
  $\gamma =\Vert \mathcal{A} \Vert_{F \rightarrow 2}$, i.e., $ \gamma$ is the operator norm of $ \mathcal{A} $. Furthermore,
  suppose that there is
  $Y= \left( Y_1, \ldots, Y_r \right)= \mathcal{A}^* z $ for some
  $z \in \mathbb{C}^L$ such that
  \begin{align}
    \Vert \mathcal{P}_{\mathcal{T}} Y - \text{sgn} \left( X^0 \right) \Vert_{F} 
    &\le \alpha \label{dualcertificatecondition1} \\
    \Vert \mathcal{P}_{\mathcal{T}_i^{\perp}} Y_i \Vert_{2 \rightarrow 2} 
    &\le \beta   \text{ for all } i \in \lbrack r \rbrack , \label{dualcertificatecondition2}
  \end{align}
  where $\alpha, \beta \ge 0 $ are constants such that
  $ 1- \beta - \frac{\alpha \gamma}{\sqrt{1- \delta}} \ge
  \frac{1}{2}$,
  $ \alpha \le 1 $, and $\sqrt{1-\delta} \ge \frac{1}{2} $. Then if
  $ \hat{X} $ is a minimizer of
  \begin{align*}
    \text{minimize } & \Vert X \Vert_{\ast} \\
    \text{subject to } & \Vert \mathcal{A} \left( X \right) - \hat{y}  \Vert_{\ell_2} \le \tau
  \end{align*}
  we have that
  \begin{equation}\label{recoveryerror}
    \Vert \hat{X} - X^{0} \Vert_F \lesssim \tau  \left(1 + \gamma \right) \left(1 + \Vert z \Vert_{\ell_2} \right). 
  \end{equation}
\end{lemma}

\begin{proof}
  Set $V= \left(V_1, \ldots, V_r \right)= \hat{X} - X^0$ and note that we seek to estimate
  $\Vert V \Vert_F \le \Vert \mathcal{P}_{\mathcal{T}}  \left(V\right) \Vert_F +
  \Vert \mathcal{P}_{\mathcal{T}^{\perp}} \left( V \right) \Vert_F $
  from above. We observe that
  \begin{equation}\label{intern30}
    \Vert \mathcal{A} \left( V \right) \Vert_{\ell_2} \le \Vert \mathcal{A} ( \hat{X} ) - \hat{y} \Vert_{\ell_2} + \Vert \hat{y} - \mathcal{A} \left( X^0 \right) \Vert_{\ell_2} \le 2 \tau.
  \end{equation}
% First, we are going to find an upper bound for $ \Vert \mathcal{P}_{\mathcal{T}} \left(V\right) \Vert_F $. 
 Together with the $\delta$-local isometry property (\ref{ineq:deltalocalisometryoutline}), the definition of $\gamma$, and the triangle inequality we obtain
  \begin{align*}
    \Vert \mathcal{P}_{\mathcal{T}} \left(V\right)  \Vert_{F} 
    &\le \frac{1}{\sqrt{1-\delta}}  \Vert \mathcal{A} \left( \mathcal{P}_{\mathcal{T}}  \left(V\right) \right) \Vert_{\ell_2}  \le \frac{1}{\sqrt{1-\delta}} \left( \Vert \mathcal{A} \left( \mathcal{P}_{\mathcal{T}^{\perp}}  \left( V\right) \right) \Vert_{\ell_2} + \Vert \mathcal{A} \left(V\right) \Vert_{\ell_2}  \right)\\
    & \le  \frac{1}{\sqrt{1-\delta}} \left( \gamma \Vert \mathcal{P}_{\mathcal{T}^{\perp}} \left(V\right) \Vert_{F}  + 2 \tau \right).
  \end{align*}
  Thus it remains to find an upper bound for
  $\Vert \mathcal{P}_{\mathcal{T}^{\perp}} \left(V\right) \Vert_{F} $. For that
  purpose, choose $Z= \left(Z_1, \ldots, Z_r\right) $ such that for
  all $i \in \lbrack r \rbrack $ we have that
  $Z_i \in {\mathcal{T}}^{\perp}_i$,
  $\Vert Z_i \Vert_{2 \rightarrow 2} \le 1-\beta$, and
  $ \langle Z_i , V_i \rangle_F = \left( 1- \beta \right) \Vert
  \mathcal{P}_{{\mathcal{T}}_i^{\perp}} V_i \Vert_{\ast} $.
  This is possible by duality of the norms
  $ \Vert \cdot \Vert_{2\rightarrow 2} $ and
  $\Vert \cdot \Vert_{\ast} $ (see \cite[Section 4.2]{bhatia2013matrix}). Observe that
  and
  $ \Vert \text{sgn} \left( X^0_i \right) + \mathcal{P}_{\mathcal{T}^{\perp}_i} Y_i +Z_i \Vert_{2 \rightarrow 2}
  \le 1 $ as both the row and column spaces of $\text{sgn} \left( X^0_i \right)  $ and $ \mathcal{P}_{\mathcal{T}^{\perp}_i} Y_i +Z_i $ are orthogonal. Thus, again using  the duality between $ \Vert \cdot \Vert_{2 \rightarrow 2} $ and $ \Vert \cdot \Vert_{\ast} $, we obtain
  \begin{align*}
  \Vert X_i^0 + V_i \Vert_{\ast} &= \underset{W \in \mathbb{C}^{K_i \times N_i}, \ \Vert W \Vert_{2 \rightarrow 2} \le 1}{\sup} \vert \langle W, X_i^0 + V_i  \rangle_F  \vert\\
  &\ge \text{Re} \left( \langle \sgn \left( X^0_i \right) + \mathcal{P}_{\mathcal{T}^{\perp}_i} Y_i + Z_i  , X^0_i + V_i \rangle_F \right) \\
  &\ge \Vert X^0_i \Vert_{\ast} + \text{Re} \left( \langle \sgn \left( X^0_i \right) + \mathcal{P}_{\mathcal{T}^{\perp}_i} Y_i, V_i \rangle_F  \right) + \left(1-\beta \right)  \Vert \mathcal{P}_{\mathcal{T}_i} V_i  \Vert_{\ast}
  \end{align*}
Here, in the second inequality we used that  $ \mathcal{P}_{\mathcal{T}^{\perp}_i} Y_i +Z_i \in \mathcal{T}^{\perp}_i $ and $ \langle  \text{sgn} \left(  X^0_i  \right), X^0_i  \rangle_F = \Vert X^0_i \Vert_{\ast} $. Thus, by definition of $ \Vert X^0 + V \Vert_{\ast} $ we obtain
  \begin{align*}
    \Vert X^0 + V \Vert_{\ast} &\ge \sum_{i=1}^{r} \Vert X^0_i \Vert_{\ast} + \sum_{i=1}^{r} \text{Re} \left( \langle \text{sgn} \left( X^0_i \right) + \mathcal{P}_{\mathcal{T}^{\perp}_i} Y_i    , V_i \rangle_F \right) + \left(1- \beta \right) \sum_{i=1}^{r}  \Vert \mathcal{P}_{\mathcal{T}_i} V_i  \Vert_{\ast}  \\
    & =  \Vert X^0 \Vert_{\ast} + \text{Re} \left( \langle \text{sgn} \left( X^0 \right) - \mathcal{P}_{\mathcal{T}}  Y   , V \rangle_{F} + \langle Y,V \rangle_F \right) + \left(1- \beta \right) \Vert \mathcal{P}_{\mathcal{T}^{\perp}} V  \Vert_{\ast}.
  \end{align*}
%  Thus,
%  \begin{align*}
%    \Vert X^0 +V \Vert_{\ast} % &\ge  \Vert X^0 \Vert_{\ast} + \text{Re} \left( \langle  \text{sgn } X^0 + \mathcal{P}_{T^{\perp}}  Y + Z  ,V  \rangle_{F} \right) 
%    & \ge \Vert X^0 \Vert_{\ast} + \text{Re} \left( \langle \text{sgn } X^0 - \mathcal{P}_{\mathcal{T}}  Y    , V \rangle_{F} + \langle Y,V \rangle_{F} + \langle Z,V \rangle_{F}  \right) \\
%    &=  \Vert X^0 \Vert_{\ast} + \text{Re} \left( \langle \text{sgn } X^0 - \mathcal{P}_{\mathcal{T}}  Y   , V \rangle_{F} + \langle Y,V \rangle_F \right) + \left(1- \beta \right) \Vert \mathcal{P}_{\mathcal{T}^{\perp}} V  \Vert_{\ast}.
%  \end{align*}
%  The second line follows from our choice of $Z$.
 Now observe that by Cauchy-Schwarz, (\ref{dualcertificatecondition1}) and our upper bound for $\Vert \mathcal{P}_{\mathcal{T}} \left( V \right) \Vert_{\ell_2} $
\begin{align*}
  \text{Re} \left( \langle \text{sgn} \left( X^0 \right) - \mathcal{P}_{\mathcal{T}} \left( Y  \right)  , V \rangle_{F} \right)  &\ge  -\Vert \text{sgn} \left( X^0 \right) - \mathcal{P}_{\mathcal{T}} \left( Y  \right) \Vert_{F} \Vert \mathcal{P}_{\mathcal{T}} \left(V\right) \Vert_{F}\\
   &\ge \frac{- \alpha}{\sqrt{1-\delta}} \left(  \gamma \Vert \mathcal{P}_{\mathcal{T}^{\perp}} V \Vert_{F} + 2 \tau \right).
\end{align*}
  Furthermore, we note that by Cauchy-Schwarz and (\ref{intern30})
  \begin{equation*}
    \text{Re} \left(  \langle Y, V \rangle_F \right) = \text{Re} \left(  \langle \mathcal{A}^* z , V \rangle_F \right) = \left(  \langle z, \mathcal{A} \left( V \right) \rangle_{\ell_2} \right) \ge -2 \Vert z \Vert_{\ell_2} \tau.
  \end{equation*}	
  Combining the last three calculations and using that the nuclear norm is greater or equal than the Frobenius norm we obtain
  \begin{align*}
    \Vert \hat{X} \Vert_{\ast} \ge \Vert X^0 \Vert_{\ast} + \left(1- \beta -  \frac{\alpha \gamma}{\sqrt{1 - \delta}} \right)  \Vert \mathcal{P}_{\mathcal{T}^{\perp}} V \Vert_{\ast} - 2\tau \left( \Vert z \Vert_{\ell_2} +\frac{\alpha }{\sqrt{1 - \delta}} \right) .
  \end{align*}
  As $ \hat{X} $ is the nuclear norm minimizer and we have that $ \Vert X^0 \Vert_{\ast} \ge  \Vert \hat{X} \Vert_{\ast} $ this yields
  \begin{equation*}
    \left(1- \beta -  \frac{\alpha \gamma}{\sqrt{1 - \delta}} \right)  \Vert \mathcal{P}_{\mathcal{T}^{\perp}} \left(V\right) \Vert_{\ast} \le 2\tau \left( \Vert z \Vert_{\ell_2} +\frac{\alpha }{\sqrt{1 - \delta}} \right).
  \end{equation*}
  By our assumptions on $\alpha$, $\beta$, and $\delta$ this implies
  \begin{equation*}
    \Vert  \mathcal{P}_{\mathcal{T}^{\perp}} \left(V\right) \Vert_{F} \lesssim \tau \left( \Vert z \Vert_{\ell_2}+1  \right).
  \end{equation*}
  Thus, using again the upper bound for $\Vert \mathcal{P}_{\mathcal{T}} \left(V\right) \Vert_{F} $, which was calculated above, and again our assumptions on $\alpha$, $\beta$, and $\delta$ we obtain
  \begin{align*}
    \Vert V \Vert_F &\le \Vert \mathcal{P}_{\mathcal{T}} \left(V\right) \Vert_F + \Vert  \mathcal{P}_{\mathcal{T}^{\perp}} \left(V\right) \Vert_F \lesssim  \left( 1+ \gamma  \right) \Vert \mathcal{P}_{\mathcal{T}^{\perp}} \left( V \right) \Vert_{F} +  \tau  \lesssim \tau \left( 1+ \gamma  \right)  \left( 1+ \Vert z \Vert_{\ell_2}  \right),
  \end{align*}
  which finishes the proof.
\end{proof}
As already mentioned in the introduction, the noiseless case is also of interest for us. Note that in this
situation we may set $\tau= 0$ and Lemma \ref{dualcertificate} shows
that the existence of a dual certificate implies that the convex program
(\ref{eq:nucbpdn}) recovers the signal $X^0$ exactly.

%\begin{remark}
% Note that in this case Lemma \ref{dualcertificate} gives
% exact recovery, i.e. $\hat{X}=X_0 $, for the semi-definite program
% ().
%\end{remark}
% In order to use Lemma \ref{dualcertificate} one needs to estimate $\gamma$. The following lemma gives (with high probability) an upper bound for $\gamma$.
% \begin{lemma}[\cite{ARR2012}, Lemma 1]\label{alphaoperatornorm}
%Let $ \alpha \ge 1 $. Then with probability exceeding $1 - L^{-\alpha-1}$
%	\begin{equation*}
%	\Vert \mathcal{A}_i \Vert_{2\rightarrow 2} \le \sqrt{N \left( \log \left(\frac{NL}{2}\right) + \left( \alpha +1 \right) \log L \right)}.
%	\end{equation*}
%\end{lemma}

\begin{remark}\label{dualcertificateremark}
Note that we still have the freedom to choose the parameters $ \alpha $ and $ \beta $ in Lemma \ref{dualcertificate}. In Section \ref{subsec:dual:certificate} we will construct a dual certificate $Y$ for the following choice of parameters: We set $\beta= \frac{1}{4} $ and assume that $\delta \le \frac{1}{4} $. In order to fulfill the condition $1-\beta - \frac{\alpha \gamma}{\sqrt{1-\delta}} \ge \frac{1}{2} $ it is  then enough to choose $\alpha = \frac{1}{8 \gamma}.$
\end{remark}
Note that in the noisy case the error estimate in Lemma \ref{dualcertificate} depends linearly on the operator norm of $ \mathcal{A}$ as (\ref{recoveryerror}) states. Thus, we need an upper bound for the operator norm of $ \mathcal{A} $ which holds with high probability.
\begin{lemma}\label{lemma:operatornormbound}
Let $\omega \ge 1 $. Then with probability at least $1 - 2  L^{-\omega}  $ we have that
\begin{equation*}
\Vert \mathcal{A} \Vert_{F \rightarrow 2} \le 2  \sqrt{  \omega \max  \left\{ 1; \frac{r K_{\mu} N}{L}   \right\}  \log \left(L+rKN\right) }.
\end{equation*}
\end{lemma}

\begin{proof}
  The result will be proven by using Corollary
  \ref{gaussianconcentration}. %We claim that $\mathcal{A}$ can be represented as a matrix, which is a matrix Gaussian series in the sense of Corollary \ref{gaussianconcentration}.
  Indeed, we can represent each operator $\mathcal{A}_i$ as
  $ \mathcal{A}_i = \sum_{\ell \in L} \sum_{j=1}^{K_i}
  \mathcal{B}_{\ell,j} $ such that each operator
  $ \mathcal{B}_{\ell,j} $ depends linearly on the $(\ell,k)$th entry
  of $C_i$, i.e.,
  $ \left(C_i\right)_{\ell,k} \sim \mathcal{CN} \left(0,1\right)
  $. Thus, we need to estimate the operator norms of
  $ \mathbb{E} \left[ \mathcal{A}^* \mathcal{A} \right] $ and
  $ \mathbb{E} \left[ \mathcal{A} \mathcal{A}^* \right]$. Observe that
\begin{equation*}
  \mathcal{A}^* \mathcal{A} = \left( \mathcal{A}_1^* \left(  \sum_{i=1}^{r} \mathcal{A}_i \right) , \ldots, \mathcal{A}_r^*  \left( \sum_{i=1}^{r} \mathcal{A}_i  \right) \right).
\end{equation*}
Note that the operators $ \left\{ \mathcal{ A }_i \right\}^r_{i=1} $
are independent with expectation
$ \mathbb{E} \left[ \mathcal{A}_i \right] = 0 $ for all
$ i \in \left[r\right] $. Thus
$ \mathbb{E} \left[ \mathcal{A}^* \mathcal{A} \right] = \left(
  \mathbb{E} \left[ \mathcal{A}^*_1 \mathcal{A}_1 \right], \ldots,
  \mathbb{E} \left[ \mathcal{A}^*_r \mathcal{A}_r \right] \right)
$. Let $Z= \left( Z_1, \ldots, Z_r\right) \in \mathcal{M} $. Using
(\ref{equ:Aitransposed}) we compute
\begin{equation}\label{equ:operatornormboundinline1}
  \mathbb{E} \left[ \left(  \mathcal{A}_i^* \mathcal{A}_i \right) \left( Z_i \right) \right]= \sum_{\ell=1}^{L} \mathbb{E} \left[ \left( \mathcal{A}_i \left(Z_i\right) \left( \ell \right) \right) b_{i,\ell}  c^*_{i, \ell}   \right]  = \sum_{\ell=1}^{L} \mathbb{E} \left[ b_{i, \ell } b^*_{i, \ell } Z_i c_{i,\ell} c^*_{i, \ell} \right] = \sum_{\ell=1}^{L} b_{i,\ell} b^*_{i,\ell} Z_i = Z_i
\end{equation}
Thus,
$ \mathbb{E} \left[ \mathcal{A}^* \mathcal{A} \left(Z\right) \right] =
Z$ for any $Z \in \mathcal{M} $, which implies
$ \mathbb{E} \left[ \mathcal{A}^* \mathcal{A}\right] = \Id $. To
compute $ \mathbb{E} \left[ \mathcal{A} \mathcal{A}^* \right] $ let
$ y \in \mathbb{C}^L $ be arbitrary. We compute with similar arguments
as before
\begin{align}
\mathbb{E} \left[  \left( \mathcal{A} \mathcal{A}^* y \right) \left(\ell\right) \right] &= \sum_{i=1}^{r} \mathbb{E}\left[ \left( \mathcal{A}_i \mathcal{A}^*_i  y \right) \left(\ell\right) \right] = \sum_{i=1}^{r}    \mathbb{E} \left[ b^*_{i,\ell}  \left( \mathcal{A}^*_i y \right)  c_{i,\ell} \right]  \notag \\
& \overset{(\ref{equ:Aitransposed})}{=} \sum_{i=1}^{r} \sum_{\ell '=1}^{L} y \left( \ell' \right) \mathbb{E} \left[  b^*_{i, \ell}  b_{i, \ell'}  c^*_{i, \ell'}  c_{i, \ell}  \right] \notag\\
&= y \left( \ell \right) \sum_{i=1}^{r} \mathbb{E} \left[  b^*_{i, \ell}  b_{i, \ell}  c^*_{i, \ell}  c_{i, \ell}   \right]= y \left( \ell \right)  \sum_{i=1}^{r}  \Vert b_{i,\ell} \Vert^2_{\ell_2} N_i . \label{equ:exp}
\end{align}
This shows that $\mathcal{A} \mathcal{A}^* $ can be represented as a diagonal matrix with entries $ \sum_{i=1}^{r}  \Vert b_{i,\ell} \Vert^2_{\ell_2} N_i $. Thus, by definition of $ K_{i,\mu} $ (\ref{def:Kmu}), $ \Vert \mathbb{E} \left[  \mathcal{A} \mathcal{A}^* \right]  \Vert_{2 \rightarrow 2} \le \frac{N \sum_{i=1}^{r} K_{i,\mu} }{L} \     $, which implies, together with (\ref{equ:operatornormboundinline1})
\begin{equation*}
\sigma^2 = \max \left\{ \Vert \mathbb{E} \left[ \mathcal{A}^* \mathcal{A} \right] \Vert_{F \rightarrow F}; \ \Vert \mathbb{E} \left[  \mathcal{A} \mathcal{A}^* \right]  \Vert_{2 \rightarrow 2}   \right\} \le   \max \left\{  1; \  \frac{N \sum_{i=1}^{r} K_{i,\mu} }{L}  \right\}. 
\end{equation*}
Consequently, Corollary \ref{gaussianconcentration} with  $ t= \omega  \log L$ yields that with probability exceeding $1 - 2 L^{-{\omega}} $ 
\begin{align*}
\Vert \mathcal{A} \Vert_{F \rightarrow 2} &\le     \max \left\{  1; \ \sqrt{ \frac{  N \sum_{i=1}^{r} K_{i,\mu}}{L} }  \right\}  \sqrt{  2 \left( \omega \log L  + \log \left(L+rKN\right) \right)  },
%&\lesssim   \max \left\{  1; \  \sqrt{ \frac{r KN \mu^2_{\max}}{L} } \right\}  \sqrt{ \omega \log L}
\end{align*}  
which implies the result.
%The last inequality follows from $ \log \left(L + KN\right) \lesssim \log L $ which is due to (\ref{assumptionsonL}). The result follows from $ \frac{r KN \mu^2_{\max}}{L}  \lesssim \min \left\{ K \mu^2_{\max}, N  \right\} $, which again is due to (\ref{assumptionsonL}).
\end{proof}

\begin{remark}
Note that in \eqref{equ:exp} and other places below, only a weighted sum of the $\Vert b_{i,\ell} \Vert^2_{\ell_2}$ appears. If the summands vastly differ, this may be too crude, and one may consider attempting an averaging argument similar to the one in \cite{krahmer2014stable}. This would, however, require that the proof is completely reworked in some parts. To achieve condition \eqref{dualcertificatecondition2}, for example, we currently rely very much on bounding each $K_{i,\mu}$ individually.
\end{remark}

\subsection{Local isometry property}\label{sec:localisometry}
In this subsection, we establish an isometry of $ \mathcal{A}$,  respectively of $ \mathcal{A}^p$, on $\mathcal{T}$, respectively $ \mathcal{T}^p$. More precisely, we establish the following theorem. 
\begin{theorem}\label{localisometry}
	Fix $ \omega \ge 1 $. Suppose that 
	\begin{equation}\label{intern3}
	Q \ge C_{\omega}  \delta^{-2} r \left( K_{\mu} \log \left(L\right) \log^2 \left( K_{\mu} \right) + N \mu^2_h \right).
	\end{equation}
	Then with probability $1- \mathcal{O} \left( L^{-\omega} \right)$ the following is true: All $X \in \mathcal{T}$ fulfill
	\begin{equation}\label{isometryproperty2}
	\left(1 - \delta \right) \Vert X \Vert^2_{F} \le \Big\Vert  \mathcal{A} \left( X \right) \Big\Vert_{\ell_2}^2     \le \left(1 + \delta \right) \Vert X \Vert^2_{F}
	\end{equation}
	and for all $p \in \lbrack P \rbrack $ every $Y \in \mathcal{T}^p =  \mathcal{T} + \mathcal{S}^p \mathcal{T} $ satisfies
	\begin{equation} \label{isometryproperty1}
	\left(1 - \delta \right) \sum_{i=1}^{r} \Vert T^{1/2}_{i,p} Y_i \Vert^2_{F} \le \frac{L}{Q} \Big\Vert \mathcal{A}^p \left( Y \right) \Big\Vert_{\ell_2}^2     \le \left(1 + \delta \right) \sum_{i=1}^{r} \Vert T^{1/2}_{i,p} Y_i \Vert^2_{F},
	\end{equation}
	where $ T^{1/2}_{i,p} $ denotes the unique positive, self-adjoint matrix whose square is equal to $T_{i,p} $.
\end{theorem}
The proof of this theorem is divided into several steps. For the proof we need some additional notation. Recall that the incoherence parameter $ \mu^2_h $ measures the alignment between the vectors $h_i \in \mathbb{C}^{K_i}$ and $b_{i,\ell} \in \mathbb{C}^{K_i}$. As the operators $\mathcal{A}$ and $ \mathcal{A}_i $ are defined on matrices, it will to be useful to generalize the notion
of incoherence from vectors to matrices. This is achieved by the following definition.
\begin{definition}
	For all $i \in \lbrack r \rbrack $, vectors $z \in \mathbb{C}^{K_i }$
	and  matrices $Z_i \in \mathbb{C}^{K_i \times N_i}$ define
	\begin{align*}
	\Binorm{z} = 
	\sqrt{L}~ \underset{  \ell  \in \lbrack L \rbrack }{\max} \vert z^* b_{i, \ell}   \vert
	\quad\text{and}\quad
	\Binorm{ Z_i }= 
	\sqrt{L}~\underset{  \ell  \in \lbrack L \rbrack }{\max}~\Vert Z_i^*  b_{i, \ell} \Vert_{\ell_2}.
	\end{align*}
	For $Z= \left(Z_1, \ldots, Z_r \right) \in \mathcal{M}$ we define
	\begin{equation*}
		\Bnorm{Z}=\sqrt{ L~\underset{ \ell  \in \lbrack L \rbrack}{\max} \left(
		\sum_{i=1}^{r}  \big\Vert Z^*_i   b_{i, \ell } \big\Vert^2_{\ell_2} \right)}.
	\end{equation*}
	%\begin{align}
	%  \Bnorm{X}=(L\underset{ \ell  \in \lbrack L \rbrack}{\max} 
	%  \sum_{i=1}^{r}  \big\Vert X^*_i   b_{i, \ell } \big\Vert^2_{\ell_2})^\frac{1}{2}.
	%\end{align}
\end{definition}
All these three operations are 
norms as $  \sum_{\ell=1}^{L} b_{i,\ell} b^*_{i,\ell} = \Id_{K_i}$ for all $ i \in \lbrack r \rbrack $. The following lemma provides us with some useful estimates.
\begin{lemma}
	Let $Z= \left(Z_1, \ldots, Z_r \right) \in \mathcal{M}$, 
	$i \in \left[ r \right]$ and  $z \in \mathbb{C}^{K_i} $. Then
	\begin{align}
	\Vert z \Vert_{B_i} 
	&\le \sqrt{K_{i,\mu}}\,\,\Vert z \Vert_{\ell_2}.
	\label{Bnorminequality4}\\
	\lVert Z_i\rVert_{B_i} 
	&\le \sqrt{K_{i,\mu}} \Vert Z_i \Vert_{2 \rightarrow 2} 
	\label{Bnorminequality1} \\
	\lVert Z\rVert_B
	& \le  \sqrt{ \sum_{i=1}^{r}   \Vert Z_i \Vert^2_{B_{i}}}
	\le  \sqrt{K_{\mu}} \Vert Z \Vert_{F} \label{Bnorminequality3}
	\end{align}
\end{lemma}
\begin{proof}
In order to prove \eqref{Bnorminequality1} note that for $ Z_i \in \mathbb{C}^{K_i} $ and $ \ell \in \lbrack L \rbrack $ due to the definition of $ K_{i,\mu} $
\begin{equation*}
	\Vert Z^*_i   b_{i,\ell}  \Vert^2_{\ell_2} \le  
	\Vert Z_i \Vert^2_F  \Vert b_{i,\ell} \Vert^2_{\ell_2} \overset{(\ref{def:Kmu})}{\le}  \frac{K_{i,\mu}}{L} \Vert Z_i \Vert^2_{2 \rightarrow 2}.
\end{equation*}
Taking the maximum over all $ \ell \in \lbrack L \rbrack $ shows (\ref{Bnorminequality1}). Inequality (\ref{Bnorminequality4}) can be proven analogously. (\ref{Bnorminequality3}) follows from
\begin{equation*}
	\Vert Z \Vert^2_{B} \le L\,
	\sum_{i=1}^{r} \underset{\ell \in \lbrack L \rbrack }{\max} \Vert
	Z^*_i b_{i, \ell} \Vert^2_{\ell_2}
	=\sum_{i=1}^{r} \Vert Z_i \Vert^2_{B_{i}}
\end{equation*}	
combined with (\ref{Bnorminequality1}) and the definition of $\Vert Z \Vert_F$.
\end{proof}
The notion of $\Vert \cdot \Vert_{B}$-norms together with Theorem \ref{KMR} allows us to state the following abstract isometry result, where we will use the notation $ d_B \left( \mathcal{X} \right) = \underset{X \in \mathcal{X}}{\sup}  \Vert X \Vert_{B}$.
\begin{proposition}\label{theorem3}
	Let $\mathcal{X}= - \mathcal{X} \subset \mathcal{M}$ be a symmetric set and consider
	\begin{align*}
	\widehat{E}&= \frac{\gamma_2 \left( \mathcal{X}, \Bnorm{\cdot} \right)}{\sqrt{Q}} \left(  \frac{\gamma_2 \left( \mathcal{X}, \Bnorm{\cdot}  \right)}{\sqrt{Q}} + d_{F} (\mathcal{X}) \right)  \\
	\widehat{V}&=  \frac{d_{B} \left( \mathcal{X} \right)}{\sqrt{Q}} \left(\frac{\gamma_2 \left( \mathcal{X}, \Bnorm{\cdot}  \right)}{\sqrt{Q}} +   d_{F} (\mathcal{X}) \right)  \\
	\widehat{U}&= \frac{1}{Q}  d^2_{B} \left( \mathcal{X} \right).
	\end{align*}
	Then, for $t>0$ and all $ p \in \lbrack P \rbrack $, 
	\begin{align}
	&\mathbb{P} \left( \underset{X \in \mathcal{X}}{\sup} \Big\vert  \frac{L}{Q} \Vert \mathcal{A}^p \left( X \right) \Vert_{\ell_2}^2 - \sum_{i=1}^{r} \Vert T^{1/2}_{i,p} X_i \Vert^2_F  \Big\vert \ge \tilde{c}_1 \widehat{E} + t \right) \le 2 \exp \left( - \tilde{c}_2 \min \left( \frac{t^2}{\widehat{V}^2}, \frac{t}{\widehat{U}} \right) \right)\label{ineq:RIP1}\\
	&\mathbb{P} \left( \underset{X \in \mathcal{X}}{\sup} \Big\vert   \Vert \mathcal{A} \left( X \right) \Vert_{\ell_2}^2 -  \Vert  X \Vert^2_F  \Big\vert \ge \tilde{c}_3 \widehat{E} + t \right) \le 2 \exp \left( - \tilde{c}_4 \min \left( \frac{t^2}{\tilde{V}^2}, \frac{t}{\widehat{U}} \right) \right), \label{ineq:RIP2}
	\end{align}
	provided $ \left\{ \Gamma_p \right\}^P_{p=1} $ is a $\omega$-admissible partition of $ \lbrack L \rbrack $. The constants $ \tilde{c}_1 $, $ \tilde{c}_2 $, $ \tilde{c}_3$, and $ \tilde{c}_4$ are universal.
\end{proposition}

\begin{proof}
%In the following we will use the index set $ \left\{ l_1, \cdots, l_Q  \right\}  = \Gamma_p $.
We will start by proving (\ref{ineq:RIP1}). Fix $ p \in  \lbrack P \rbrack $. For $X= \left(X_1, \ldots, X_r \right)  \in \mathcal{X}$ let $H_{X} \in \C^{L \times Q \sum_{i=1}^{r}  N_i }$ be the block diagonal matrix, whose diagonal elements, indexed by $ \ell \in \Gamma_p $ are given by row vectors of the form $ \sqrt{\frac{L}{Q}} \left(   b^*_{1, \ell}  X_1, \ldots,  b^*_{r,\ell}  X_r  \right) $. Furthermore, set $\mathcal{H}_{X} = \lbrace H_X: X \in \mathcal{X} \rbrace $. Observe that
\begin{align}
 \Vert H_X \Vert_F^2 &= \frac{L}{Q} \sum_{\ell \in \Gamma_p} \sum_{i=1}^{r} \Vert X_{i}^*  b_{i,\ell }\Vert_{\ell_2}^2 = \sum_{i=1}^{r} \trace \left( X_i X^*_i T_{i,p}  \right) = \sum_{i=1}^{r} \Vert  T^{1/2}_{i,p} X_i \Vert^2_{F}, \label{eq:inline33} \\
\Vert H_X \Vert_{2 \rightarrow 2} &= \sqrt{\frac{L}{Q}} \max_{\ell \in \Gamma_p} \Vert \left(   b^*_{1, \ell}  X_1, \ldots,  b^*_{r,\ell}  X_r  \right)\Vert_{\ell_2} \le \frac{1}{\sqrt{Q}} \Bnorm{X}. \label{eq:inline34}
\end{align}
Let $\xi^{\left(p\right)}$ be the concatenation of all the random bases vectors $c_{i,\ell }$, where $i \in \lbrack r \rbrack, \ell  \in  \Gamma_p $. Then 
%Let $\xi \in \mathbb{R}^{\left(\sum_{i=1}^{r} N_i \right)Q} $ be the vector of all the random bases, i.e., the concatenation of all the $c_{i,\ell }$, where $i \in \left\{ 1, \ldots, r \right\}, \ell  \in  \Gamma_p $. Then
\begin{align*}
\frac{L}{Q}\Vert \mathcal{A}^p \left(X\right) \Vert_{\ell_2}^2 = \frac{L}{Q} \sum_{\ell \in \Gamma_p} \vert  \mathcal{A}^p \left(X\right) \left( \ell \right) \vert^2=\frac{L}{Q} \sum_{ \ell \in \Gamma_p} \Big\vert \sum_{i=1}^{r}   b_{i, \ell }^* X_i c_{i, \ell } \Big\vert^2 = \Vert H_X \xi^{\left(p\right)} \Vert_{\ell_2}^2 
\end{align*}
and
\begin{align*}
%\frac{L}{Q} \mathbb{E} \left[ \Vert \mathcal{A}^p \left(X\right) \Vert^2_{\ell_2}  \right]   = \Vert H_X \Vert_F^2 = \frac{1}{Q} \sum_{l \in \Gamma_p} \sum_{i=1}^{r} \Vert X_{i}^*  b_{i,l}\Vert_{\ell_2}^2 = \sum_{i=1}^{r} \trace \left( X_i X^*_i T_{i,p}  \right) = \sum_{i=1}^{r} \Vert  T^{1/2}_{i,p} X_i \Vert^2_{F} .
 \sum_{i=1}^{r} \Vert  T^{1/2}_{i,p} X_i \Vert^2_{F} = \Vert H_X \Vert_F^2  = \mathbb{E} [ \Vert H_X \xi^{\left(p\right)} \Vert^2_{\ell_2}  ] .  %=\frac{L}{Q} \mathbb{E} \left[ \Vert \mathcal{A}^p \left(X\right) \Vert^2_{\ell_2}  \right] 
\end{align*}
Consequently
\begin{align*}
\underset{X \in \mathcal{X}}{\sup}\Big\vert  \frac{L}{Q} \Vert \mathcal{A}^p \left( X \right) \Vert_{\ell_2}^2 - \sum_{i=1}^{r} \Vert T^{1/2}_{i,p}  X_i \Vert^2_F  \Big\vert =  \underset{X \in \mathcal{X}}{\sup} \Big\vert \Vert H_X \xi^{\left(p\right)} \Vert_{\ell_2}^2  - \mathbb{E} \left[ \Vert H_X \xi^{\left(p\right)} \Vert_{\ell_2}^2 \right]  \Big\vert 
\end{align*}
%So Theorem \ref{KMR} applies. The result follows 
and inequality (\ref{ineq:RIP1}) follows from Theorem \ref{KMR}, equation (\ref{eq:inline33}), (\ref{eq:inline34}) combined with the fact that  $ \sum_{i=1}^{r} \Vert  T^{1/2}_{i,p} X_i \Vert^2_{F} \overset{(\ref{partitionequation})}{\le} 2 \Vert X \Vert^2_F $. Inequality (\ref{ineq:RIP2}) follows in an analogous way by letting $H_X$ be the block diagonal matrix, whose diagonal elements, indexed by $ \ell \in \lbrack L \rbrack $, are given by $ \left(   b^*_{1, \ell}  X_1, \ldots,  b^*_{r,\ell}  X_r  \right) $. Furthermore, one uses $ \sum_{\ell =1}^{L} b_{i, \ell} b^*_{i,\ell} = \Id $ instead of $ \frac{L}{Q}  \sum_{\ell \in \Gamma_p} b_{i,\ell} b^*_{i,\ell} = T_{i,p} $.
\end{proof}
Our strategy to prove Theorem \ref{localisometry} will now be to apply Proposition \ref{theorem3} with appropriately chosen sets $\mathcal{X}$. For $\mathcal{T}_m$, $ \mathcal{T}_h$, and $ \mathcal{T}_{\mathcal{S}^p h}$ as in (\ref{equ:Tsubspaces}), define
\begin{align*}
B^m &= \left\{ X \in \mathcal{T}_m :  \Vert X \Vert_F =1  \right\} \\
B^h &= \left\{  X \in \mathcal{T}_h :  \Vert X \Vert_F =1   \right\}   \\
B^{\mathcal{S}^p h} &= \left\{    X \in \mathcal{T}_{\mathcal{S}^p h} :  \Vert X \Vert_F =1    \right\}
\end{align*}
and observe that in order to prove the $\delta$-local isometry property on $ \mathcal{T} $ it is enough to apply Proposition \ref{theorem3} to the set $\mathcal{W}$ defined by
\begin{equation}\label{definition:W}
\mathcal{W}= B^h  + B^m.
\end{equation}
Similarly, in order to prove the $\delta$-local isometry property on $ \mathcal{T}^p $ for $ p \in \lbrack P \rbrack $  it is enough to apply Proposition \ref{theorem3} to the set $\mathcal{W}^p$ defined by
\begin{equation}\label{definition:Wp}
\mathcal{W}^p= B^h + B^{\mathcal{S}^p h} + B^m.
\end{equation}
That is, it remains to estimate the $ \gamma_2$-functionals of $\mathcal{W}$ and $\mathcal{W}^p$ with respect to $ \Vert \cdot \Vert_B $.  By Dudley's inequality (\ref{dudley}) one can bound the $ \gamma_2$-functional by an integral involving covering numbers. To estimate those, we need the following technical lemmas.
\begin{lemma}\label{splittinglemma}
Let  $B^m$ be the above defined set. Then
\begin{align*}
	&N \left( B^m, \Bnorm{\cdot}, \varepsilon \right)\\ 
	&\le  N \left( B \left(0,1\right) \subset \mathbb{R}^r , \Vert \cdot \Vert_{\ell_2}, \frac{\varepsilon}{2  \sqrt{ K_{\mu} }}  \right) \prod_{i=1}^{r} N \left( B\left(0,1\right) \subset \mathbb{C}^{K_i} , \Binorm{\cdot}, \frac{\varepsilon}{2} \right).
\end{align*}
(By $B\left(0,1\right) $ we always denote the closed unit ball with respect to the $ \Vert \cdot \Vert_{\ell_2} $-norm.)
\end{lemma}
This lemma is actually a slight modification of \cite[Lemma 3.1]{candes2011tight}. For the convenience of the reader we have included a proof in Appendix \ref{appendixsplittinglemma}.
\begin{lemma}\label{maureyapplied}
	For all $i \in \lbrack r \rbrack$
	\begin{equation}\label{coveringboundbymaurey}
	\log N \left( B \left(0,1\right) \subset \mathbb{C}^{K_i} , \Binorm{\cdot}  ,\frac{\varepsilon}{2} \right) \lesssim  \frac{ K_{i,\mu} }{\varepsilon^2} \log \left( L \right)  .
	\end{equation}
\end{lemma}
\begin{proof}
%First, we define
%\begin{align*}
%\mathcal{U} = \left\{ \text{Re} \left(b_{i,l}\right): \ l \in \Gamma_{p}  \right\} \cup \left\{ \text{Im} \left(b_{i,l}\right): \ l \in \Gamma_{p}  \right\}
%\end{align*}	
%and note that for all $x \in \mathbb{R}^{K_i}$ we have $\Vert x \Vert_{B_{i,p}} \le \sqrt{2}\  \underset{u \in \mathcal{U}}{\max} \vert \langle u, x \rangle  \vert =  \sqrt{2} \Vert x \Vert_{\text{conv} \left( \mathcal{U} \right)^0 } $. This shows
%\begin{align*}
%\log N \left( B \left(0,1\right) \subset \mathbb{R}^{K_i} ,\Vert \cdot \Vert_{B_{i,p}} ,\frac{\varepsilon}{2} \right) & \le \log N \left( B \left(0,1\right)  ,\Vert \cdot \Vert_{  \text{conv} \left( \mathcal{U} \right)^0  } ,\frac{\varepsilon}{2 \sqrt{2} } \right) \\
%& \lesssim N \left( \text{conv} \left( \mathcal{U} \right), \Vert \cdot \Vert_{\ell_2} , cu   \right),
%\end{align*}
%where in the last line we have applied Theorem \ref{dualitymetricentropy} and made use of the fact that $ \text{conv} \left( \mathcal{U} \right)^{00}= \text{conv} \left( \mathcal{U} \right)$. 
Our goal is to apply Theorem \ref{dualitymetricentropy} to $\log N \left( B \left(0,1\right) \subset \mathbb{C}^{K_i} ,\Binorm{\cdot},\frac{\varepsilon}{2} \right)$. However, as $ \Binorm{\cdot} $ is a norm defined on a complex vector space we first need to transfer this setting into an appropriate real vector space framework.  For that goal we will use the isometric embedding $ P: \C^{K_i} \rightarrow \mathbb{R}^{2K_i} $ given by $ x = \left(x_1, \ldots, x_{K_i} \right) \in \C^{K_i} \mapsto \left(   \left( \text{Re } x\right)_1, \left( \text{Im } x\right)_1, \ldots, \left( \text{Re } x\right)_{K_i} , \left( \text{Im } x\right)_{K_i}  \right)  $. Furthermore, note that for all $x \in \mathbb{C}^{K_i } $
\begin{align}
\Vert x \Vert_{B_{i}} &= \sqrt{L}\  \underset{\ell \in \lbrack L \rbrack}{\max} \vert \langle x, b_{i, \ell} \rangle \vert = \sqrt{L}\ \underset{\ell \in \lbrack L \rbrack}{\max}  \sqrt{ \left( \text{Re }  \langle x, b_{i,\ell} \rangle \right)^2 +  \left( \text{Im }  \langle x, b_{i,\ell} \rangle \right)^2 } \label{ineq:778} \\
%&\le 2 \ \underset{ \ell \in \Gamma_p }{\max} \max \left\{ \big\vert \text{Re}  \langle x,   b_{i, \ell}  \rangle_{\ell_2}  \big\vert ; \big\vert \text{Im} \langle x,  b_{i, \ell} \rangle_{\ell_2} \big\vert   \right\}
& \le \sqrt{2L} \ \underset{\ell \in \lbrack L \rbrack}{\max} \max \left\{ \big\vert \text{Re}\  \langle x, b_{i, \ell} \rangle   \big\vert ; \big\vert \text{Im}\ \langle x, b_{i, \ell} \rangle  \ \big\vert  \right\}. \label{ineq779}
%& = 2 \sqrt{2} \ \underset{\ell \in \Gamma_p}{\max} \max \left\{  \big\vert  \langle \text{Re} \left(x\right) , \text{Re}  \left( b_{i, \ell} \right) \rangle_{\ell_2} \big\vert;  \big\vert \langle \text{Im}  \left(x\right), \text{Im}  \left( b_{i, \ell} \right)  \rangle_{\ell_2} \big\vert;  \big\vert  \langle \text{Re} \left(x\right), \text{Im}   \left( b_{i, \ell} \right) \rangle_{\ell_2} \big\vert;   \big\vert \langle \text{Im} \left(x\right), \text{Re}  \left( b_{i, \ell} \right) \rangle_{\ell_2} \big\vert  \right\}
\end{align}
Setting
\begin{align*}
u_{\ell} = \left( \left( \text{Re}\ b_{i, \ell}  \right)_1, - \left( \text{Im}\ b_{i, \ell}  \right)_1, \left( \text{Re}\ b_{i, \ell}  \right)_2, \ldots, - \left( \text{Im}\ b_{i, \ell}  \right)_{K_i-1} , \left( \text{Re}\ b_{i, \ell}  \right)_{K_i},  - \left( \text{Im}\ b_{i, \ell}  \right)_{K_i} \right)
\end{align*}
yields $   \text{Re} \left( \langle  x, b_{i,\ell} \rangle_{\ell_2} \right)   =  \langle Px, u_{\ell}  \rangle_{\ell_2}  $ for all $x \in \C^{K_i}$ and all $ \ell \in \lbrack L \rbrack$. Similarly, setting 
\begin{align*}
v_{\ell} = \left( \left( \text{Im}\ b_{i, \ell}  \right)_1,  \left( \text{Re}\ b_{i, \ell}  \right)_1, \left( \text{Im}\ b_{i, \ell}  \right)_2, \ldots,  \left( \text{Re}\ b_{i, \ell}  \right)_{K_i-1} , \left( \text{Im}\ b_{i, \ell}  \right)_{K_i},   \left( \text{Re}\ b_{i, \ell}  \right)_{K_i} \right)
\end{align*}
yields $ \text{Im} \left( \langle  x, b_{i,\ell}\rangle \right)  = \langle Px, v_{\ell} \rangle $  for all $x \in \C^{K_i}$ and all $  \ell \in \lbrack L \rbrack $. We define
\begin{equation*}
\mathcal{U} = \underset{\ell \in \lbrack L \rbrack}{ \bigcup } \left\{ u_{\ell}; v_{\ell} \right\}
\end{equation*}
and observe
\begin{align}\label{ineq:inline888}
\underset{u \in \mathcal{U}}{\max} \ \Vert u \Vert_{\ell_2}= \underset{ \ell \in \lbrack L \rbrack }{\max} \ \Vert b_{i,\ell} \Vert_{\ell_2} \le \sqrt{ \frac{ K_{i,\mu} }{L}  }.
\end{align}
By (\ref{ineq:778}, \ref{ineq779}) and the definition of $ \mathcal{U}  $ we obtain
\begin{equation}\label{ineq:inline777}
\Vert x \Vert_{B_{i}} \le \sqrt{2L} \ \underset{u \in \mathcal{U}}{\max}\ \langle Px, u \rangle=  \sqrt{2L} \underset{u \in \text{conv}\ \mathcal{U}}{\max}\ \langle Px, u \rangle    = \sqrt{2L} \Vert Px \Vert_{ \left( \text{conv}\ \mathcal{U} \right)^{\circ} }.
\end{equation}
(For the definition of $\Vert \cdot \Vert_{ \left( \text{conv} \ \mathcal{U} \right)^{\circ} } $ see Section \ref{sec:coveringnumbers}.) Inequality (\ref{ineq:inline777}) together with Theorem \ref{dualitymetricentropy} yields
\begin{align*}
\log N \left( B \left(0,1\right) \subset \mathbb{C}^{K_i} ,\Vert \cdot \Vert_{B_{i}} ,\frac{\varepsilon}{2} \right) & \le \log N \left( B \left(0,1\right) \subset \mathbb{R}^{2 K_i}  ,\Vert \cdot \Vert_{  \text{conv} \left( \mathcal{U} \right)^{\circ}  } ,\frac{\varepsilon}{2 \sqrt{2 L} } \right) \\
& \lesssim \log N \left( \text{conv} \left( \mathcal{U} \right), \Vert \cdot \Vert_{\ell_2} ,   \frac{ \tilde{c} \varepsilon }{\sqrt{L}}    \right),
\end{align*}
for some numerical constant $ \tilde{c} >0 $, due to $ \text{conv} \left( \mathcal{U} \right)^{\circ \circ} = \text{conv} \left( \mathcal{U} \right) $. In order to estimate this covering number from above we will use Lemma \ref{Maurey}. For that purpose let $M \in \mathbb{N}$ and assume  $ \left( u_1, \ldots , u_M  \right) \in \mathcal{U}^M $. By Jensen's inequality
\begin{align*}
\mathbb{E} \Big\Vert \sum_{m=1}^{M} \varepsilon_m u_m \Big\Vert_{\ell_2} \le \sqrt{ \mathbb{E} \Big\Vert  \sum_{m=1}^{M} \varepsilon_m u_m \Big\Vert^2_{\ell_2}  } = \sqrt{ \sum_{m=1}^{M} \Vert u_m \Vert^2_{\ell_2} }\le  \sqrt{M}~ \underset{u \in \mathcal{U}}{\max} \Vert u \Vert_{\ell_2} .
\end{align*}
%Thus, we may apply Lemma \ref{Maurey} with $A= \sqrt{ K_i \mu^2_{\max} } $.
Thus, by Lemma \ref{Maurey} applied with $ A= \underset{u \in \mathcal{U}}{\max} \Vert u \Vert_{\ell_2}$ we obtain
\begin{align*}
\log N \left( \text{conv} \left( \mathcal{U} \right), \Vert \cdot \Vert_{\ell_2}, \frac{\tilde{c} \varepsilon}{\sqrt{L}}   \right) \lesssim \frac{ L }{\varepsilon^2}  \underset{u \in \mathcal{U}}{\max} \Vert u \Vert^2_{\ell_2}     \log \vert \mathcal{U} \vert   \lesssim \frac{ K_{i,\mu} }{\varepsilon^2} \log L,
\end{align*}
where in the second inequality we have used (\ref{ineq:inline888}). This completes the proof.

\end{proof}
The previous two lemmas allow us to find an upper bound for the $ \gamma_2$-functional, which is needed to prove Theorem \ref{localisometry}.
\begin{lemma}\label{lemmagammabound}
%Set 
%\begin{equation}
%S= B_h  + B_m
%\end{equation}
%and for $ p \in \lbrack P \rbrack $
%\begin{equation}\label{Sdefinition}
%S^p= B_h + B_{S_p h} + B_m.
%\end{equation}
Suppose that $ \mathcal{X}  = \mathcal{W}$ or $  \mathcal{X} = \mathcal{W}^p$ for some $ p \in \lbrack P \rbrack $. (For the definition of $\mathcal{W}$ and $ \mathcal{W}^p $ see (\ref{definition:W}) and (\ref{definition:Wp}).) Then
\begin{align*}
	d_{F} \left( \mathcal{X} \right)&\le 3,  \\
	d_{B}\left( \mathcal{X} \right)&\le 3 \sqrt{ K_{\mu} },\\
	\gamma_2 \left( \mathcal{X},  \Bnorm{\cdot} \right) &\lesssim  \sqrt{r \left(   K_{\mu}   \log \left(L\right) \log^2 \left( K_{\mu} \right)  + N \mu^2_h \right)}.
\end{align*}
\end{lemma}

\begin{proof}
	The first inequality follows from the triangle inequality. For the second one note that for $X \in \mathcal{X} $ by (\ref{Bnorminequality3}) one obtains the inequality 
	\begin{equation*}
	\Bnorm{ X } \le  \sqrt{ K_{\mu} } \Vert X \Vert_F \le 3 \sqrt{ K_{\mu} }.
	\end{equation*}
	The last line is more involved. We will present the proof only in the case of $ \mathcal{X}= \mathcal{W}^p $. If $\mathcal{X} = \mathcal{W} $ the inequality can be proven analogously. By Lemma \ref{gamma_2sum} we obtain
	\begin{equation}\label{gammaintern3}
          \gamma_2 \left(\mathcal{W}^p,  	\Bnorm{\cdot} \right)  \lesssim  \  \gamma_2 \left(B^h, 	\Bnorm{\cdot} \right) + \gamma_2 \left(B^{\mathcal{S}_p h},\Bnorm{\cdot} \right)   + \gamma_2 \left(B^m, \Bnorm{\cdot}\right) .
	\end{equation}
	We will estimate the three $\gamma_2$-functionals separately.\\
    \textbf{Step 1:} To bound $ \gamma_2 \left(B^h, \Bnorm{\cdot} \right) $, let $U= \left( h_1 u^*_1, \ldots, h_r u^*_r \right), V= \left( h_1 v^*_1, \ldots, h_r v^*_r \right) \in B^h  $. Observe that by definition
        \begin{align*}
          \Bnorm{U - V} &= \underset{\ell \in \lbrack L \rbrack}{\max}  \sqrt{ L\ \sum_{i=1}^{r} \Big\Vert \left(h_i u^*_i -h_i v^*_i\right)^*b_{i, \ell}  \Big\Vert_{\ell _2}^2}   = \underset{\ell \in \lbrack L \rbrack}{\max}  \sqrt{ L\ \sum_{i=1}^{r} \Vert u_i - v_i \Vert_{\ell_2}^2  \vert h^*_i b_{i,\ell}  \vert^2} \\
                        & \le \mu_h   \sqrt{ \sum_{i=1}^{r} \Vert u_i - v_i \Vert_{\ell_2}^2  }
                        %	&  \le  \underset{1 \le i \le r}{\max} \mu_i  \sqrt{ \sum_{i'=1}^{r} \Vert u_{i'} - v_{i'} \Vert_{\ell_2}^2 }\\
          = \mu_h \Vert U-V \Vert_{F},
        \end{align*}
        where the last equality is due to $\Vert h_i \Vert_{\ell_2}=1 $ for all $ i \in \lbrack r \rbrack $. This implies
	\begin{align}\label{step1}
          \gamma_2 \left(B^h, \Bnorm{\cdot} \right) &\le  \mu_h   \gamma_2 \left(B^h, \Vert \cdot \Vert_{F}\right) 
                                                      \lesssim   \mu_h   \int_{0}^{1} \sqrt{ \log N \left( B^h  ,\Vert \cdot \Vert_{F}, \varepsilon \right) }  d \varepsilon\lesssim   \mu_h \sqrt{rN},
	\end{align}
	where the second inequality follows from the Dudley inequality
        (\ref{dudley}).
        The third inequality follows from the fact that $ \left( B^h,
          \Vert \cdot \Vert_{F}  \right) $ is isometric to $ \left( B \left(0,1\right) \subset \mathbb{R}^{2 \sum_{i=1}^{r} N_i}, \Vert \cdot \Vert_{\ell_2}   \right) $
        %Euclidean ball 
        %with dimension $ 2\sum_{i=1}^{r}N_i $ and radius $1$ 
        and from a standard volumetric estimate.\\
        \textbf{Step 2:} To bound $ \gamma_2 \left(B^{\mathcal{S}_p h}, \Bnorm{\cdot} \right) $ let $U= \left( S_{1,p} h_1 u^*_1, \ldots, S_{r,p} h_r u^*_r\right)$ and\\
         $V= \left( S_{1,p} h_1 v^*_1, \ldots, S_{r,p} h_r v^*_r \right) \in B_h $. Then
	\begin{align*}
          \Bnorm{U - V} &= \underset{\ell \in \lbrack L \rbrack}{\max}  \sqrt{ L\ \sum_{i=1}^{r} \Big\Vert \left(S_{i,p} h_i  u^*_i - S_{i,p}h_i v^*_i\right)^* b_{i, \ell}  \Big\Vert_{\ell _2}^2} \\
                        & = \underset{\ell \in \lbrack L \rbrack}{\max}  \sqrt{ L\ \sum_{i=1}^{r} \Vert u_i - v_i \Vert_{\ell_2}^2  \vert h^*_i  S_{i,p} b_{i,\ell}  \vert^2}  \le \mu_h   \sqrt{ \sum_{i=1}^{r} \Vert u_i - v_i \Vert_{\ell_2}^2  }\\
                        & = \mu_h \sqrt{  \sum_{i=1}^{r} \Vert u_i - v_i \Vert^2_{\ell_2} \Vert h_i \Vert^2_{\ell_2}  }= \mu_h \sqrt{  \sum_{i=1}^{r} \Vert u_i - v_i \Vert^2_{\ell_2} \Vert T_{i,p} S_{i,p} h_i \Vert^2_{\ell_2}  }  \\
                        & \le \left(1+ \nu \right)  \mu_h \sqrt{ \sum_{i=1}^{r}  \Vert u_i - v_i \Vert^2_{\ell_2} \Vert S_{i,p} h_i \Vert^2_{\ell_2}   }
                        %	&  \le  \underset{1 \le i \le r}{\max} \mu_i  \sqrt{ \sum_{i'=1}^{r} \Vert u_{i'} - v_{i'} \Vert_{\ell_2}^2 }\\
          \lesssim \mu_h \Vert U-V \Vert_{F}.
        \end{align*}
        % We have used in the last line that $ 1= \Vert h_i \Vert_{\ell_2} = \Vert T_{i,p} S_{i,p} h_i \Vert_{\ell_2} \le \Vert T_{i,p} \Vert_{2 \rightarrow 2}  \Vert S_{i,p} h_i \Vert_{\ell_2} \le (1+\nu) \Vert S_{i,p} h_i \Vert_{\ell_2} $ and that $ \nu= \frac{1}{32} $.  
        In the third line we used that $\Vert h_i \Vert_{\ell_2} =1 $ and in the last line we used that $ \Vert T_{i,p} \Vert_{2 \rightarrow 2} \le 1 + \nu $ and $\nu = \frac{1}{32} $. An analogous reasoning as in (\ref{step1}) then yields
        \begin{align*}
          \gamma_2 \left(B^{S_p h}, \Bnorm{\cdot} \right)  \lesssim   \mu_h \sqrt{rN}.
        \end{align*}
        \textbf{Step 3:} To bound $ \gamma_2 \left(B^m, \Bnorm{\cdot}  \right) $ note that  inequality (\ref{dudley}) and the fact that $ d_B \left(B^m\right) \le \sqrt{ K_{\mu} } $ imply
        \begin{align*}
          \gamma_2 \left( B^m, \Bnorm{\cdot} \right) \lesssim  \int_{0}^{\sqrt{  K_{\mu}  }} \sqrt{  \log N \left( B_m  ,\Bnorm{\cdot} , \varepsilon \right) } d\varepsilon.
        \end{align*}
        Thus, by Lemma \ref{splittinglemma}	
        \begin{align}\label{gammaintern5}
        \begin{split}
          \gamma_2 \left( B^m,\Bnorm{\cdot} \right)& \lesssim    \int_{0}^{\sqrt{ K_{\mu}  }}  \sqrt{ \log N \left( B \left(0,1\right) \subset \mathbb{R}^r , \Vert \cdot \Vert_{\ell_2}, \frac{\varepsilon}{2\sqrt{ K_{\mu}  }}  \right) } d  \varepsilon   \\
          & +   \int_{0}^{\sqrt{ K_{\mu} }}  \sqrt{ \sum_{i=1}^{r}  \log \left( N \left( B (0,1) \subset \mathbb{C}^{K_i} , \Binorm{\cdot}, \frac{\varepsilon}{2} \right) \right)}  d\varepsilon.\\
          & \le    \int_{0}^{\sqrt{ K_{\mu}  }}  \sqrt{ \log N \left( B \left(0,1\right) \subset \mathbb{R}^r , \Vert \cdot \Vert_{\ell_2}, \frac{\varepsilon}{2\sqrt{ K_{\mu}  }}  \right) } d  \varepsilon \\
          & + \sqrt{r}  \int_{0}^{\sqrt{ K_{\mu}  }} \underset{i \in \lbrack r \rbrack }{\max} \sqrt{ \log \left( N \left( B (0,1) \subset \mathbb{C}^{K_i} , \Binorm{\cdot}, \frac{\varepsilon}{2} \right) \right)}  d\varepsilon. 
          \end{split}
	\end{align}
	The first integral can be bounded by
	\begin{align}\label{test}
	\begin{split}
	&\int_{0}^{\sqrt{ K_{\mu}  }}  \sqrt{ \log N \left( B \left(0,1\right) \subset \mathbb{R}^r , \Vert \cdot \Vert_{\ell_2}, \frac{\varepsilon}{2\sqrt{ K_{\mu} }}  \right) } d  \varepsilon\\ \le &\sqrt{r} \int_{0}^{\sqrt{ K_{\mu}  }} \sqrt{ \log \left(1 + \frac{4 \sqrt{ K_{\mu}  }}{\varepsilon}\right)} d\varepsilon 
	 \lesssim  \sqrt{r K_{\mu} },
	 \end{split}
	\end{align}
where we have used a volumetric estimate and a change of variables. In order to deal with the second term we will split the integrals into two parts: For small $\varepsilon$ we will use a volumetric estimate and for large $ \varepsilon $ we will apply Lemma \ref{maureyapplied}. First we consider the case that $ \varepsilon \in \left(0,1\right) $. Therefore, note that 
\begin{equation*}
B \left(0,1\right) \subset  \sqrt{ K_{i,\mu} } B_{\Vert \cdot \Vert_{B_{i}}} \left( 0,1 \right)  := \left\{ x \in \C^{K_i}: \ \Binorm{x} \le \sqrt{ K_{i,\mu} }  \right\}
\end{equation*}
by inequality (\ref{Bnorminequality4}). This fact combined with a volumetric estimate yields
	\begin{align*}
	%	N \left( S^{K_i-1} ,\Vert \cdot \Vert_{B_{i,p}} ,\varepsilon \right) &\le N \left( B_{\Vert \cdot \Vert_{B_{i,p}}} \left(0 \right) , \Vert \cdot \Vert_{B_{i,p}} , \frac{\varepsilon}{\sqrt{K \mu^2_{\max} }}  \right) \le \left(1 +\frac{2 \sqrt{K  \mu^2_{\max} }}{\varepsilon} \right)^{K_i}.	
	\underset{i \in \lbrack r \rbrack }{\max}\ N \left( B \left(0,1\right) \subset \mathbb{C}^{K_i} ,\Binorm{\cdot} ,\varepsilon \right) &\le \underset{i \in \lbrack r \rbrack }{\max}\ N \left( B_{ \Binorm{\cdot} }\left(0,1\right)  , \Binorm{\cdot} ,  \frac{\varepsilon}{ \sqrt{ K_{i,\mu} } }  \right)  \\
	&\le \left(1 +\frac{2 \sqrt{ K_{\mu}  }}{\varepsilon} \right)^{2K}.	
	\end{align*}
By a change of variables and an elementary integral inequality (see \cite[Lemma C.9]{FR2013}) this implies
%	\begin{align*}
%		 \int_{0}^{\alpha} \sqrt{ \log   N \left( S^{K-1} ,\Vert \cdot \Vert_B , \frac{\varepsilon}{2} \right) } d\varepsilon & \le \sqrt{K} \int_{0}^{\alpha} \sqrt{ \log \left( e \left(1+ \frac{4 \sqrt{K}}{\varepsilon} \right) \right)}  d\varepsilon \\
%		& \le \sqrt{K} \alpha \sqrt{\log \left( e \left(  1 + \frac{4 \sqrt{K}}{\alpha} \right) \right) }
%	\end{align*}
\begin{align*}
	\int_{0}^{1} \underset{i \in \lbrack r \rbrack}{\max} \ \sqrt{ \log   N \left( B \left(0,1\right) ,\Binorm{\cdot} , \frac{\varepsilon}{2} \right) } d\varepsilon &\le \sqrt{2K}  \int_{0}^{1} \sqrt{\log \left(1 + \frac{2 \sqrt{ K_{\mu} }}{\varepsilon}  \right)}   d\varepsilon   \\
	& \le \sqrt{2 K \log \left( e \left(  1 + 2 \sqrt{ K_{\mu} } \right) \right) }.
\end{align*}
Next, we are going to deal with the case that $ \varepsilon \in \left(1, \sqrt{ K_{\mu} }\right) $. Using Lemma \ref{maureyapplied} we get
\begin{align*}
\int_{1}^{\sqrt{ K_{\mu} }} \underset{i \in \lbrack r \rbrack}{\max} \ \sqrt{ \log \left( N \left( B \left(0,1\right),\Binorm{\cdot}, \frac{\varepsilon}{2} \right) \right)}  d\varepsilon  &\lesssim  \int_{1}^{\sqrt{ K_{\mu} }} \frac{ \sqrt{ K_{\mu}  \log L }}{\varepsilon}   d\varepsilon\\	
 &  \lesssim \sqrt{ K_{\mu}   \log  L  } \log \left(  K_{\mu} \right)   .
\end{align*}
%Adding the integral inequality for $ \varepsilon \in \left(0, 1\right) $ and the integral inequality for $ \varepsilon \in \left(1, \sqrt{K \mu^2_{\max}}\right) $ yields
Summing up the two integral inequalities yields
\begin{align*}
&\sqrt{r}~\underset{i \in \lbrack r \rbrack}{\max} \int_{0}^{\sqrt{ K_{\mu}  }} \sqrt{ \log \left( N \left( B (0,1) \subset \mathbb{C}^{K_i} , \Binorm{\cdot}, \frac{\varepsilon}{2} \right) \right)}  d\varepsilon\\
\lesssim &  \sqrt{r K_{\mu}  \log \left( L \right) } \log \left(  K_{\mu}  \right).
\end{align*}
This inequality together with (\ref{gammaintern5}) and (\ref{test}) shows that
%	\begin{align*}
%	\int_{0}^{\sqrt{K_i}} \sqrt{ \log \left( N \left( S^{K_i-1}, \Vert \cdot \Vert_{B_{i,p}}, \frac{\varepsilon}{2} \right) \right)}  d\varepsilon  &\le \sqrt{ K_i \log \left( e \left(  1 + 4 \sqrt{K_i} \right) \right) }\\
%	& +C \sqrt{K_i  \log \left( Q \right)  \log^3 \left( K_i \right) }  .
%	\end{align*}
%	Set $\alpha=1 $ and observe
	\begin{align*}
	\gamma_2 \left( B^m, \Bnorm{\cdot} \right) &\lesssim   \sqrt{r K_{\mu}  \log \left(L\right)  \log^2 \left( K_{\mu}  \right)   }	.
	%\gamma_2 \left( S_m, \Vert \cdot \Vert_{B_p} \right) &\le  C \left(\sqrt{rK \log \left( K \right) }  + \sqrt{rK \log \left(Q\right) \log \left(K \right)} \left( \log K \right) + \sqrt{rK} \right) \\
	%& \le C \sqrt{rK \log \left(Q\right) \log \left(K\right) }.
	\end{align*}
	The result then follows from inequality (\ref{gammaintern3}).
\end{proof}
Combining the upper bounds for the $ \gamma_2$-functionals in the last lemma with the abstract isometry result Proposition \ref{theorem3} we are able to prove the main result in this section.
\begin{proof}[Proof of Theorem \ref{localisometry}]
%It is enough to prove (\ref{isometryproperty1}). By setting $P=1$ (or, equivalently, $Q=L$) (\ref{isometryproperty2]}) and using $T_{i,p}= \Id $ follows directly from (\ref{isometryproperty1}). \\
Fix $p \in \lbrack P \rbrack$. Using Lemma \ref{lemmagammabound} and choosing the constant $C_{\omega}$ in (\ref{intern3}) large enough we get for the quantities arising in Proposition \ref{theorem3} that $\widehat{E} \le \frac{\delta}{2\tilde{c}_1} $, $ \widehat{V} \le \frac{\delta}{\sqrt{ \tilde{c}_2 \omega \log L  }} $, and $ \widehat{U} \le \frac{\delta}{\tilde{c}_2 \omega \log L} $, where we have set $\mathcal{X} = \mathcal{W}^p $ (see (\ref{definition:Wp})) and $ \tilde{c}_i$ are the constants appearing in Proposition \ref{theorem3}. Thus inequality (\ref{ineq:RIP1}) of Proposition \ref{theorem3} for $t= \frac{\delta}{2}$ shows that (\ref{isometryproperty1}) holds with probability $1- \mathcal{O} \left(L^{-\omega}\right) $ for fixed $p$.\\  
In order to prove (\ref{isometryproperty2}) we may argue analogously (with $ \mathcal{X}= \mathcal{W}$ and $t= \frac{\delta}{2})$ and apply inequality (\ref{ineq:RIP2}) of Proposition \ref{theorem3}. Thus, (\ref{isometryproperty1}) holds with probability at least $1 - \mathcal{O} \left(L^{-\omega}\right) $. Replacing $ \omega$ by $\omega+1$ in the argument above and using a union bound argument one observes that (\ref{isometryproperty1}) and (\ref{isometryproperty2}) are satisfied for all $ p \in \lbrack P \rbrack$ with probability at least $ 1- (P+1) \mathcal{O}\left(  L^{-\omega -1} \right) = 1- \mathcal{O} \left( L^{-\omega} \right) $, which finishes the proof.

%By repeating the same argument with $\tilde{\alpha}= \alpha +1 $ and a union bound we get $ (\ref{isometryproperty1}) $ for all $ p \in \left\{1 ; \cdots ; P \right\} $ with probability $1 - \mathcal{O} \left( L^{-\alpha} \right) $.

\end{proof}

\subsection{Constructing the Dual Certificate}\label{subsec:dual:certificate}
\subsubsection{The Golfing Scheme}\label{subsec:golfing}
The goal of this section is to construct $Y \in \text{Range} \left( \mathcal{A}^* \right)$ such that the conditions (\ref{dualcertificatecondition1}) and (\ref{dualcertificatecondition2}) in Lemma \ref{dualcertificate} are fulfilled with high probability. The construction itself will make use of the Golfing Scheme, an iterative method which has been introduced in \cite{gross2011recovering} for the first time. We set
\begin{align*}
Y_0 &= 0 \\
Y_p&=Y_{p-1} + \frac{L}{Q}  \left( \mathcal{A}^p \right)^* \mathcal{A}^p \mathcal{S}^p \left( \text{sgn } \left( X^0 \right) -  \mathcal{P}_{\mathcal{T}} \left( Y_{p-1} \right) \right) \quad \text{for } p \in \lbrack P \rbrack. 
%&= \sum_{p'=1}^{p}  \sqrt{\frac{L}{Q}}   \left( \mathcal{A}^{p'} \right)^* \lambda_{p'-1}
\end{align*}

We will make use of the notation
\begin{align}\label{equ:wpdefinition}
W_p= \text{sgn} \left(X^0\right) - \mathcal{P}_{\mathcal{T}} \left(Y_{p}\right) \quad \text{ for } 0\le p \le P.
\end{align}
The individual components of $W_p$ will be denoted by $W_{i,p}$ for $ i \in \lbrack r \rbrack $, i.e., $W_p = \left( W_{1,p}, \ldots, W_{r,p} \right) $. Then the dual certificate will be given by 
\begin{equation*}
Y= Y_P = \sum_{p=1}^{P} \frac{L}{Q}  \left( \mathcal{A}^p \right)^* \mathcal{A}^p \mathcal{S}^p \left( W_{p-1}  \right).
\end{equation*}
Our Golfing Scheme is set up in the same way as in \cite{lingstrohmer}. In particular, they also use the operator $ \mathcal{S}^p $ as a corrector function as explained in the following remark.
\begin{remark}\label{Remark:Golfingmodified}
	The reason for the appearance of the operator $ \mathcal{S}^p$ is the following: Observe that
	\begin{equation*}
	\mathbb{E} \left[ \left( \mathcal{A}^p  \right)^* \mathcal{A}^p \left(X\right) \right] = \frac{L}{Q} \left( T_{i,p} X_1 , \ldots,    T_{r,p}  X_r    \right).
	\end{equation*}
	Recall that $T_{i,p}$ may only be approximately equal to the identity matrix (see (\ref{partitionequation})). Thus, $ \left( \mathcal{A}^p \ \right)^* \mathcal{A}^p $ is not necessarily
	an unbiased estimator. However,
	\begin{equation*}
	\mathbb{E} \left[ \frac{L}{Q} \left( \mathcal{A}^p  \right)^* \mathcal{A}^p \mathcal{S}^p \left(X\right) \right] = \frac{L}{Q} \left( T_{1,p} S_{1,p} X_1 , \ldots,    T_{r,p} S_{r,p}  X_r    \right) = \left(X_1, \ldots, X_r  \right)= X.
	\end{equation*}
	Thus, we get that $  \mathbb{E} \left[ \frac{L}{Q} \left( \mathcal{A}^p  \right)^* \mathcal{A}^p \mathcal{S}^p  \right] = \Id $. Note that $ \mathcal{S}^p \left(W_{p-1}  \right)$ is, in general, not an element of  the subspace $ \mathcal{T}$. However, due to definition of $ \mathcal{T}^p$ we observe that $ \mathcal{S}^p \left(W_{p-1}  \right) \in \mathcal{T}^p $. This is the reason why we require the operator $ \mathcal{A}^p$ to satisfy the $ \delta$-local isometry property not only on $ \mathcal{T} $,  but also on $ \mathcal{T}^p$.
\end{remark}

Let us check that $ Y \in \text{Range} \left( \mathcal{A}^* \right) $: Recall that the $\mathcal{A}^p \mathcal{S}^p  \left( W_{p-1}\right) $ is obtained by setting the vector $\mathcal{A} \mathcal{S}^p  \left( W_{p-1}\right)  $ zero in those components, which do not belong to $\Gamma_p$ (see Section \ref{section:partitioncoherence}). In particular, this implies that $ \left(\mathcal{A}^p\right)^*  \mathcal{A}^p \mathcal{S}^p  \left( W_{p-1}\right) = \mathcal{A}^*   \mathcal{A}^p \mathcal{S}^p  \left( W_{p-1}\right) $. Thus, setting
 \begin{equation}\label{hdefinition}
 z= \sum_{p=1}^{P} \mathcal{A}^p  \mathcal{S}^p \left( W_{p-1}  \right).
 \end{equation}
we get that $Y= \mathcal{A}^* z $. The vector $z$ will also be important when we prove an upper bound for the estimation error in the presence of noise. In the remaining part of the proof we will verify that $Y$ satisfies the conditions in Lemma \ref{dualcertificate} with the constants $ \alpha = \frac{1}{8 \gamma} $, $ \beta= \frac{1}{4} $, and $ \delta = \frac{1}{4} $ (cf. Remark \ref{dualcertificateremark}).

\subsubsection{Exponential Decay}
In this section we will verify condition (\ref{dualcertificatecondition1}) in Lemma \ref{dualcertificate}. In other words, we have to show that the quantity
\begin{equation*}
\Vert W_P \Vert_F =  \Vert \text{sgn} \left(X^0\right) -\mathcal{P_T} \left(Y\right) \Vert_F
\end{equation*}
is small enough. An important observation, which we will need in the proof, is that $W_0 = \text{sgn} \left(X^0 \right) $ and one has the recurrence relation
\begin{equation}\label{equ:Wpiteration}
W_p = W_{p-1} - \frac{L}{Q} \left( \mathcal{P}_{\mathcal{T}} \left( \mathcal{A}^p \right)^* \mathcal{A}^p  \mathcal{S}^p \right) \left( W_{p-1} \right)  \quad \text{for all } p \in \lbrack P \rbrack ,
\end{equation}
which is a direct of consequence of the definition of $W_p $ (see equation (\ref{equ:wpdefinition})). In Lemma \ref{thmexponentialdecay}, we will prove that $W_p$ decays exponentially fast. We will need the following rather technical inequalities.
\begin{lemma}
Let $ \nu = \frac{1}{32} $. For all $ i \in \lbrack r \rbrack $ and for all $ p \in \lbrack P \rbrack $ we have the inequalities
\begin{align}
	\Big\Vert \Id - T^{1/2}_{i,p} \Big\Vert_{2\rightarrow 2} &\le \frac{1}{32} \label{Trootinequality}\\
	\Big\Vert \left( \Id - \mathcal{S}^p \right) X \Big\Vert_F &\le \frac{1}{31} \Vert X \Vert_F \label{Sinequality1}\\
	\Big\Vert \mathcal{S}^p X  \Big\Vert_F  &\le  \frac{32}{31}  \Vert X \Vert_F  \label{Sinequality2} .
\end{align}
\end{lemma}
\begin{proof}
Inequality (\ref{Trootinequality}) follows directly from (\ref{partitionequation}) and the observation that the square-root shifts the eigenvalues of $T_{i,p}$ closer to one. The inequalities (\ref{Sinequality1}) and (\ref{Sinequality2}) follow from the observation that  for all $i \in \lbrack r \rbrack$, $p \in \lbrack P \rbrack $
\begin{align*}
  \Vert \Id - S_{i,p} \Vert_{2 \rightarrow 2} &= \max \left\{ 1-  \sigma_{\min} \left(S_{i,p}\right);  \sigma_{\max}\left(S_{i,p}\right) -1  \right\}\\
  &= \max \left\{ 1-  \sigma^{-1}_{\max} \left(T^{-1}_{i,p}\right);  \sigma^{-1}_{\min} \left(T^{-1}_{i,p}\right) -1 \right\}  \le \frac{1}{31}.
\end{align*}
\end{proof}
This allows us to prove the main lemma in this section.
\begin{lemma}\label{thmexponentialdecay}
Suppose that $\mathcal{A}^p$ satisfies the $\delta$-local isometry property on $ \mathcal{T}^p $ with $\delta = \frac{1}{32} $ for all $p \in \lbrack P \rbrack$. Then,  for all $p \in \lbrack P \rbrack $,
\begin{equation}\label{decay1}
	\Vert W_p \Vert_F \le 4^{-p} \sqrt{r}
\end{equation}
and, in particular, if $ P \ge  \frac{1}{2} \log \left( 8 \gamma \sqrt{r} \right) $,
\begin{equation}\label{decay2}
	\Vert \text{sgn} \left(X^0\right) - Y \Vert_F \le  \frac{1}{8 \gamma}.
\end{equation}
\end{lemma}
\begin{proof}
First notice that by (\ref{Trootinequality}) and the triangle inequality
\begin{align*}
\left(1 - \nu \right) \Vert X_i \Vert_F \le \big\Vert T^{1/2}_{i,p} X_i \big\Vert_F \le \left(1 + \nu \right) \Vert X_i \Vert_F
\end{align*}
for all $X_i \in \mathbb{C}^{K_i \times N_i} $. Thus, by the local isometry property (\ref{isometryproperty1})
\begin{align*}
\left( 1- \nu \right)^2  \left(1 - \delta \right) \Vert X \Vert^2_F \le \frac{L}{Q} \Big\Vert \mathcal{A}^p \left( X \right) \Big\Vert^2_{\ell_2} \le  \left(1 + \delta \right) \left( 1+ \nu \right)^2   \Vert X \Vert^2_F
\end{align*}
for all $X\in \mathcal{T}^p$. Together with $ \delta = \nu = \frac{1}{32} $ this implies
\begin{equation*}
\Big\vert \frac{L}{Q}  \Vert \mathcal{A}^p \left( X \right)  \Vert^2_{\ell_2} - \Vert X \Vert^2_F     \Big\vert \le \frac{1}{8} \Vert X \Vert^2_F  %\left(   \delta \left(1 + \nu \right)^2 + 2\nu + \nu^2 \right) \Vert X \Vert^2_F
\end{equation*}
for all $ X \in \mathcal{T}^p $, which in turn is equivalent to
\begin{equation}\label{schemeoperatorbound}
\Big\Vert \mathcal{P}_{\mathcal{T}^p} - \frac{L}{Q} \mathcal{P}_{\mathcal{T}^p} \left( \mathcal{A}^p \right)^* \mathcal{A}^p  \mathcal{P}_{\mathcal{T}^p} \Big\Vert_{F \rightarrow F} \le \frac{1}{8}, %   \delta \left(1 + \nu \right)^2 + 2\nu + \nu^2.
\end{equation}
where $ \mathcal{P}_{\mathcal{T}^p} $ denotes the orthogonal projection onto $ \mathcal{T}^p $. Now note that $ \Vert W_{p-1} - \mathcal{P}_{\mathcal{T}} \left(X\right) \Vert_F \le \Vert W_{p-1} - \mathcal{P}_{\mathcal{T}^p}  \left( X \right) \Vert_F $ for all $ X \in \mathcal{M}$ due to $ W_{p-1} \in \mathcal{T} $ and $\mathcal{T} \subset \mathcal{T}^p $. This fact together with (\ref{equ:Wpiteration}) implies that
\begin{align*}
\Vert W_p \Vert_F &\le  \Big\Vert W_{p-1} -  \left(  \frac{L}{Q} \mathcal{P}_{\mathcal{T}^p}  \left( \mathcal{A}^p \right)^* \mathcal{A}^p  \mathcal{S}^p  \right) \left( W_{p-1} \right) \Big\Vert_F\\
&= \Big\Vert W_{p-1} -  \left(  \frac{L}{Q} \mathcal{P}_{\mathcal{T}^p}  \left( \mathcal{A}^p \right)^* \mathcal{A}^p \mathcal{P}_{\mathcal{T}^p} \mathcal{S}^p  \right) \left( W_{p-1} \right) \Big\Vert_F,
\end{align*}
where in the second line we use that $ \mathcal{S}^p W_{p-1} \in \mathcal{T}^p $ by the definition of $ \mathcal{T}^p $ (see (\ref{definition:Tpsubspace})) and because of $ W_{p-1} \in \mathcal{T} $. Using this computation and (\ref{Sinequality1}), (\ref{Sinequality2}), (\ref{schemeoperatorbound}) we obtain
\begin{align*}
\Vert W_p \Vert_F &\le \Big\Vert  \left( \Id - \frac{L}{Q} \mathcal{P}_{\mathcal{T}^p}  \left( \mathcal{A}^p \right)^* \mathcal{A}^p \mathcal{P}_{\mathcal{T}^p} \right) \left(  \mathcal{S}^p  W_{p-1} \right)  \Big\Vert_F + \Big\Vert \left( \Id - \mathcal{S}^p \right) W_{p-1} \Big\Vert_F \\
& \le  \frac{1}{8}  \Vert \mathcal{S}^p W_{p-1} \Vert_F +  \frac{1}{16}  \Vert W_{p-1} \Vert_F\le \frac{1}{4} \Vert W_{p-1} \Vert_F.
\end{align*}
Thus, the previous estimate yields
\begin{align*}
\Vert W_p \Vert_F \le \left( \frac{1}{4} \right)^p \Vert W_0 \Vert_F = \left( \frac{1}{4} \right)^p \sqrt{r}.
\end{align*}
This shows ($\ref{decay1} $) and, in particular, we obtain $\Vert W_P \Vert_F \le  4^{-P} \sqrt{r} $. The assumption $P \ge \frac{1}{2} \log \left( 8 \gamma \sqrt{r} \right) $ and the definition of $W_P $ imply (\ref{decay2}), which finishes the proof.

\end{proof}

\subsubsection{Bounding the Operator Norm on $ \mathcal{T}^{\perp} $}
To apply Lemma \ref{dualcertificate} we need in addition to controlling the share of $Y$ in $ \mathcal{T} $ also a bound on $ \mathcal{T}^{\perp}_i$ for all $ i \in \lbrack r \rbrack $. For that, recall from \cite{lingstrohmer} that
\begin{align*}
\Big\Vert \mathcal{P}_{\mathcal{T}^{\perp}_i}\left(Y^P_i\right)   \Big\Vert_{2 \rightarrow 2} %&=  \Big\Vert \mathcal{P}_{\mathcal{T}^{\perp}_i}   \sum_{p=1}^{P} \left( \frac{L}{Q} \left(  \left( \mathcal{A}^{p}  \right)^*  \mathcal{A}^p \mathcal{S}^p \right) \left( W_{p-1} \right)  \right) \Big\Vert_{2 \rightarrow 2} \\
%&\le \sum_{p=1}^{P} \Big\Vert \mathcal{P}_{T^{\perp}_i} \left( \sqrt{\frac{L}{Q}} \left( \mathcal{A}^{p}  \right)^*  \lambda_{p}  \right) \Big\Vert_{2 \rightarrow 2} \\
&\le  \sum_{p=1}^{P} \Big\Vert \mathcal{P}_{\mathcal{T}^{\perp}_i} \left( \frac{L}{Q}  \left(  \left( \mathcal{A}^{p}  \right)^* \mathcal{A}^{p} \mathcal{S}^p \right) \left( W_{p-1} \right)   -W_{i,p-1} \right) \Big\Vert_{2 \rightarrow 2} \\
& \le  \sum_{p=1}^{P} \Big\Vert \frac{L}{Q}   \left( \left(  \mathcal{A}_i^{p}  \right)^* \mathcal{A}^{p} \mathcal{S}^p \right) \left( W_{p-1} \right)   -  W_{i,p-1}    \Big\Vert_{2 \rightarrow 2} = \sum_{p=1}^{P} \Vert W_{i,p} \Vert_{2 \rightarrow 2}  ,
\end{align*} 
where one uses the fact that $ W_{i,p-1} \in \mathcal{T}_i $. Thus to establish the bound $\Big\Vert \mathcal{P}_{\mathcal{T}^{\perp}_i}\left(Y^P_i\right)   \Big\Vert_{2 \rightarrow 2}  < \frac{1}{4}$ it remains to show that
   	\begin{equation*}
   \Big\Vert \frac{L}{Q} \left(   \left(  \mathcal{A}_i^{p}  \right)^* \mathcal{A}^{p} \mathcal{S}^p \right) \left(  W_{p-1} \right)   -  W_{i,p-1}   \Big\Vert_{2 \rightarrow 2} \le  \frac{1}{4^{p+1}}.
   \end{equation*}
To proceed, set for $ p \in \lbrace {0;1; \ldots; P-1} \rbrace $
\begin{equation}\label{mupdefinition}
\mu_p = \sqrt{L} \underset{\ell \in \Gamma_{p+1}, k \in \lbrack r \rbrack}{\max}  \Big\Vert  W_{k,p}^* S_{k,p+1} b_{k,\ell} \Big\Vert_{2 \rightarrow 2}.
\end{equation}
This allows us to state the following lemma.
\begin{lemma}\label{operatorbound}
	Fix $ i \in \lbrack r \rbrack $ and let $ \omega \ge 1 $. Assume that
	\begin{equation}\label{exponentialdecay2}
	\mu_p \le 4^{-p} \mu_h \text{ and } \Vert W_p \Vert_F \le 4^{-p} \sqrt{r}.
	\end{equation}
	If 
	\begin{align}\label{ineq:Qlowerboundoperator}
	Q \gtrsim_{\omega}  r \left(  K_{\mu} + N \mu^2_h  \right)  \left( \log L \right)^2 ,   
	\end{align}
	then with probability $1- \mathcal{O} \left(L^{-\omega} \right) $ the inequality
	\begin{equation}\label{operatornormbound}
	\Big\Vert  \frac{L}{Q}    \left(  \mathcal{A}_i^{p}  \right)^* \mathcal{A}^{p} \mathcal{S}^p W_{p-1}   - W_{i,p-1}  \Big\Vert_{2 \rightarrow 2} \le \frac{1}{4^{p+1}}
	\end{equation}
	is true for all $p \in \lbrack P \rbrack$ and for all $i \in \lbrack r \rbrack$ .
\end{lemma}

\begin{remark}\label{remarkdualcertificate2}
	The validity of assumption (\ref{exponentialdecay2}) is assured by Lemma \ref{thmexponentialdecay} and Lemma \ref{mudecay} below.
\end{remark}

\begin{proof}
The proof follows the same strategy as \cite[Lemma 5.12]{lingstrohmer}. Fix $ p \in\lbrack P \rbrack$ and $ i \in \lbrack r \rbrack   $. First, we will decompose $W_{i,p}$ as a sum of independent random matrices such that the matrix Bernstein inequality can be applied. For that purpose, observe that for all $ y \in \mathbb{C}^L $ and for all $ \ell \in \Gamma_p $ by definition of $ \mathcal{S}^p $ (Definition \ref{def:operatorSp}) and $ \mathcal{A}^p $
\begin{equation*}
\left( \mathcal{A}^p \mathcal{S}^p W_{p-1} \right) \left(\ell\right) =  \sum_{k=1}^{r} b^{*}_{k,\ell} S_{k,p} W_{k,p-1}  c_{k,\ell}.
\end{equation*}		
(For $ \ell \in \lbrack L \rbrack \backslash \Gamma_p $ the left-hand side is equal to zero as $\mathcal{A}^p \left( X \right) = P_{\Gamma_p} \left( \mathcal{A}
\left( X \right)  \right)$.) Using (\ref{equ:Aitransposed}) one obtains
\begin{equation*}
\left( \left( \mathcal{A}^p_i \right)^* \mathcal{A}^p \mathcal{S}^p \right) W_{p-1} = \sum_{\ell \in \Gamma_p} \sum_{k=1}^{r} b_{i,\ell} b^{*}_{k,l} S_{k,p} W_{k,p-1}  c_{k,\ell} c^*_{i,\ell}.
\end{equation*}
With $ S_{i,p} = T_{i,p}^{-1}$ and the definition of $T_{i,p}$ (see equation (\ref{partitionequation})) this implies
\begin{equation*}
W_{i,p-1} = T_{i,p} S_{i,p} W_{i,p-1} = \frac{L}{Q} \sum_{\ell \in \Gamma_p} b_{i,\ell}b^*_{i,\ell} S_{i,p} W_{i,p-1}.
\end{equation*}
In order to simplify notation we introduce the vectors $w_{k,\ell} $ defined by
\begin{equation}\label{definitionwkl}
w_{k,\ell} = W^*_{k,p-1} S_{k,p} b_{k,\ell}.
\end{equation}
Using this definition we may write	(as $S_{k,p}$ is self-adjoint)
\begin{align}
W_{i,p} =&\frac{L}{Q} \left( \left( \mathcal{A}^p_i \right)^* \mathcal{A}^p \mathcal{S}^p \right) W_{p-1}-W_{i,p-1}\label{equation:wipiteration1}\\
 = &\frac{L}{Q}   \sum_{\ell \in \Gamma_p} \sum_{k=1}^{r} b_{i,\ell} w^*_{k,\ell}  c_{k,\ell} c^*_{i,\ell} -   \frac{L}{Q} \sum_{\ell \in \Gamma_p} b_{i,\ell}w^*_{i,\ell}\\
= &\frac{L}{Q} \sum_{\ell \in \Gamma_p}  b_{i,\ell} w^*_{i,\ell} \left( c_{i,\ell} c^*_{i,\ell} - \Id  \right)     +\frac{L}{Q} \sum_{\ell \in \Gamma_p} \sum_{k \ne i} b_{i,\ell} w^*_{k,\ell}  c_{k,l} c^*_{i,\ell}= \sum_{\ell \in \Gamma_p} Z_{\ell}, \label{equation:wipiteration2}
\end{align}	
where we have set
\begin{align*}
Z_{\ell}= \frac{L}{Q} \left( \sum^{L}_{k = 1}  b_{i,\ell} w^*_{k,\ell} \left( c_{k,\ell} c^*_{i,\ell} - \mathbb{E} \left[ c_{k,\ell} c^*_{i,\ell}  \right]  \right)     \right).
\end{align*}
Note that until the last step of the proof $i$ is assumed to be fixed which is why we refrain from indicating the $i$-dependence in every step for reasons of notational simplicity. Observe that each summand of $Z_{\ell}$ and hence the the cross terms in $Z_{\ell} Z^*_{\ell} $ and $ Z^*_{\ell} Z_{\ell} $ have expectation zero. Thus using basic properties of circular symmetric normal random variables, Lemma \ref{usefullemma3} and Lemma \ref{usefullemma2} we compute
\begin{align}
	\mathbb{E}  \left[   Z_{\ell} {Z_{\ell}}^* \right] &=\frac{L^2}{Q^2}   \sum_{k=1}^{r} N_k  \big\Vert w_{k,{\ell}}\big\Vert_{\ell_2}^2 b_{i,{\ell}} b^*_{i,\ell}.\label{equ:variance1}\\
	\mathbb{E} \left[ Z^*_{\ell} Z_{\ell}  \right]&= \frac{L^2}{Q^2} \Vert b_{i,\ell} \Vert^2_{\ell_2}   \sum_{k=1}^{r}  \big\Vert  w_{k,\ell} \big\Vert_{\ell_2}^2  \Id.  \label{equ:variance2}
\end{align}	
We have to find an upper bound for the spectral norms of these quantities. First, observe that
\begin{align*}
	\Big\Vert  \sum_{\ell \in \Gamma_p} \mathbb{E} \left[ Z_{\ell} Z^*_{\ell} \right]  \Big\Vert_{2\rightarrow 2} & \le \frac{L^2 N}{Q^2} \left( \underset{k \in \lbrack r \rbrack, \ell \in \Gamma_p}{\max} \Vert w_{k,\ell} \Vert_2^2  \right) \Big\Vert \sum_{k=1}^{r}  \sum_{\ell \in \Gamma_p } b_{i,\ell} b^*_{i,\ell} \Big\Vert_{2 \rightarrow 2}    \\
	& \le  \frac{rN}{Q} \mu^2_{p-1} \Vert T_{i,p} \Vert_{2 \rightarrow 2}  \overset{\eqref{exponentialdecay2}}{\lesssim} \frac{16^{-p+1 }rN\mu^2_h}{Q}.
\end{align*}
By a similar computation we obtain
\begin{align*}
	\Big\Vert \sum_{\ell \in \Gamma_p}  \mathbb{E} \left[ Z^*_{\ell} Z_{\ell}  \right] \Big\Vert_{2 \rightarrow 2} & \le  \frac{L^2}{Q^2} \left( \underset{\ell \in \Gamma_p}{\max}  \Vert b_{i,\ell} \Vert^2_{\ell_2} \right) \sum_{k=1}^{r}  \sum_{\ell \in \Gamma_p} \Vert w_{k, \ell} \Vert^2_{\ell^2}  \\
	& \lesssim \frac{L K_{i,\mu}}{Q^2} \sum_{k=1}^{r}  \sum_{\ell \in \Gamma_p} \trace \left( W^*_{k,p-1} S_{k,p} b_{k,\ell} b^*_{k,\ell} S_{k,p} W_{k,p-1}  \right) \\
	& =   \frac{K_{i,\mu}}{Q}  \sum^r_{k=1}  \Vert S^{1/2}_{k,p} W_{k,p-1}  \Vert^2_{F} \lesssim   \frac{ K_{i,\mu} }{Q} \Vert W_{p-1} \Vert^2_{F}\le 16^{-p+1}  \frac{ r  K_{i,\mu}     }{Q}.
\end{align*}
Thus, we have obtained
\begin{equation}\label{ineq:sigmaupperbound}
\sigma^2 := \max \left\{  \Big\Vert \sum_{\ell \in \Gamma_p} \mathbb{E} \left[  Z^*_{\ell} Z_{\ell} \right] \Big\Vert_{2 \rightarrow 2}, \Big\Vert \sum_{\ell \in \Gamma_p} \mathbb{E} \left[  Z_{\ell} Z^*_{\ell} \right] \Big\Vert_{2 \rightarrow 2} \right\} \lesssim 16^{-p} \frac{  r }{Q} \max \left\{  K_{i,\mu} , N \mu^2_h  \right\}.
\end{equation}	
Observe that a lower bound for $ \sigma^2 $ is given by
\begin{align}
\sigma^2 &\ge \Big\Vert  \sum_{\ell \in \Gamma_p} \mathbb{E} \left[  Z^*_{\ell} Z_{\ell} \right]  \Big\Vert_{2 \rightarrow 2} = \frac{L^2}{Q^2} \sum_{k=1}^{r} \sum_{\ell \in \Gamma_p} \Vert b_{i,\ell} \Vert^2_{\ell_2} \Vert w_{k,\ell} \Vert^2_{\ell_2}\label{ineq:sigmalowerbound2}.
\end{align}	
Next we have to estimate $R=  \underset{\ell \in \Gamma_p}{\max} \Big\Vert \Vert Z_{\ell} \Vert_{2 \rightarrow 2}  \Big\Vert_{\psi_1} $. By Lemma \ref{lemma:orliczestimates} and inequality (\ref{ineq:Hoelderinequality}) we have that
\begin{align}
\Big\Vert \Vert Z_{\ell} \Vert_{2 \rightarrow 2} \Big\Vert_{\psi_1} &\le \frac{L}{Q} \left(     \sum_{k \ne i}  \Vert b_{i,\ell} \Vert_{\ell_2}  \Big\Vert \vert   w^*_{k,\ell}  c_{k,\ell} \vert \Vert c_{i,\ell} \Vert_{\ell_2} \Big\Vert_{\psi_1}  +  \Vert b_{i,\ell} \Vert_{\ell_2} \Big\Vert \Vert \left( c_{i,\ell} c^*_{i,\ell} - \Id \right) w_{i,\ell}  \Vert_{\ell_2} \Big\Vert_{\psi_1}      \right) \notag \\
& \lesssim \frac{L  \sqrt{N_i} }{Q} \Vert b_{i,\ell} \Vert_{\ell_2} \sum^r_{k=1}  \Vert w_{k,\ell} \Vert_{\ell_2} \label{ineq:lowerbound100}\\
& \lesssim \frac{r \sqrt{K_{i,\mu} N_i} \mu_{p-1} }{Q} \lesssim 4^{-p} \frac{r \sqrt{K_{i,\mu} N_i} \mu_h}{Q}  \lesssim 4^{-p} \frac{r \left( K_{i,\mu} + N_i \mu^2_h \right)}{Q} \notag
\end{align}
and, consequently, $ R  \lesssim 4^{-p} \frac{r \left( K_{i,\mu} + N_i \mu^2_h  \right)}{Q} $. Moreover, combining (\ref{ineq:sigmalowerbound2}) and (\ref{ineq:lowerbound100}) we obtain
\begin{align}\label{ineq:RoverSigma}
\frac{  \vert \Gamma_p \vert  R^2}{\sigma^2} &\lesssim QN  \frac{\underset{\ell \in \Gamma_p}{\max} \left( \sum_{k=1}^{r} \Vert b_{i,\ell} \Vert_{\ell_2} \Vert w_{k,\ell} \Vert_{\ell_2} \right)^2 }{ \underset{\ell \in \Gamma_p}{\max} \left( \sum_{k=1}^{r} \Vert b_{i,\ell} \Vert^2_{\ell_2} \Vert w_{k,\ell} \Vert^2_{\ell_2} \right) } \le QNr  .
\end{align}	
As $Q \le L $ by definition (\ref{ineq:Qlowerboundoperator}) implies that $\log \left( 1+ \frac{\vert \Gamma_p \vert R^2 }{\sigma^2}   \right)\lesssim \log L$. Thus, setting $t=  \left( \omega +2 \right) \log L $ we obtain from Theorem \ref{matrixbernstein} applied with $\alpha=1$ and combined with (\ref{ineq:sigmaupperbound}) that with probability $1- \mathcal{O} \left( L^{-\omega-2} \right) $
	\begin{align*}
	\Big\Vert  \sum_{\ell \in \Gamma_p} Z_{\ell} \Big\Vert_{2 \rightarrow 2} &\lesssim_{\omega} 4^{-p} \max \left\{ \sqrt{ \frac{ r \left(  K_{i,\mu}  +N \mu^2_h \right)} { Q}  \log L }   , \frac{r \left( K_{i,\mu}  + N \mu^2_h \right) }{Q}  \left( \log L \right)^2  \right\}.
	\end{align*}
	Thus, by choosing the constant in (\ref{ineq:Qlowerboundoperator}) large enough it holds that $ \Big\Vert  \sum_{\ell \in \Gamma_p} Z_{\ell} \Big\Vert_{2 \rightarrow 2} \le 4^{-p-1}  $ with probability $1- \mathcal{O} \left( L^{-\omega-2} \right) $ for fixed $p \in \left[ P \right] $ and for fixed $i \in \left[r\right]$. By taking the union bound over all $i \in \left[ r  \right] $ and over all $ p \in \left[ P \right]  $ we obtain that with probability $1- rP  \mathcal{O}  \left( L^{-\omega-2} \right) = 1 - \mathcal{O} \left( L^{-\omega} \right) $ equation (\ref{operatornormbound}) is true for all $p \in  \left[ P \right] $ and for all $ i \in \left[ r \right] $. This finishes the proof.
	
\end{proof}

\subsubsection{Proof that $\mu_p \le \frac{1}{4} \mu_{p-1} $}
Lemma \ref{operatorbound} additionaly required that $ \mu_{p} \le \frac{1}{4} \mu_{p-1}$ for all $ p \in \lbrack P-1 \rbrack$. In this section we will verify that this property holds with high probability.

\begin{lemma}\label{mudecay}
Let $ \omega \ge 1$. If
\begin{equation}\label{Qbound}
Q \gtrsim_{\omega} r \left(  K_{\mu}  + N \mu^2_h  \right) \log^2 L \ ,
\end{equation}
then with probability at least $1- \mathcal{O} \left( L^{-\omega} \right)$ it holds that $ \mu_p \le \frac{1}{4} \mu_{p-1}$ for all $ p \in  \lbrack P-1 \rbrack  $.
\end{lemma}
A similar lemma was established in \cite{lingstrohmer}. However, it was required that $L$ scales quadratically with $r$. Thus, we need to refine the argument in order to achieve a linear scaling in $r$. %This will be done with the help of the following well-known property, which can be proven in a straightforward way by a union bound.
\begin{proof}[Proof of Lemma \ref{mudecay}]
First, we will show the claim for fixed $ p \in \left\{ 0; 1; \ldots; P-1  \right\} $. Observe that it is enough to show that for all $\ell \in \Gamma_{p+1} $ and all $ i \in \left[ r \right] $
  \begin{equation}\label{verylastinequality}
\sqrt{L} \Vert   w_{i,\ell}\Vert_{\ell_2} \le \frac{1}{4} \mu_{p-1}
  \end{equation}
with $w_{i,\ell} := W_{i,p} S_{i,p+1} b_{i,\ell} $ as in (\ref{definitionwkl}).
%\begin{equation*}
%	w_{k,j}= W^*_{k,p-1} S_{k,p} b_{k,j}.
%\end{equation*}
Furthermore, observe that from the recurrence relation (\ref{equ:Wpiteration}) we obtain
\begin{equation*}
W_{i,p} = W_{i,p-1} - \frac{L}{Q} \left( \mathcal{P}_{\mathcal{T}_i} \left( \mathcal{A}_i^p \right)^* \mathcal{A}^p  \mathcal{S}^p \right) \left( W_{p-1} \right).
\end{equation*}
Due to the definition of $ \mathcal{T}_i $ and  $ \Vert h_i \Vert_{\ell_2} =  \Vert m_i \Vert_{\ell_2}  =1$ we may write for all $ Z \in \mathbb{C}^{K_i \times N_i} $
\begin{equation*}
    \mathcal{P}_{\mathcal{T}_i} Z = h_ih_i^* Z + \left( \Id - h_i h_i^*\right) Z m_im_i^*.
\end{equation*}
Together with (\ref{equation:wipiteration1}, \ref{equation:wipiteration2}) this implies
  \begin{align*}
  W_{i,p}= &\frac{L}{Q} \sum_{j \in \Gamma_p}  \Big\lbrack h_i h^*_i b_{i,j} w^*_{i,j} \left( \Id - c_{i,j}c^*_{i,j}  \right) + \left( \Id - h_i h^*_i\right) b_{i,j} w^*_{i,j} \left( \Id - c_{i,j} c^*_{i,j}  \right) m_im_i^* \Big\rbrack - \\
 &\frac{L}{Q} \sum_{k \ne i} \sum_{j \in \Gamma_p}  \Big\lbrack h_i^* h_i b_{i,j} w^*_{k,j} c_{k,j} c^*_{i,j}+ (\Id - h_ih_i^*)b_{i,j} w^*_{k,j} c_{k,j} c^*_{i,j} m_im_i^* \Big\rbrack.
   \end{align*}
   We define for all $ j \in \Gamma_p $
   \begin{align*}
     \textbf{z}_{i,j} &=  \frac{L}{Q}   \left( \Id - c_{i,j} c^*_{i,j}
                        \right)  w_{i,j}  b^*_{i,j} h_ih_i^* S_{i,p+1}  b_{i,\ell}, \\
     %z_{i,j}&=  \frac{L}{Q} b^*_{i,\ell} S_{i,p+1} \left( \Id - h_ih_i^* \right) b_{i,j} w^*_{i,j}  \left( \Id - c_{i,j} c^*_{i,j}  \right) m_i 
     z_{i,j} &=  \frac{L}{Q}  m^*_i  \left( \Id - c_{i,j} c^*_{i,j}  \right) w_{i,j} b^*_{i,j} \left( \Id - h_ih_i^* \right) S_{i,p+1}  b_{i,\ell} 
   \end{align*}
   and for all $k\neq i$ and for all $ j \in \Gamma_p $
   \begin{align*}
     \textbf{z}_{k,j}&= \frac{L}{Q}  c_{i,j} c^*_{k,j}   w_{k,j}  b^*_{i,j} h_ih_i^* S_{i,p+1} b_{i,\ell},\\
     %z_{k,j}&= \frac{L}{Q}  b^*_{i,\ell} S_{i,p+1} (\Id - h_ih_i^*)b_{i,j}  w^*_{k,j}   c_{k,j} c^*_{i,j} m_i. 
     z_{k,j}&= \frac{L}{Q} m^*_i c_{i,j} c^*_{k,j} w_{k,j} b^*_{i,j}   (\Id - h_ih_i^*) S_{i,p+1} b_{i,\ell}. 
   \end{align*}
%Hence, we may write
%   \begin{align*}
%     b^*_{i,\ell} S_{i,p+1} W_{i,p}=  \sum_{j \in \Gamma_p} \left[ \textbf{z}^*_{i,j} + z_{i,j} m^*_i \right]- \sum_{k \ne i} \sum_{j \in \Gamma_p} \left[ \textbf{z}^*_{k,j} + z_{k,j} m^*_k \right].
%   \end{align*}
Hence, to establish (\ref{verylastinequality}) by the triangle inequality it is sufficient to prove that with high probability
   \begin{align}
     \Big\Vert  \sum_{j \in \Gamma_p}  \textbf{z}_{i,j} \Big\Vert_{\ell_2} &\le \frac{1}{16 \sqrt{L} } \mu_{p-1}, \label{muestimate1}  \\
     \Big\vert  \sum_{j \in \Gamma_p} z_{i,j} \Big\vert &\le \frac{1}{16 \sqrt{L} } \mu_{p-1}, \label{muestimate2} \\
     \Big\Vert   \sum_{k \ne i} \sum_{j \in \Gamma_p}  \textbf{z}_{k,j} \Big\Vert_{\ell_2} &\le  \frac{1}{16 \sqrt{L} } \mu_{p-1},  \label{muestimate3} \\
     \Big\vert  \sum_{k \ne i} \sum_{j \in \Gamma_p}  z_{k,j}  \Big\vert &\le  \frac{1}{16 \sqrt{L} } \mu_{p-1}  . \label{muestimate4} 
   \end{align}   
   \textbf{Step 1: Proof of (\ref{muestimate1})} In order to apply Theorem \ref{matrixbernstein} we compute using Lemma \ref{usefullemma2}
   \begin{align*}
     \Big\Vert \mathbb{E} \left[  \sum_{j \in \Gamma_p} \textbf{z}_{i,j} \textbf{z}^*_{i,j} \right] \Big\Vert_{2 \rightarrow 2} = & \frac{L^2}{Q^2}  \vert h_i^*  S_{i,p+1} b_{i,\ell} \vert^2 \sum_{j \in \Gamma_p}  \vert  b^*_{i,j} h_i  \vert^2 \Vert w_{i,j} \Vert^2_{\ell_2}\\
    \le & \frac{1}{QL} \mu^2_h \mu^2_{p-1} \Vert T^{1/2}_{i,p} h_i \Vert_{\ell_2}^2 \lesssim \frac{1}{QL} \mu^2_h \mu^2_{p-1}.                             
   \end{align*}
	Analogously, using Lemma \ref{usefullemma3}
   \begin{align*}
     \mathbb{E} \left[ \sum_{j \in \Gamma_p}  \textbf{z}^*_{i,j}
     \textbf{z}_{i,j} \right]  &= \frac{L^2 N_i}{Q^2}  \sum_{j \in \Gamma_p} \Vert w_{i,j}\Vert_{\ell_2}^2 \vert b^*_{i,j} h_i \vert^2 \vert b^*_{i,\ell} S^*_{i,p+1} h_i \vert^2 \lesssim \frac{N_i}{QL}  \mu^2_{p-1} \mu_h^2 .
   \end{align*}
	Next, we estimate $R = \underset{j \in \Gamma_p}{\max} \Big\Vert \Vert \textbf{z}_{i,j} \Vert_{\ell_2 } \Big\Vert_{\psi_1} $. For that purpose we apply Lemma \ref{lemma:orliczestimates} to observe that
   \begin{align}
    R= \underset{j \in \Gamma_p}{\max}~\Big\Vert \Vert \textbf{z}_{i,j} \Vert_{\ell_2}
     \Big\Vert_{\psi_1} 
     &= \frac{L}{Q} \   \underset{j \in \Gamma_p}{\max} \ \left( \vert b^*_{i,j} h_i \vert \vert h^*_i S_{i,p+1} b_{i,\ell} \vert \Vert  \left( \Id - c_{i,j} c^*_{i,j} \right) w_{i,j}  \Vert_{\psi_1} \right) \notag \\
     & \lesssim  \frac{L \sqrt{N_i}}{Q} \   \underset{j \in \Gamma_p}{\max} \ \left( \vert h^*_i S_{i,p+1} b_{i,\ell} \vert  \vert b^*_{i,j} h_i \vert \Vert w_{i,j} \Vert_{\ell_2} \right) \label{ineq:intern9999} \\
     & \lesssim \frac{\sqrt{N_i} \mu^2_h }{Q \sqrt{L}}  \mu_{p-1}. \notag
   \end{align}
   Furthermore, (\ref{ineq:intern9999}) yields, analogously to the derivation of (\ref{ineq:RoverSigma}), that
   \begin{align}\label{ineq:fraction1}
     \frac{ \vert \Gamma_p \vert R^2}{\sigma^2} \le \vert \Gamma_p \vert \frac{ \underset{j \in \Gamma_p}{\max}  \vert h^*_i S_{i,p+1} b_{i,\ell} \vert^2  \vert b^*_{i,j} h_i \vert^2\Vert w_{i,j} \Vert_{\ell_2}^2 }{ \sum_{j \in \Gamma_p} \Vert w_{i,j} \Vert^2_{\ell_2} \vert b^*_{i,j} h_i \vert^2  \vert b^*_{i,\ell} S^*_{i,p+1} h_i \vert^2} \lesssim Q \le L.
   \end{align}
   Applying Theorem \ref{matrixbernstein} with $t= \left( \omega +2 \right)   \log L $ and $ \alpha =1 $ we obtain that with probability $1 - \mathcal{O} \left( L^{-\omega-2}\right)$ 
   \begin{equation*}
     \Big\Vert   \sum_{j \in \Gamma_p} \textbf{z}_{i,j} \Big\Vert_{\ell_2} \lesssim_{\omega}  \frac{ \mu_{p-1}}{\sqrt{L}} \max \left\{ \sqrt{ \frac{ N_i \mu^2_h  }{Q}    \log L };  \frac{ \sqrt{N_i} \mu^2_h  }{Q}   \left( \log L \right)^2   \right\},
   \end{equation*}
   which implies (\ref{muestimate1}), if the numerical constant in (\ref{Qbound}) is chosen large enough.\\
   % Thus, with probability at least $1 - L^{-\alpha} $ we find
   % \begin{align*}
   % \Big\Vert  \sum_{j \in \Gamma_p}  \textbf{z}_{i,j} \Big\Vert_{\ell_2} & \le C_{\alpha} \sigma \sqrt{\ln L}   \\
   % & \le C_{\alpha} \sqrt{\mu^2_h \mu^2_{p-1} QN \ln L}.
   % \end{align*}
   % This implies that with probablity at least $1 - L^{-\alpha} $ we have (\ref{muestimate1}) provided $Q \ge C_{\alpha} \mu^2_h N \ln L $ for some sufficiently large constant $C_{\alpha} $ only depending on $ \alpha $. \\
 \textbf{Step 2: Proof of (\ref{muestimate2})} By Lemma \ref{lemma:orliczestimates} we obtain that
\begin{align*}
\Big\Vert \vert z_{i,j} \vert \Big\Vert_{\psi_1} & \lesssim \frac{L}{Q} \vert b^*_{i,j} \left( \Id - h_i h^*_i \right) S_{i,p+1} b_{i,\ell} \vert  \Vert w_{i,j} \Vert_{\ell_2}\\
& \le \frac{L}{Q} \Vert b_{i,j} \Vert_{\ell_2} \Vert \Id - h_i h^*_i \Vert_{2 \rightarrow 2} \Vert S_{i,p+1} \Vert_{2 \rightarrow 2} \Vert b_{i,\ell} \Vert_{\ell_2} \Vert w_{i,j} \Vert_{\ell_2}\\
& \lesssim \frac{L}{Q} \Vert b_{i,j} \Vert_{\ell_2} \Vert b_{i,\ell} \Vert_{\ell_2} \Vert w_{i,j} \Vert_{\ell_2}  \lesssim \frac{ K_{i,\mu} }{Q\sqrt{L}} \mu_{p-1}
\end{align*}
and
\begin{align*}
\sum_{j \in \Gamma_p} \Big\Vert \vert z_{i,j} \vert \Big\Vert^2_{\psi_1} &\lesssim \frac{L^2}{Q^2} \left( \underset{j \in \Gamma_p}{\max} \ \Vert w_{i,j} \Vert^2_{\ell_2} \right) \sum_{j \in \Gamma_p} \vert b^*_{i,j} \left( \Id - h_i h^*_i \right) S_{i,p+1} b_{i,\ell} \vert^2\\
&= \frac{L}{Q}  \left( \underset{j \in \Gamma_p}{\max} \ \Vert w_{i,j} \Vert^2_{\ell_2} \right) \Vert T^{\frac{1}{2}}_{i,p}  \left( \Id - h_i h^*_i \right) S_{i,p+1} b_{i,\ell} \Vert^2_{\ell_2}\\
& \lesssim \frac{L}{Q}\Vert b_{i,\ell} \Vert^2_{\ell_2} \Vert w_{i,j} \Vert^2_{\ell_2}  \lesssim \frac{ K_{i,\mu} }{Q L} \mu^2_{p-1}.
\end{align*}
Consequently, Theorem \ref{bernsteinonedimensional} applied with $ t= \left( \omega +2 \right)   \log L $ yields that
\begin{align*}
\Big\vert \sum_{j \in \Gamma_p} z_{i,j} \Big\vert \lesssim_{\omega} \frac{\mu_{p-1}}{\sqrt{L}} \max \left\{ \sqrt{ \frac{   K_{i, \mu} \log L   }{Q}};\  \frac{ K_{i, \mu}  }{Q }  \log L   \right\}
\end{align*}
   with probability $1 - \mathcal{O} \left( L^{-\omega-2} \right) $, which shows  (\ref{muestimate2}).\\
	\textbf{Step 3: Proof of (\ref{muestimate3})} 
	As for $k_1 \ne i, k_2 \ne i$ the vectors $\textbf{z}_{k_1,j}$ and $\textbf{z}_{k_2,j} $ are not independent, %. Thus, standard concentration inequalities like the matrix Bernstein inequality or Theorem \ref{gaussianconcentration} cannot be applied directly. %In order to deal with this denote by $ \mathbb{E}_c [ \cdot] $ the expectation of some random variable (or random matrix) conditioned onto $ \left\{ c_{i,j} \right\}_{j \in \Gamma_p}  $. We observeIn order to deal with this 
	we will condition on the random variables $ \left\{ c_{i,j} \right\}_{j \in \Gamma_p}  $ and then apply Corollary~\ref{gaussianconcentration}. For that, we bound
	\begin{align}
	\Big\vert \sum_{k \ne i} \sum_{j \in \Gamma_p} \mathbb{E}\left[  \textbf{z}^*_{k,j}  \textbf{z}_{k,j} \Big\vert \left\{ c_{i,j} \right\}_{j \in \Gamma_p}   \right]  \Big\vert  & = \frac{L^2}{Q^2} \sum_{k \ne i} \sum_{j \in \Gamma_p} \Vert w_{k,j} \Vert^2_{\ell_2}  \Vert c_{i,j}  \Vert_{\ell_2}^2 \vert h^*_i b_{i,j} \vert^2 \vert h^*_i  S_{i,p+1} b_{i,\ell} \vert^2 \notag \\
	%&= \sum_{k \ne i} \sum_{j \in \Gamma_p} \Vert W^*_{i,p-1} S^*_{i,p} b_{i,j} \Vert_{\ell_2}^2 \Vert c_{i,j}  \Vert_{\ell_2}^2 \vert h^*_i b_{i,j} \vert^2 \vert h^*_i  S^*_{i,p+1} b_{i,l} \vert^2 \\
	%& \le \frac{L^2}{Q^2} \left( \underset{j \in \Gamma_p}{\max} \ \Vert c_{i,j} \Vert_{\ell_2}^2 \right) \sum_{k \ne i} \sum_{j \in \Gamma_p} \Vert w_{k,j} \Vert_{\ell_2}^2 \vert h^*_i b_{i,j} \vert^2 \vert h^*_i  S_{i,p+1} b_{i,\ell} \vert^2\\
	&  \le \mu^2_{p-1} \frac{ \mu^2_h }{Q^2} \left( \underset{j \in \Gamma_p}{\max} \ \Vert c_{i,j} \Vert_{\ell_2}^2 \right)  \sum_{k \ne i} \sum_{j \in \Gamma_p}   \vert h^*_i b_{i,j} \vert^2 \label{internreferenzierung1}  \\
		&  \le \mu^2_{p-1} \frac{ \mu^2_h }{LQ} \left( \underset{j \in \Gamma_p}{\max} \ \Vert c_{i,j} \Vert_{\ell_2}^2 \right)  \sum_{k \ne i}  \Vert T^{1/2}_{i,p} h_i \Vert^2_{\ell_2} \notag \\
	&\lesssim \mu^2_{p-1} \frac{ r  \mu^2_h }{Q L} \left( \underset{j \in \Gamma_p}{\max} \  \Vert c_{i,j} \Vert_{\ell_2}^2 \right)    . \notag
	\end{align}
	Analogously, using the triangle inequality,
	\begin{align*}
	&\Big\Vert \sum_{k \ne i} \sum_{j \in \Gamma_p} \mathbb{E} \left[ \textbf{z}_{k,j} \textbf{z}^*_{k,j} \Big\vert \left\{ c_{i,j} \right\}_{j \in \Gamma_p}  \right]   \Big\Vert_{2 \rightarrow 2}\\
	 =&  \frac{L^2}{Q^2} \Big\Vert  \sum_{k \ne i} \sum_{j \in \Gamma_p}  c_{i,j} c^*_{i,j}  \mathbb{E} \left[ \vert c_{k,j}^* w_{k,j}  \vert^2 \right]  \vert h^*_i b_{i,j} \vert^2 \vert h^*_i S_{i,p+1} b_{i,\ell} \vert^2  \Big\Vert_{2\rightarrow 2}  \\
	%&\le  \sum_{k \ne i} \sum_{j \in \Gamma_p} \Vert c_{i,j} c^*_{i,j} \Vert_{2 \rightarrow 2}  \mathbb{E}_c \left[  \Vert c_{k,j}^* W^*_{i,p-1} S^*_{i,p} b_{i,j} \Vert_{\ell_2}^2   \right]  \vert  h^*_i b_{i,j} \vert^2 \vert h^*_i  S^*_{i,p+1} b_{i,l} \vert^2 \\
	 \le & \frac{L^2}{Q^2} \sum_{k \ne i} \sum_{j \in \Gamma_p} \Vert c_{i,j} \Vert_{\ell_2}^2  \Vert w_{k,j} \Vert_{\ell_2}^2   \vert  h^*_i b_{i,j} \vert^2 \vert h^*_i  S_{i,p+1} b_{i, \ell} \vert^2 \\
	%&\le \mu^2_{p-1}  \mu_h^2 \left( \underset{j \in \Gamma_p}{\max} \Vert c_{i,j} \Vert_{\ell_2}^2 \right) \sum_{k \ne i} \sum_{j \in \Gamma_p}   \vert  h^*_i b_{i,j} \vert^2  \\
	 %\lesssim & \mu^2_{p-1} \frac{  \mu_h^2}{Q^2} \left( \underset{j \in \Gamma_p}{\max} \ \Vert c_{i,j} \Vert_{\ell_2}^2 \right)   \sum_{k \ne i} \sum_{j \in \Gamma_p}     \vert  h^*_i b_{i,j} \vert^2    \\
	 \overset{(\ref{internreferenzierung1})}{\lesssim} &  \mu^2_{p-1} \frac{ r  \mu_h^2}{QL} \left( \underset{j \in \Gamma_p}{\max} \ \Vert c_{i,j} \Vert_{\ell_2}^2 \right).
	\end{align*}
	Conditionally on $ \left\{ c_{i,j} \right\}_{j \in \Gamma_p} $, we can now apply Corollary \ref{gaussianconcentration} with $ t = \left( \omega +2 \right)  \log L $. Together with the last two estimates this yields that with probability $ 1 - \mathcal{O} \left( L^{-\omega-2}  \right) $
	\begin{align*}
	\Big\Vert  \sum_{k \ne i} \sum_{j \in \Gamma_p}  \textbf{z}_{k,j}   \Big\Vert_{\ell_2} \lesssim_{\omega} \mu_{p-1}  \sqrt{ \frac{  r \mu^2_h \left( \underset{j \in \Gamma_p}{\max} \Vert c_{i,j} \Vert_{\ell_2}^2 \right)        \log L }{Q L} }.
	\end{align*}
	Then, by Lemma \ref{conditioninglemma} we obtain that inequality (\ref{muestimate3}) holds with probability $ 1 - \mathcal{O} \left(  L^{-\omega-2} \right) $, if the constant in (\ref{Qbound}) is chosen large enough.\\
	%This implies
	%\begin{align*}
	%\frac{1}{Q} \Big\Vert \sum_{k \ne i} \sum_{j \in \Gamma_p}  \textbf{z}_{k,j}   \Big\Vert_{\ell_2} & \le  C_{\alpha} \sqrt{ \frac{  K \log \left(rL \right)  \mu^2_{p-1} \ r  Q \mu_h^2   \log L }{Q} } \\
	%& \le C_{\alpha} \sqrt{ \frac{  K  \mu^2_{p-1} \ r  Q \mu_h^2 \left(  \log L \right)^2 }{Q} } \\
	%& \le \frac{1}{8} \mu_{p-1}.
	%\end{align*}
	%The first inequality is due to $L\ge r$ and the second one follows from (). This finishes the third step.\\
\textbf{Step 4: Proof of (\ref{muestimate4})}
Note that conditionally on $ \left\{ c_{i,j} \right\}_{j \in \Gamma_p} $   $ \sum_{k \ne i} \sum_{j \in \Gamma_p} z_{k,j}$ is a circular symmetric random variable with variance
	\begin{align*}
	\mathbb{E} \left[ \sum_{k \ne i} \sum_{j \in \Gamma_p}  \vert z_{k,j} \vert^2 \Big\vert \left\{ c_{i,j} \right\}_{j \in \Gamma_p}  \right] & = \frac{L^2}{Q^2}  \sum_{k \ne i} \sum_{j \in \Gamma_p}  \vert b^*_{i,\ell} S_{i,p+1} \left( \Id - h_i h_i^*\right) b_{i,j} \vert^2    \Vert w_{k,j} \Vert_{\ell_2}^2   \vert c^*_{i,j} m_i \vert^2  \\
	%& \le \frac{L \mu^2_{p-1}}{Q^2} \left(  \underset{j \in \Gamma_p, k \ne i}{\max}   \vert c^*_{i,j} m_i \vert^2 \right)      \sum_{k \ne i} \sum_{j \in \Gamma_p}  \vert b^*_{i,\ell} S_{i,p+1} \left( \Id - h_i h_i^*\right) b_{i,j} \vert^2 \\
	&  \le \mu^2_{p-1} \frac{1 }{Q} \left(  \underset{j \in \Gamma_p}{\max}   \vert c^*_{i,j} m_i \vert^2 \right)  \sum_{k \ne i} \Vert T^{1/2}_{i,p} \left( \Id - h_i h^*_i \right) S_{i,p+1} b_{i,\ell}  \Vert^2_{\ell_2} \\
	& \lesssim \mu^2_{p-1}  \frac{r  K_{i, \mu}  }{Q L}.
	\end{align*}
	Consequently, one obtains that with probability at least $ 1- \mathcal{O} \left( L^{-\omega-2} \right) $
	%\begin{align*}
	% \Big\vert \frac{1}{Q}  \sum_{k \ne i} \sum_{j \in \Gamma_p} z_{k,j} \Big\vert  &\le  C \sqrt{\frac{ \left(  \underset{j \in \Gamma_p}{\max} \vert c^*_{i,j} m_i \vert^2 \right) \mu^2_{p-1} \left(r-1\right)  K  \log L^{\alpha} }{Q}} \\
	%&= C_{\alpha} \sqrt{\frac{ \left(  \underset{j \in \Gamma_p}{\max} \vert c^*_{i,j} m_i \vert^2 \right) \mu^2_{p-1} \left(r-1\right)  K  \log L }{Q}}.
	%\end{align*}
	\begin{align*}
	 \Big\vert   \sum_{k \ne i} \sum_{j \in \Gamma_p} z_{k,j} \Big\vert \lesssim_{\omega} \mu_{p-1}  \sqrt{\frac{ \left(  \underset{j \in \Gamma_p}{\max} \vert c^*_{i,j} m_i \vert^2 \right) r   K_{i,\mu}   \log L }{Q L}}.  
	\end{align*}
	Thus, by Lemma \ref{conditioninglemma} inequality (\ref{muestimate4}) holds with probability at least $1- \mathcal{O} \left( L^{-\omega-2}  \right) $, if the constant in (\ref{Qbound}) is chosen large enough.\\
	%which implies
	%\begin{align*}
	%\frac{1}{Q}  \Big\vert  \sum_{k \ne i} \sum_{j \in \Gamma_p} z_{k,j} \Big\vert  & \le  C_{\alpha} \sqrt{\frac{  \log \left(rL\right) \mu^2_{p-1} r  K  \log L }{Q}} \\
	%& \le  C_{\alpha} \sqrt{\frac{   \mu^2_{p-1} r  K   \left ( \log L \right)^2 }{Q}} \\
	%& \le  \frac{1}{8} \mu_{p-1}.
	%\end{align*}
	%The first inequality follows from () and the second one is true due to $L > r$. This finishes the fourth step. \\ \\
\textbf{Union bound: }%By a union bound argument we see that
%        for fixed $p$ we have $\mu_p \le \frac{1}{2} \mu_{p-1} $ with
%        probability $1- rQ \mathcal{O} \left(L^{-\alpha}\right)
%        $.
%        Thus the lemma holds with probability at least
%        $1- rL \mathcal{O} \left( L^{-\alpha} \right) $. Repeating the
%        argument for $ \tilde{\alpha} = \alpha +2 $ we get that the
%        lemma is true for
%        $1- rL \mathcal{O} \left( L^{-\alpha-2} \right) = 1 -
%        \mathcal{O} \left( L^{-\alpha} \right) $,
%        which finishes the proof.     
   By the previous four steps we see that for fixed $ p \in \lbrack P \rbrack $, $ \ell \in \Gamma_{p+1} $, and $i \in \lbrack r \rbrack$ the inequalities (\ref{muestimate1}), (\ref{muestimate2}), (\ref{muestimate3}), (\ref{muestimate4}) hold with probability $ 1- \mathcal{O} \left( L^{-\omega-2} \right) $. Thus, by (\ref{verylastinequality}) and a union bound we have $ \mu_{p-1} \le \frac{1}{4} \mu_p$ with probability $ 1- rQ\ \mathcal{O} \left( L^{-\omega-2} \right) $ for fixed $p \in \lbrack P-1 \rbrack$. Thus, with probability at most $ 1- rPQ\ \mathcal{O} \left( L^{-\omega-2} \right) $ we obtain $  \mu_{p-1} \le \frac{1}{4} \mu_p $ for all $ p\in  \lbrack P-1 \rbrack $. We obtain the desired result as we find $ r \lesssim Q\le  L$ and $ PQ = L $.
        
	%\begin{equation*}
	%\mathbb{P} \left( \Big\vert  \sum_{k \ne i} \sum_{j \in \Gamma_p} z_{k,j} \Big\vert   \ge t \right) \le c_1  \exp \left(  \frac{ c_2 t^2 }{\sigma_4^2} \right)
	%\end{equation*}
	%for some constants $c_1$ and $c_2$.  Setting $t= c_{\alpha, 4} \sqrt{Q} \mu_{p-1} $, where $c_{\alpha, 4}$ will be determined later, we get
	%\begin{align*}
	%\mathbb{P} \left( \Big\vert  \sum_{k \ne i} \sum_{j \in \Gamma_p} z_{k,j} \Big\vert   \ge c_{\alpha} \mu_{p-1} \sqrt{Q} \right) &\le c_1  \exp \left( - \frac{ c_2 c^2_{\alpha} Q }{\sigma_4^2} \right) \\
	%&\le c_1 \exp \left(  \frac{c_2 c^2_{\alpha} Q}{\left(  \underset{j \in \Gamma_p}{\max} \vert c^*_{i,j} m_i \vert^2 \right) \mu^2_{p-1} \left(r-1\right) Q K} \right)
	%\end{align*}
	
\end{proof}

\subsubsection{An upper bound for $ \Vert z \Vert_{\ell_2} $}
In the case of noise, the error bound given by Lemma \ref{dualcertificate}  is proportional to $ \Vert z \Vert_{\ell_2} $, where $z$ is the dual certificate as constructed in (\ref{hdefinition}). Thus, one needs an upper bound for $ \Vert z \Vert_{\ell_2}$. This will be accomplished by the following lemma.
\begin{lemma}\label{lemma:zupperbound}
Let $z \in \mathbb{C}^L$ be given by (\ref{hdefinition}) and assume that $\Vert W_p \Vert_F \le 4^{-p} \sqrt{r}  $. Furthermore, suppose that $\mathcal{A}^p $ satisfies the $\delta$-local isometry property (\ref{isometryproperty1}) with $ \delta \le \frac{1}{4} $ on $\mathcal{T}^p$ for all $p \in \left[P \right]$. Then
\begin{equation*}
\Vert z \Vert_{\ell_2} \lesssim \sqrt{r}.
\end{equation*}
\end{lemma}
\begin{proof}
Observe that
\begin{align*}
\Vert z \Vert_{\ell_2} \le \sum_{p=1}^{P} \Vert \mathcal{A}^p \mathcal{S}^p \left( W_{p-1} \right) \Vert_{\ell_2} \lesssim  \sum_{p=1}^{P} \Vert  W_{p-1} \Vert_F  \lesssim  \sum_{p=0}^{P-1} 4^{-p} \sqrt{r} \lesssim \sqrt{r},
\end{align*}
where the first equality follows from the definition of $z$ (\ref{hdefinition}) and the triangle inequality. The second inequality is due to the local isometry property (\ref{isometryproperty1}) and (\ref{Sinequality2}). We derive by (\ref{decay1}) the desired bound. 
\end{proof}

\subsection{Proof of Theorem \ref{theorem:mainwithnoise}}\label{subsec:proofmaintheorem}
First of all, recall that by Lemma \ref{lemma:operatornormbound} with probability at least  $ 1- 2 \exp \left( -t\right) $ it holds that
\begin{equation}\label{ineq:operatornormboundintern}
	\gamma = \Vert \mathcal{A} \Vert_{F \rightarrow 2} \le  2 \sqrt{  \omega \max  \left\{ 1; \frac{r  K_{\mu} N }{L}   \right\}  \log \left( L +r KN \right)}.
\end{equation}
In the following, let $ \left\{ \Gamma_p \right\}^P_{p=1} $ be an $ \omega$-admissible partition of $ \lbrack L \rbrack $ (see Definition \ref{definition:admissiblepartition}), which is a minimizer of (\ref{definition:muhomega}). From Definition \ref{definition:admissiblepartition} combined with the assumptions on $L$ (see (\ref{Lbound})) we infer that
    \begin{align}
    Q = \frac{L}{P} &\gtrsim r \left(  K_{\mu}  \log \left(  K_{\mu}  \right) + N \mu^2_h \right) \left( \log L \right)^2 \label{ineq:Qboundinternlast} \\
    P &\ge  \frac{1}{2}  \log \left( 8 \gamma  \sqrt{r} \right).  \label{ineq:Pboundinternlast}
    \end{align}
    Note that due to Theorem \ref{localisometry} and our assumptions
    on $L$ and $Q$ (and also $\log K_{\mu} \le \log L $) we may assume that the inequalities
    (\ref{isometryproperty2}) and (\ref{isometryproperty1}) hold with
    probability $ 1- \mathcal{O} \left( L^{-\omega}  \right)$ and
    constant $ \delta= \frac{1}{32} $. Thus, by Lemma
    \ref{dualcertificate} applied with $  \alpha= \frac{1}{8 \gamma}
    $, $ \beta = \frac{1}{4}$, and $ \delta= \frac{1}{4}$ it is enough
    to construct $ Y \in \text{Range} \left( \mathcal{A}^* \right)$
    which satisfies (\ref{dualcertificatecondition1}) and
    (\ref{dualcertificatecondition2}). This is achieved by the Golfing
    Scheme as explained in Section \ref{subsec:golfing}: Note that the assumption of Lemma \ref{thmexponentialdecay} is given by (\ref{ineq:Pboundinternlast}) and (\ref{isometryproperty1}). Thus, it holds that $\Vert W_p \Vert_F \le 4^{-p}
    \sqrt{r} $ for all $ p \le P $ and, by (\ref{equ:wpdefinition}), $Y=Y_P$ satisfies Condition
    (\ref{dualcertificatecondition1}). Furthermore, observe that Lemma \ref{mudecay} implies that with
    probability $  1- \mathcal{O} \left( L^{-\omega}  \right) $ one has $ \mu_{p} \le \frac{1}{4} \mu_{p-1} $ for all $ p \in \lbrack P-1 \rbrack$. Using this fact and $
    \Vert W_p \Vert_F \le 4^{-p} \sqrt{r}$ it follows from Lemma
    \ref{operatorbound} that Condition
    (\ref{dualcertificatecondition2}) is fulfilled. Using a union bound we conclude that with
    probability   $  1- \mathcal{O} \left( L^{-\omega}  \right) $ the
    approximate dual certificate $Y=Y_P$ satisfies the assumptions in
    Lemma \ref{dualcertificate}. Thus, if $\hat{X}$ is a minimizer of \eqref{eq:nucbpdn} it satisfies the estimation error (\ref{recoveryerror}) .  \\
% \delta= \frac{1}{4}, \beta = \frac{1}{4}, \alpha= \frac{1}{8 \gamma}
%\begin{align*}
%\left(1 - \delta \right) \Vert X \Vert^2_{F} &\le \Big\Vert  \mathcal{A} \left( X \right) \Big\Vert_{\ell_2}^2     \le \left(1 + \delta \right) \Vert X \Vert^2_{F}\\	
%\left(1 - \delta \right) \sum_{i=1}^{r} \Vert T^{1/2}_{i,p} Y_i \Vert^2_{F} &\le \frac{L}{Q} \Big\Vert \mathcal{A}^p \left( Y \right) \Big\Vert_{\ell_2}^2     \le \left(1 + \delta \right) \sum_{i=1}^{r} \Vert T^{1/2}_{i,p} Y_i \Vert^2_{F}.
%\end{align*}
It remains to prove the upper bound for the estimation error in order to obtain inequality (\ref{ineq:estimationerror}). Note that by Lemma \ref{lemma:zupperbound} we have that $\Vert z \Vert_{\ell_2} \lesssim \sqrt{r} $. Thus, in combination with (\ref{ineq:operatornormboundintern}) we derive
\begin{align*}
\Vert \hat{X} - X^{0} \Vert_F &\lesssim \left(1 + \gamma \right) \left(1 + \Vert z \Vert_{\ell_2} \right) \tau \\
	& \lesssim_{\omega} \tau \sqrt{ r  \max  \left\{ 1; \frac{r  K_{\mu} N }{L}   \right\}  \log L }  .
\end{align*}
This finishes the proof. \qed
\section{Outlook}
\label{sec:numerics}

\newcommand{\deriv}[2]{\ensuremath{\frac{\partial #1}{\partial #2}}}
\newcommand{\bigtimes}{\times}

Although the convex formulation in \eqref{eq:nucbpdn} is important for
theoretical investigations it is also obvious that for many real-word
applications nuclear minimization is not feasible due to its
computional complexity as lifting considerably increases the number of
optimization variables.  For the case $r=1$ a nonconvex approach has
been proposed by \cite{Li2016} which has been demonstrated not only to
be considerably more efficient but also to achieve a better empirical
performance. Shortly before the completion of our work this line of
research has been extended to $r\geq 1$ with explicit guarantees
\cite{Ling:2017}, but again for a number of measurements depending
quadratically on $r$. As in \cite{lingstrohmer}, the dependence observed in numerical experiments is linear.  We expect that
the mathematical analysis conducted in this paper will also  be important for establishing near-optimal performance guarantees for more efficient
algorithms. For this reason we include such a nonconvex approach
similar to the one analysed in \cite{Ling:2017}
in our numerical experiments, comparing it to nuclear norm minimization as analyzed in this paper.

More precisely, we consider a gradient-based (Wirtinger flow)
recovery algorithm minimizing the residual
\begin{equation}
  F( h,x) := \|\mathcal{A}( h_1  x_1^*, \ldots,  h_r  x_r^*) - y
  \|_{\ell_2}^2
  \label{eq:residual}
\end{equation}
where $h:=(h_1, \ldots, h_r)$ and $x:=(x_1, \ldots, x_r)$.  Observe
that in the noiseless case one has $F( h, x)=0$ for the ground truth.
Note that, while minimizing $F$ has been shown empirically in
\cite{Ling:2017} to have good recovery properties, where
guarantees only apply to a regularized variant.
As $F$ is highly non-convex in $(h,x)$ and possesses many local
minima, it is essential to find a good initial guess to start the
minimization process (cf.~\cite{Li2016,Ling:2017}). 
%Now, observe that the adjoint mapping
%$\mathcal{A}_i^* : \C^L \to \C^{K_i\times N_i}$ of $\mathcal{A}_i$ is
%given by $\mathcal{A}^*_iy=B_i^* \diag(y) C_i$ and therefore
%$\mathcal{A}^*y=(\mathcal{A}_1^*y,\dots,\mathcal{A}_r^*y)$.  Since the
%matrices $C_1, \ldots, C_r$ have i.i.d. complex-normal distributed 
%the initial guess $\mathcal{A}^* y$ sharply concentrates around its
%expectation $ h_0 x_0^*$.biggest
%It will be therefore close to a vector of rank-1 matrices (only one singular value is large whereas
%the remaining singular values are negligible).  This leads to the
%initialization given in Algorithm \ref{algo:init}.
Eq. \eqref{equ:operatornormboundinline1} motivates the
initialization given in the following algorithm.
\begin{algorithm}[H]
	\caption{Initialization}\small
	\label{algo:init}
	\begin{algorithmic}
		\State \textbf{Input:} Observation $ y$.
		%\State \Procedure{Init}{}
		\State $\left(  Z_1, \ldots,  Z_r \right) \gets \mathcal{A}^*  y$.
		\For{$k = 1, \ldots, r$}
		\State $d_k\gets$  largest singular value of $ Z_k$.
		\State Let $ v^{(0)}_k$ and $ u^{(0)}_k$ be the corresponding left and right singular vectors, respectively.
		\State $ v^{(0)}_k \gets \sqrt{d_k}  v^{(0)}_k$ and  $ u^{(0)}_k \gets \sqrt{d_k}  u^{(0)}_k$
		\EndFor 
		%\State \Return $ v^{(0)},  u^{(0)}$.
		%\EndProcedure
		\State \textbf{Output:} Initial guesses $ v^{(0)},  u^{(0)}$.
	\end{algorithmic}
\end{algorithm}
To minimize $F$ a gradient descent approach is used. 
Here the gradient of a function $f : \C^n \to \C$ at $z_0\in \C^n$ is given
by $\nabla_{ z} f( z_0) = \left(\deriv{f}{ z}( z_0)\right)^* \in \C^n$
where for
$z=u+iv\in\C$ the Wirtinger derivatives are
$\deriv{}{z} = \frac{1}{2} \left( \deriv{}{u} - i \deriv{}{v}\right)$
and
$\deriv{}{\overline{z}} = \frac{1}{2} \left( \deriv{}{u} + i
  \deriv{}{v} \right)$.
Since for real-valued complex
functions $f : \C^n \to \R$ one has $\deriv{f}{\overline{z}} = \overline{\deriv{f}{z}}$, we do not need to consider $\deriv{f}{\overline{z}}$ here. Consequently, we obtain
\begin{align*}
&\nabla_{ h_i} F( h,  x)
= \left( \diag \left( C_i  \overline{x_i} \right)  B \right)^* {\left( \mathcal{A}( h  x^*) -  y \right)};
\\
&\nabla_{ x_i} F( h,  x)
= \left( \diag( B_i  h_i) C_i \right)^T \overline{\left( \mathcal{A}( h  x^*) -  y \right)}
\end{align*} 
%where $h x^* := (h_1 x_1^*,\dots, h_r x_r^*)$.  The gradient descent method is then Algorithm \ref{algo:gradient}.
To estimate a suitable stepsize $\eta$ for each iteration
we use the backtracking line search.
\begin{algorithm}[H]
	\caption{Wirtingers gradient descent with backtracking}
	\label{algo:gradient}
	\begin{algorithmic}\small
		%\State \textbf{Input:} Initial values $ v^{(0)},  u^{(0)}$, initial stepsize $\eta^{(0)}$, stepsize scaling factor $\beta$.
		\State \textbf{Input:} Initial values $ v^{(0)},  u^{(0)}$.
		%\State
		%\Procedure{Gradient}{}
		\For{$i = 1, \ldots$}
		%\State $\eta \gets \textsc{Backtrack} \left( \eta_0, \beta,  v^{(i-1)},  u^{(i-1)} \right)$
		\State $\eta \gets \textsc{line-search} \left( v^{(i-1)},  u^{(i-1)} \right)$
		\State $ v^{(i)} \gets  v^{(i-1)} - \eta \nabla_{ h} F \left( v^{(i-1)},  u^{(i-1)} \right)$
		\State $ u^{(i)} \gets  u^{(i-1)} - \eta \nabla_{ x} F \left( v^{(i-1)},  u^{(i-1)} \right)$
		%\If{$F \left( v^{(i)},  u^{(i)} \right) < \varepsilon^2 \textbf{ or } \| \nabla F \left( v^{(i)},  u^{(i)} \right)\|_{\ell_2}^2 < \varepsilon^2$} 
		\If{$ \| \nabla F \left( v^{(i)},  u^{(i)} \right)\|_{\ell_2} < \varepsilon$} 
		\State \Return $ v^{(i)},  u^{(i)}$
		\EndIf
		\EndFor
		%\EndProcedure
		\State \textbf{Output:} Approximate solutions $ v^{(i)},  u^{(i)}$.
	\end{algorithmic}
\end{algorithm}

\paragraph{Numerical Results:}
We have investigated both nuclear norm minimization
\eqref{eq:nucbpdn} and Algorithms \ref{algo:init} and
\ref{algo:gradient} in the noiseless case for different values of $r$
and $L$
with equal channel dimensions $K=K_1=\ldots=K_r=8$ and signal dimensions
$N=N_1=\ldots=N_r=8$. The success rates per device are estimated
numerically and plotted as a function of $\rho=L/\sum_{i=1}^r(K_i+N_i)$.  The convex program
\eqref{eq:nucbpdn} is solved using the Matlab CVX toolbox.  For each
experiment the matrices $C_i\in\C^{L\times N}$, the signal vectors
$x^0_i\in\C^N$, and the channel coefficients $h^0_i\in\C^K$ are
generated with i.i.d.  complex normal distributed entries.
Recovery is considered successful for a device
if the corresponding signal pair $(h_i,x_i)$ for $i \in \lbrack r \rbrack $ 
fullfils 
$\lVert h_ix_i^*-h^0_ix^{0*}_i\rVert_F/\lVert
h^0_ix^{0*}_i\rVert_F\leq 1\%$.  Furthermore, the stopping
criterion for the Wirtinger approach is chosen to be $\epsilon=10^{-4}$
and the maximal number of iterations is limited to $1000$.\\

Our experiments confirm the findings of \cite{lingstrohmer} and
\cite{Ling:2017} that for both the convex and the non-convex approach
the scaling is linear.  The results in Figure
\ref{fig:experiments:phasetrans} show that -- almost independently of
$r$ -- the phase transition for \eqref{eq:nucbpdn} occurs at
$\rho\approx 2.75$ while the Wirtinger flow approach performs
considerably better with a phase transition
(for larger $r$) at $\rho\approx 1.17$.

\begin{figure*}[h]
	\centering
	\subfloat[][]{
		\includegraphics[width=.5\linewidth]{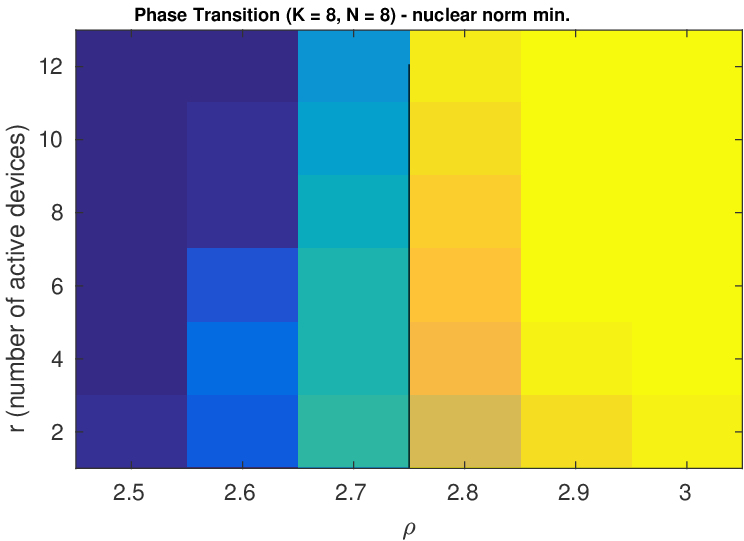} 
	}
	\subfloat[][]{
		\includegraphics[width=.5\linewidth]{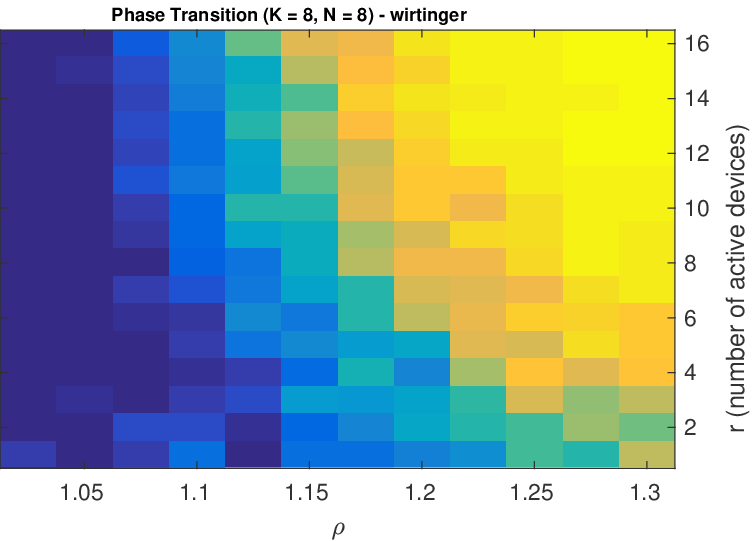} 
	}
	
	\caption{Phase transition of the success rates per device for (a)
          the convex approach \eqref{eq:nucbpdn} and (b) the Wirtinger
          approach for $K=N=8$ where $\rho=L/\sum_{i=1}^r(K_i+N_i)$.
        }
	\label{fig:experiments:phasetrans}
\end{figure*}

%\section*{Conclusions}
\section*{Acknowledgements}

The three authors acknowledge support by the Hausdorff Institute for
Mathematics (HIM), where part of this work was completed in the
context of the HIM Trimester Program {\em Mathematics of Signal
Processing}. This work has been supported by German Science Foundation
(DFG) in the context of the joint project {\em Bilinear
Compressed Sensing} (JU2795/3-1, KR 4512/2-1) as part of the Priority Program 1798. Furthermore, the
authors want to thank David Gross, Richard Kueng, Kiryung Lee, Shuyang Ling, and Tom Szollmann for fruitful discussions.

\bibliographystyle{amsalpha}
\bibliography{notes}

\appendix

\begin{appendices}
  \section{Construction of the partition
    $ \left\{ \Gamma_p \right\}_{p \in
      \left[P\right]}$}
  \label{constructionpartition}

\subsection{Proof of Lemma \ref{partitionlemma}}

The goal of this section is to prove Lemma \ref{partitionlemma}. Our proof will rely on the following lemma.
\begin{lemma}\label{candeslemma}
Fix $i \in \left[ r \right] $ and let $Q \in \left(0,L \right) $, $ \delta> 0 $ and $ \nu \in \left(0,1\right) $. Assume that
\begin{equation}\label{onepartition}
Q \ge C \frac{ K_{i,\mu}   }{\nu^2}  \log \frac{K_i}{\delta},
\end{equation}
where $ C>0 $ is an absolute constant and let $\hat{\delta}_1, \ldots, \hat{ \delta}_L$ be independent, identically distributed random variables such that
\begin{equation*}
\mathbb{P} \left( \hat{\delta}_1 = 1 \right) = \frac{Q}{L}  \quad \text{and}  \quad  \mathbb{P} \left( \hat{ \delta}_1= 0 \right) = 1-\frac{Q}{L} .
\end{equation*} Then with probability exceeding $1- \delta $ we have that
\begin{equation*}
\Big\Vert \frac{L}{Q}  \sum_{\ell=1}^{L} \hat{ \delta}_{\ell} b_{i,\ell} b^*_{i, \ell} - \Id  \Big\Vert_{2\rightarrow 2} \le \nu.
\end{equation*}
\end{lemma}
A proof of this lemma can be obtained using arguments contained in the proof of Theorem 1.2 in \cite{candes2007sparsity}. For the sake of completeness we will give a proof below (relying on different techniques). Our proof of Lemma \ref{partitionlemma}  will use essentially the same ideas as in $\cite{ARR2012} $, but has been slightly refined.
\begin{proof}[Proof of Lemma \ref{partitionlemma}] 
Let $ \hat{\delta}_1, \ldots, \hat{\delta}_k$ be independent, uniformly distributed random variables which take values in $ \lbrack P \rbrack $. For $ p \in \lbrack P \rbrack $ we define
\begin{equation*}
 \Gamma_p = \left\{ \ell  \in \lbrack L \rbrack: \ \hat{\delta}_{\ell}= p   \right\}.
\end{equation*}
Thus, $  \left\{ \Gamma_p \right\}_{p \in \lbrack P \rbrack} $ is a partition of $ \lbrack L \rbrack$. To finish the proof it is enough to show that with positive probability  the partition $ \left\{ \Gamma_p \right\}_{p \in \lbrack P \rbrack} $ has the required properties, i.e., for all $ p \in \left[ P \right]  $, (\ref{partitionequation}) holds and $ \frac{1}{2} Q \le \vert \Gamma_p \vert \le  \frac{3}{2} Q $ . For $ i \in \lbrack r \rbrack $ and $ p \in \lbrack P \rbrack $ we define the event
\begin{equation*}
A_{i,p} = \left\{ (\ref{partitionequation}) \text{ fails}  \right\}  = \left\{ \Big\Vert  \frac{L}{Q} \sum_{ \ell \in \Gamma_p}  b_{i,\ell} b^*_{i,\ell}  - \Id   \Big\Vert_{2 \rightarrow 2} > \nu  \right\} .
\end{equation*}
Set $ \delta=  \frac{1}{3rP } $ and note that $ \log (\frac{K}{\delta})  = \log \left(3 rP K \right) \lesssim  \log \left( \max \left\{ r; P; K \right\} \right) $. Thus, by Lemma \ref{candeslemma} we get that $ \mathbb{P} \left(A_{i,p}\right) \le \frac{1}{3rP} $, if the constant in inequality (\ref{Qminimalsize}) is chosen large enough. By a union bound over all choices of $i$ and $p$, (\ref{partitionequation}) follows with probability at least $\frac{1}{3}$.
%\begin{equation}\label{requiredproperty1}
% \mathbb{P} \left( \underset{i \in \lbrack r \rbrack, p \in \lbrack P \rbrack}{\bigcup} A_{i,p} \right) \le \sum_{i \in \lbrack r \rbrack, p \in \lbrack P \rbrack} \mathbb{P} \left( A_{i,p} \right) \le rP \underset{i \in \lbrack r \rbrack, p \in \lbrack P \rbrack}{\max}   \mathbb{P} \left( A_{i,p} \right)   \le \frac{1}{3}.
% \end{equation}
It remains to control the size of the sets $ \lbrace  \Gamma_p  \rbrace_{p \in \lbrack P \rbrack }$. By the Bernstein inequality for bounded random variables (e.g., \cite[Corollary 7.31]{FR2013}) we obtain that for fixed $ p \in \lbrack P \rbrack $ one has $ \frac{Q}{2} \le \vert \Gamma_p \vert \le \frac{3Q}{2} $ with probability at least $ 1- 2 \exp \left( \frac{-Q}{10} \right) \ge 1 - \frac{1}{2P} $, where the last inequality follows from (\ref{Qminimalsize}), if the constant $C$ is chosen large enough. Thus, by a another union bound we observe
\begin{equation*}
\mathbb{P} \left(  \frac{Q}{2} \le \vert \Gamma_p \vert \le \frac{3Q}{2} \text{ for all } p \in \lbrack P \rbrack \right)  >  \frac{1}{2}.
\end{equation*}
Thus with positive probability the partition $ \left\{ \Gamma_p \right\}_{p \in \lbrack P \rbrack} $ has the required properties. In particular, this implies the existence of a partition $ \left\{ \Gamma_p \right\}_{p \in \lbrack P \rbrack} $ with the properties stated in Lemma \ref{partitionlemma}.

\end{proof}

\subsection{Proof of Lemma \ref{candeslemma}}
As already mentioned before this lemma can be proven using arguments from the proof Theorem 1.2 in \cite{candes2007sparsity}.  The arguments in this article are based on Talagrand's inequality \cite{talagrandconcentration} and Rudelson's Lemma \cite{rudelsonisotropic}. Recent technical advances (see \cite{tropplecture}) allow us to give a simplified proof.
\begin{proof}
The goal is to use the matrix Bernstein inequality to estimate the spectral norm of
\begin{equation*}
Y = \frac{L}{Q} \sum_{\ell=1}^{L} \hat{\delta}_{\ell} b_{i, \ell} b^*_{i,\ell} - \Id. 
\end{equation*}
We will decompose $Y$ into a sum of independent random matrices with mean zero. Thus, by setting
\begin{equation*}
Y_{\ell} = \left( \hat{\delta}_{\ell} - \frac{Q}{L}   \right) \frac{L}{Q}  b_{i,\ell} b^*_{i,\ell }
\end{equation*}
we obtain $  Y = \sum_{\ell=1}^{L} Y_{\ell} $ and $ \mathbb{E} Y_{\ell} = 0 $ for all $ \ell \in \lbrack L \rbrack $ due to $ \Id =\sum_{\ell=1}^{L} b_{i,\ell} b^*_{i,\ell}  $. To apply the matrix Bernstein inequality we need first to obtain an upper bound for $ \Vert \mathbb{E} Y^2 \Vert_{2 \rightarrow 2} $. For that purpose note that
\begin{align*}
\mathbb{E} Y^2 = \sum_{\ell=1}^{L} \mathbb{E} Y^2_{\ell} = \sum_{\ell=1}^{L} \mathbb{E} \left[ \left( \hat{\delta}_{\ell} - \frac{Q}{L}   \right)^2 \right] \frac{L^2}{Q^2} \Vert b_{i,\ell} \Vert^2_{\ell_2}    b_{i,\ell} b^*_{i,\ell } %= \frac{L-Q}{L} \sum_{\ell=1}^{L}  \frac{\Vert b_{i, \ell} \Vert^2_{\ell_2}}{Q}   \frac{b_{i,\ell}  b^*_{i,\ell} }{L}
\end{align*}
Observe that $ \mathbb{E} \left[ \left( \hat{\delta}_{\ell} - \frac{Q}{L}   \right)^2 \right] = \frac{Q \left(L - Q\right)}{L^2}  $, which implies
\begin{equation*}
\mathbb{E} Y^2 =  \frac{L-Q}{L} \sum_{\ell=1}^{L}  \frac{ L \Vert b_{i, \ell} \Vert^2_{\ell_2}}{Q}   b_{i,\ell}  b^*_{i,\ell} 
\end{equation*}
Thus, by $ \sum_{\ell=1}^{L} b_{i,\ell} b^*_{i,\ell} = \Id $ and the definition of $K_{i,\mu} $ we get
\begin{align*}
\Vert  \mathbb{E} Y^2 \Vert_{2 \rightarrow 2} \le \frac{L-Q}{L} \left(  \underset{\ell \in \lbrack L \rbrack}{\max}~ L\Vert b_{i, \ell} \Vert^2_{\ell_2}  \right)  \Big\Vert   \sum_{\ell=1}^{L}    b_{i,\ell}  b^*_{i,\ell}  \Big\Vert_{2 \rightarrow 2} \le  \frac{ K_{i,\mu} }{Q}.
\end{align*}
Furthermore, for all $ \ell \in \lbrack L \rbrack $ we have
\begin{align*}
\Vert Y_{\ell} \Vert_{2 \rightarrow 2} \le \max \left\{  \frac{Q}{L}; \frac{L-Q}{L}  \right\}  \frac{L}{Q}  \Vert b_{i, \ell} \Vert^2_{\ell_2}  \le  \frac{L}{Q} \Vert b_{i, \ell} \Vert^2_{\ell_2}   \le \frac{ K_{i,\mu} }{Q} \quad \text{almost surely.}
\end{align*}
Thus, we can apply the matrix Bernstein inequality in the version of \cite[Theorem 6.6.1]{tropplecture} to obtain
\begin{align*}
\mathbb{P} \left(  \Vert Y \Vert_{2 \rightarrow 2} \ge \nu \right) &\le K \exp \left(  \frac{ - \nu^2 /2}{ \left(1 + \frac{\nu}{3}\right)  K_{i,\mu}  / Q  }\right) \overset{(\ref{onepartition})}{\le} K \exp \left( \frac{-C  \log \left( K/\delta\right)}{2 \left(1 + \frac{\nu}{3}\right) }  \right).
\end{align*}
As we have $ 0 < \nu < 1 $ this yields the claim if the constant $C>0$ in (\ref{onepartition}) is chosen large enough.

\end{proof}

\section{Circular-symmetric Complex Normal Random Variables}\label{sec:complexgaussian}
In this section we will recall some useful facts concerning random variables which have a circular--symmetric complex normal distribution $\mathcal{CN}(0,\sigma^2)$ with zero mean and variance $\sigma^2$. This means that their real and imaginary parts are uncorrelated jointly Gaussian with
zero mean and variance $\sigma^2/2$ (and are therefore independent). For more details concerning this probability distribution we refer to \cite[Section A.1.3]{Tse:FundamentalsWirelessCommunication}.
%Usually they will have distribution
%$\mathcal{N}\left(0, \sigma^2 \right) + i \mathcal{N}\left(0, \sigma^2
%\right) $.
%This means that the complex and the real part have mean 0 and variance
%$\sigma^2 $. 
The following two well-known lemmas are concerned with two useful identities. A proof of them can be found for example in \cite[Lemma 11 and 12]{ARR2012}.
\begin{lemma}\label{usefullemma3}
	Assume that $c \in \mathbb{C}^n$ is a random vector with independent
	entries
	%$c_i \sim \mathcal{N} \left(0, \frac{1}{2}\right) + i
	%\mathcal{N}\left( 0, \frac{1}{2} \right) $
	$c_i\sim\mathcal{CN}(0,1)$. Then we have
	\begin{equation*}
	\mathbb{E} \left[ \left( \Id - c c^* \right)^2 \right]= n \Id.
	\end{equation*}
\end{lemma}

\begin{lemma}\label{usefullemma2}
	Let $q \in \mathbb{C}^n$ be any deterministic vector. Furthermore,
	assume that $c \in \mathbb{C}^n$ is a random vector with independent
	entries
	%$c_i \sim \mathcal{N} \left(0, \frac{1}{2}\right) + i \mathcal{N}
	%\left( 0, \frac{1}{2} \right) $. 
	$c_i\sim\mathcal{CN}(0,1)$.
	Then we have
	\begin{align*}
	\mathbb{E} \left[ \left( cc^* - \Id \right)    q q^* \left( cc^* - \Id \right)  \right] &= \Vert q \Vert^2_{\ell_2} \Id.
	\end{align*}  
\end{lemma}
The following lemma summarizes well-known facts regarding the tail decay of certain quantities which involve circular-symmetric normal random variables. For the sake of completeness we include a proof.
\begin{lemma}\label{lemma:orliczestimates}
	Suppose that $c\in \mathbb{C}^N $ is a random vector 
	with independent entries $c_i\sim\mathcal{CN}(0,1)$.
	%,whose entries $c_i $ have distribution
	%$ c_i \sim \mathcal{N} \left(0, \frac{1}{2} \right) + i \mathcal{N}
	%\left(0, \frac{1}{2} \right)$.
	Let $p,q \in \mathbb{C}^N $ be arbitrary. Then we have the following
	inequalities:
	\begin{align}
	\Big\Vert \Vert c \Vert_{\ell_2} \Big\Vert_{\psi_2} & \lesssim \sqrt{N} \label{ineq:complexgaussian1}\\
	\Big\Vert \vert c^* q \vert \Big\Vert_{\psi_2} & \lesssim  \Vert q \Vert_{\ell_2}  \label{ineq:complexgaussian2}\\
	\Big\Vert \Vert \left( c c^* - \Id \right) q \Vert_{\ell_2} \Big\Vert_{\psi_1} &\lesssim  \sqrt{N} \Vert q \Vert_{\ell_2}  \label{ineq:complexgaussian3}\\
	\Big\Vert p^* \left(  cc^* - \Id \right) q \Big\Vert_{\psi_1} &\lesssim \Vert p \Vert_{\ell_2} \Vert q \Vert_{\ell_2} \label{ineq:complexgaussian4}
	\end{align}
\end{lemma}

\begin{proof}
	In order to prove (\ref{ineq:complexgaussian1}) note that
	\begin{align*}
	%\Big\Vert \Vert c \Vert_{\ell_2} \Big\Vert^2_{\psi_2}\le \Big\Vert  \sum_{i=1}^{n} \vert c_i \vert  \Big\Vert^2_{\ell_2} \le \sum_{i=1}^{n} \Big\Vert \vert c_i \vert \Big\Vert^2_{\psi_2} \sim N.
	\Big\Vert \Vert c \Vert_{\ell_2} \Big\Vert^2_{\psi_2}\lesssim \Big\Vert \Vert c \Vert^2_{\ell_2} \Big\Vert_{\psi_1} \le \sum_{i=1}^{N} \Big\Vert \vert c_i \vert^2 \Big\Vert_{\psi_1} \lesssim N.
	\end{align*}
	The first inequality follows from \cite[Lemma 5.14]{vershyninlecturenotes} and for the second one we used the triangle inequality. In order to prove (\ref{ineq:complexgaussian2}) it is enough to note that $ c^*q \sim \mathcal{CN} \left( 0, \Vert q \Vert^2_{\ell_2}   \right)$. (\ref{ineq:complexgaussian3}) follows from the inequality chain
	\begin{align*}
	\Big\Vert \Vert \left( c c^* - \Id \right) q \Vert_{\ell_2} \Big\Vert_{\psi_1} &\le \Big\Vert \Vert c \Vert_{\ell_2} \vert c^* q \vert + \Vert q \Vert_{\ell_2} \Big\Vert_{\psi_1} \le  \Big\Vert \Vert c \Vert_{\ell_2}  \Big\Vert_{\psi_2} \Big\Vert \vert c^* q \vert  \Big\Vert_{\psi_2} + \Big\Vert  \Vert q \Vert_{\ell_2} \Big\Vert_{\psi_1}\\
	& \lesssim \sqrt{N} \Vert q \Vert_{\ell_2} + \Vert q \Vert_{\ell_2} \lesssim \sqrt{N} \Vert q \Vert_{\ell_2}.
	\end{align*}
	In the second inequality we have used the Hoelder inequality (\ref{ineq:Hoelderinequality}) and the second line follows directly from (\ref{ineq:complexgaussian1}) and (\ref{ineq:complexgaussian2}). In a similar way one proves (\ref{ineq:complexgaussian4}).
\end{proof}
We will also need the following standard fact, which follows from a union bound.
\begin{lemma}\label{conditioninglemma}
	Let $\omega, L \ge 1$ and $ \Gamma $ a finite set. For all $ i \in \lbrack r \rbrack$ let $ m_i \in \mathbb{C}^{N_i}  $ such that $\Vert m_i \Vert_{\ell_2}=1 $. Furthermore, assume that $c_{i,j} \in \mathbb{C}^{N_i}$,  $ i \in \lbrack r \rbrack $, $ j \in \Gamma$, are independent random vectors with i.i.d.~entries distributed according to $ \mathcal{CN} \left(0,1\right) $. Then with probability at least $1 - \mathcal{O} \left( L^{-\omega}\right) $ one has
	\begin{align*}
	\underset{i \in \lbrack r \rbrack, j \in \Gamma}{\max} \Vert c_{i,j} \Vert_{\ell_2} &\lesssim_{\omega} \max \left\{\sqrt{N  \log \left(r \vert \Gamma \vert \right) }; \ \sqrt{N  \log L}    \right\}    \\
	\underset{i \in \lbrack r \rbrack, j \in \Gamma}{\max} \vert c^*_{i,j} m_i \vert &\lesssim_{\omega} \max \left\{  \sqrt{ \log \left(r \vert \Gamma \vert \right)};  \sqrt{\log L} \right\} .
	\end{align*}
\end{lemma}

We conclude this section with a proof of Corollary \ref{gaussianconcentration}.
\begin{proof}[Proof of Corollary \ref{gaussianconcentration}]
	Observe that 
	\begin{equation*}
	\big\Vert Z \big\Vert_{2\rightarrow 2} \le \big\Vert  \sum_{i=1}^{n} \text{Re} \left( \gamma_i \right) X_i  \big\Vert_{2 \rightarrow 2} + \big\Vert  \sum_{i=1}^{n} \text{Im} \left( \gamma_i \right) X_i   \big\Vert_{2 \rightarrow 2} .   
	\end{equation*}
	By Theorem \cite[Theorem 4.1.1]{Tropp2015} we obtain that with probability at least $ 1- \exp \left(-t \right) $
	\begin{equation*}
	\big\Vert  \sum_{i=1}^{n} \text{Re} \left( \gamma_i \right) X_i  \big\Vert_{2 \rightarrow 2} \le \frac{1}{\sqrt{2}} \sigma \sqrt{ t+ \log \left( d_1 + d_2 \right) }
	\end{equation*}
	and with probability at least $  1- \exp \left(-t \right) $
	\begin{equation*}
	\big\Vert  \sum_{i=1}^{n} \text{Im} \left( \gamma_i \right) X_i  \big\Vert_{2 \rightarrow 2} \le \frac{1}{\sqrt{2}} \sigma \sqrt{ t+ \log \left( d_1 + d_2 \right) }.
	\end{equation*}
	Combining these facts yields the result.
	%Then both terms on the right-hand side may be bounded with high probability by Theorem \cite[Theorem 4.1.1]{Tropp2015}, which proves the corollary.
\end{proof}

\section{Proof of Lemma \ref{splittinglemma}}\label{appendixsplittinglemma}
	For $ i \in \lbrack r \rbrack $ let $\mathcal{N}_i$ be an $\frac{\varepsilon}{2} $-cover of $B \left(0,1\right) \subset \mathbb{C}^{K_i} $ with respect to the $\Binorm{\cdot}$-norm. Furthermore, let $\mathcal{O} $ be an $\frac{\varepsilon}{2 \sqrt{K_{\mu} }}$-cover of $ B \left(0,1\right) \subset \mathbb{R}^r $ with respect to the $ \Vert \cdot \Vert_{\ell_2}$-norm. We will show that any $Z =  \left( u_1 m_1^*, \ldots, u_r m_r^* \right) \in B^m $ can be approximated by $ Y = \left( \sigma_1 y_1 m_1^*, \ldots, \sigma_r y_r m^*_r  \right) $, where $ \sigma= \left( \sigma_1, \ldots, \sigma_r \right) \in \mathcal{O} $ and $ y_i \in \mathcal{N}_i $. This proves the claim, as the number of such $Y$'s is bounded by the right-hand side. For that choose $\sigma= \left( \sigma_1, \ldots, \sigma_r \right) \in \mathcal{O}  $ such that 
	\begin{equation}\label{ineq:inline501}
	\sqrt{\sum_{i=1}^{r} \left( \Vert u_i \Vert_{\ell_2} - \sigma_i \right)^2 }  \le \frac{\varepsilon}{2 \sqrt{ K_{\mu} }} 
	\end{equation}
	and  $y_i \in \mathcal{N}_i$ such that
	\begin{equation}\label{ineq:inline500}
	\Binormbig{ \frac{1}{\Vert u_i \Vert_{\ell_2}} u_i - y_i } \le \frac{\varepsilon}{2}.
	\end{equation}
	Then one has for $\hat{Y}= \left( \Vert u_1  \Vert_{\ell_2} y_1 m^*_1, \ldots, \Vert u_r \Vert_{\ell_2} y_r m^*_r \right) $	
	%Set $Y= \left( \sigma_1 y_1, \cdots, \sigma_r y_r  \right)$ and $\hat{Y}= \left( \Vert x_1 \Vert_{\ell_2} y_1, \cdots, \Vert x_r \Vert_{\ell_2} y_r \right)$. We notice first that
	\begin{align*}
	\Big\Vert Z- \hat{Y} \Big\Vert^2_{B} & \le \sum_{i=1}^{r} \Big\Vert u_i m_i^*- \Vert u_i \Vert_{\ell_2} y_i m^*_i \Big\Vert^2_{B_{i}} = \sum_{i=1}^{r} \Big\Vert u_i- \Vert u_i \Vert_{\ell_2} y_i \Big\Vert^2_{B_{i}}\\ 
	%	&\le \sum_{i=1}^{r} \Vert x_i \Vert_{\ell_2}^2 \Big\Vert \frac{1}{\Vert x_i \Vert_{\ell_2}}x_i - y_i \Big\Vert^2_{B_p}\\
	&\le \frac{\varepsilon^2}{4} \sum_{i=1}^{r} \Vert u_i \Vert_{\ell_2}^2 = \frac{\varepsilon^2}{4} \Vert Z \Vert^2_F  \le \frac{\varepsilon^2}{4}.
	\end{align*}
	The first inequality follows from (\ref{Bnorminequality3}) and the next equality follows from 
	\begin{equation*}
	\Vert m_i \left( u_i - \Vert u_i \Vert_{\ell_2} y_i \right)^* b_{i,\ell} \Vert_{\ell_2} = \vert \left( u_i - \Vert u_i \Vert_{\ell_2} y_i \right)^* b_{i, \ell}  \vert 
	\end{equation*}
	which is due to $ \Vert m_i \Vert_{\ell_2}=1$. The subsequent inequality is a consequence of (\ref{ineq:inline500}). The second equality again follows from $ \Vert m_i \Vert_{\ell_2}=1 $ for all $ i \in \lbrack r \rbrack $. Similarly,
	\begin{align*}
	\Vert \hat{Y} - Y  \Vert_{B} & \le \sqrt{ \sum_{i=1}^{r} \Big\Vert \left( \Vert u_i \Vert_{\ell_2} - \sigma_i  \right) y_i m^*_i \Big\Vert^2_{B_{i}}  }   = \sqrt{ \sum_{i=1}^{r} \left( \Vert u_i \Vert_{\ell_2} - \sigma_i \right)^2 \Vert y_i \Vert_{B_{i}}^2 }\\
	& \le   \sqrt{  K_{\mu}  \sum_{i=1}^{r} \left(  \Vert u_i \Vert_{\ell_2} - \sigma_i \right)^2 }
	\le \frac{\varepsilon}{2}.
	\end{align*}
	Here the second inequality follows from 
	\begin{equation*}
	\Binorm{y_i} = \sqrt{L} \ \underset{\ell \in \lbrack L \rbrack}{\max} \ \vert y^*_i b_{i, \ell} \vert \le \sqrt{L} \Vert y_i \Vert_{\ell_2} \underset{\ell \in \lbrack L \rbrack}{\max} \Vert b_{i,\ell} \Vert_{\ell_2} \le \sqrt{  K_{\mu}  }
	\end{equation*}
	and the last inequality is a consequence of (\ref{ineq:inline501}). Combining the two inequalities gives $\Vert Z-Y \Vert_{B} \le \varepsilon $ which finishes the proof. \qed

\end{appendices}
\end{document}